\documentclass{article}
\usepackage{geometry}
\usepackage[english]{babel} 
\usepackage{type1cm}
\usepackage[utf8]{inputenc}
\usepackage[T1]{fontenc}

\usepackage{mathptmx, calc, indentfirst, fancyhdr, graphicx, epstopdf, lastpage, ifthen, float, amsmath, amssymb, setspace, enumitem, mathpazo, booktabs, titlesec, etoolbox, tabto, xcolor, colortbl, soul, multirow, microtype, tikz, totcount, changepage, attrib, upgreek, array, tabularx, pbox, ragged2e, tocloft, marginfix, enotez, amsthm, scrextend, url, newfloat, caption, seqsplit, hyperref, sectsty, amsfonts, subcaption, pgfplots}
\usepackage[right]{lineno}
\hypersetup{
	colorlinks=true,
	linkcolor=blue,
	citecolor=blue,
	filecolor=blue,
	urlcolor=blue
}
\sectionfont{\large}
\subsectionfont{\normalsize}
\subsubsectionfont{\normalfont\itshape}
\usepackage[square,numbers,compress]{natbib}
\usepackage{mathtools}
\usepackage{stmaryrd}
\usepackage{xfrac}
\pgfplotsset{compat=1.13}
\usetikzlibrary{arrows,calc,fillbetween,decorations.markings,positioning,patterns}

\newtheorem{Theorem}{Theorem}[section]
\newtheorem{Corollary}{Corollary}[Theorem]
\newtheorem{Lemma}{Lemma}[section]

\newtheorem*{Remark}{Remark}

\newtheorem{Notation}{Notation}
\newtheorem{Definition}{Definition}
\newtheorem{Assumption}{Assumption}
\newtheorem{WeakPropertyStar}{Property}

\newtheorem{StrictPropertyStar}{Property}

\newtheorem{PropertyM}{Property}

\newtheorem{PropertyU}{Property}

\newtheorem{Example}{Example}

\usepackage[nolist,nohyperlinks]{acronym}
\acrodef{RPL}{Routing Protocol for Low-Power and Lossy Networks}
\acrodef{ID}{indentifier}
\acrodef{PID}{Partial Information Decomposition}
\acrodef{KL}{Kullback-Leiber}

\definecolor{my-red}{HTML}{ac0000}
\definecolor{my-green}{HTML}{008000}
\definecolor{my-plum}{HTML}{92268F}
\definecolor{magenta}{HTML}{E00000}
\newcommand{\extralength}{25mm}
\newcommand{\definedAs}{\ensuremath{\coloneqq}}
\newcommand{\area}[1]{\ensuremath{\textnormal{Area}\left(#1\right)}}

\newcommand{\mycircle}[2]{\tikz[baseline=(char.base)]{\node[#1,shape=circle,draw,inner sep=1pt] (char) {\color{black}#2};}}
\newcommand{\redundancy}{\mycircle{fill=orange!40}{\tiny 3}}
\newcommand{\uniqueA}{\mycircle{fill=blue!40}{\tiny 1}}
\newcommand{\uniqueB}{\mycircle{fill=my-green!40}{\tiny 2}}
\newcommand{\synergy}{\mycircle{fill=magenta!40}{\tiny 4}}

\newcommand{\refEq}[1]{Equation~(\ref{#1})}
\newcommand{\refFig}[1]{Figure~\ref{#1}}
\newcommand{\refSubFig}[2]{{Figure~\ref{#1}#2}}

\title{Quantifying redundancies and synergies\\with measures of inequality}
\author{Tobias Mages~\textsuperscript{1*} and Christian Rohner~\textsuperscript{1}\\
	\small \textsuperscript{1} Department of Information Technology, Uppsala University, Uppsala, Sweden \\
	\small \textsuperscript{*} Correspondence: tobias.mages@it.uu.se}
\date{{\small July 5, 2024}}
\begin{document}
	\maketitle
	\begin{abstract}
		Inequality measures provide a valuable tool for the analysis, comparison, and optimization based on system models. This work studies the relation between attributes or features of an individual to understand how redundant, unique, and synergetic interactions between attributes construct inequality. For this purpose, we define a family of inequality measures (f-inequality) from f-divergences. Special cases of this family are, among others, the Pietra index and the Generalized Entropy index. We present a decomposition for any f-inequality with intuitive set-theoretic behavior that enables studying the dynamics between attributes. Moreover, we use the Atkinson index as an example to demonstrate how the decomposition can be transformed to measures beyond f-inequality. The presented decomposition provides practical insights for system analyses and complements subgroup decompositions. Additionally, the results present an interesting interpretation of Shapley values and demonstrate the close relation between decomposing measures of inequality and information. 
	\end{abstract} \hspace{10pt}
	
	
	\section{Introduction}
	\label{sec:intro}
%
Understanding the structure of how resources are provided or how value is distributed directly leads to the question of decomposing inequality. 
Besides applications in economics and social sciences, the decomposition of inequality can be used to analyze, compare, and optimize systems in engineering. For an example from computer science, consider:
\begin{itemize}
	\item Energy and communication:\vspace{-1mm}
	\begin{itemize}
		\item How is the required energy distributed between nodes in some wireless routing protocol?\vspace{-1mm}
		\item How is the provided network capacity distributed between nodes?
	\end{itemize}
	\item Data and prediction:\vspace{-1mm}
	\begin{itemize}
		\item How is the privacy of different user groups impacted for obtaining data?\vspace{-1mm}
		\item How well does a machine learning model perform for the needs of different user groups?
	\end{itemize}\vspace{-1mm}
\end{itemize}
Such analyses can be split into three components, out of which this work addresses the latter two:
\begin{enumerate}
	\item \textbf{How can we quantify the property of interest?}\newline 
	 The \textit{indicator variable} of an individual shall give a non-negative value for how `good` the system is for the participant. In the examples above, this could be a measure of energy, network capacity, privacy, or prediction performance. The design of indicator variables is challenging since it requires detailed insights from a domain expert, depends on the research question, and influences the resulting notion of inequality. This work assumes that the indicator variable is given to maintain domain independence.
	\item \textbf{How can we quantify inequality?}\newline
	We introduce a family of inequality measures ($f$-inequality) that generalize the Pietra and Generalized Entropy Index (Section \ref{subsec:generalized-ineq-measure}). They are derived from  $f$-divergence/$f$-information and deepen the relation between information theory and inequality measures previously established by \citet{theilBook} and \citet{shorrocks1980}.
	\item \textbf{How can we decompose inequality for gaining insights?}\newline
	We present a novel decomposition for studying the interactions between attributes of individuals. For the initial examples, attributes of an individual could be its device type, network position, interest group or age. The decomposition is inspired by recent work in information theory~\cite{williams-beer,mages2024} and is constructed using the lattice formed by the Atkinson criterion.
	The decomposition is designed to provide a practical operational interpretation and satisfy a set-theoretic intuition as shown in \refFig{fig:intro-proposal}. We demonstrate the decomposition for any $f$-inequality (Section \ref{subsec:decomposition-ineq-measure}) and their transformations, such as the Atkinson index (Section \ref{subsec:transforming-ineq-measure}).
\end{enumerate}

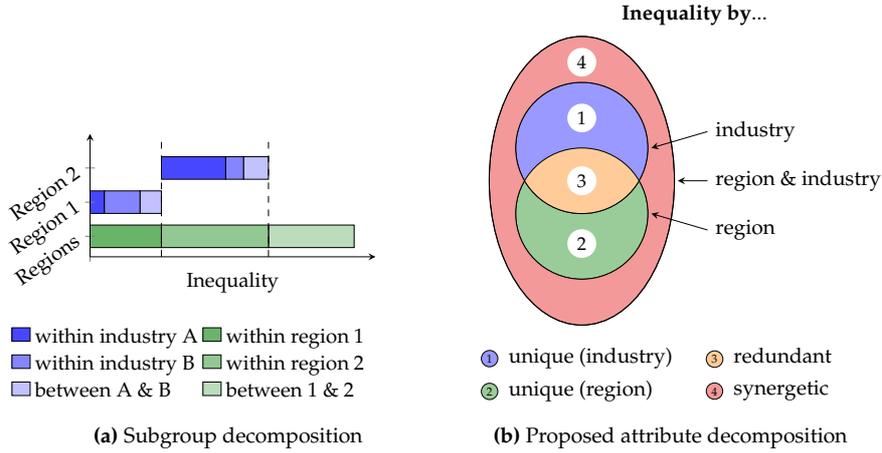
\begin{figure}[t]
	\centering
%
\definecolor{blue1}{HTML}{4747FF}
\definecolor{blue2}{HTML}{8585FF}
\definecolor{blue3}{HTML}{C2C2FF}
\definecolor{green1}{HTML}{6AB46A}
\definecolor{green2}{HTML}{92C892}
\definecolor{green3}{HTML}{BBDDBB}

\begin{subfigure}[t]{0.38\linewidth}
	\scalebox{0.85}{\begin{tikzpicture}
	\begin{axis}[
			set layers=axis on top,
			width=6cm, height=3.5cm, xbar stacked, axis lines=left,
			xlabel={\small Inequality},	ylabel={},
			xlabel style = {yshift=2mm},
			ymin=-3, ymax=6, xmin=0, xmax=8,
			xtick={0}, xticklabels={},
			ytick={-1.5, 1, 3.5}, yticklabels={\small Regions, \small Region 1, \small Region 2},
			yticklabel style = {rotate=35, anchor=east},
			legend cell align={left}, legend style={draw=none,at={(0.35,-0.5)}, anchor=north}, legend columns=2,
		]
		\draw [dashed] (2, -3) -- (2, 5.5);
		\draw [dashed] (5, -3) -- (5, 5.5);
		
		\addplot+[fill=white,draw=none, opacity=100,forget plot] coordinates {(0, 1)  (2, 3.5)  (0,-1.5)};
		\addplot [fill=blue1] coordinates {(0.4, 1)  (1.8, 3.5)  (0,-1.5)};
		\addplot [fill=green1] coordinates {(0, 1)  (0, 3.5) (2, -1.5)};
		
		\addplot [fill=blue2] coordinates {(1, 1)  (0.5, 3.5)  (0,-1.5)};
		\addplot [fill=green2] coordinates {(0, 1)  (0, 3.5) (3, -1.5)};
		
		\addplot [fill=blue3] coordinates {(0.6, 1)  (0.7, 3.5)  (0,-1.5)};
		\addplot [fill=green3] coordinates {(0, 1)  (0, 3.5) (2.4, -1.5)};
		
		\legend{\small within industry A, \small  within region 1, \small  within industry B, \small within region 2, \small between A \& B, \small between 1 \& 2};
	\end{axis}
	\end{tikzpicture}}
	\caption{Subgroup decomposition}
	\label{fig:intro-indicator}
\end{subfigure}%
~
\begin{subfigure}[t]{0.325\linewidth}
	\scalebox{0.87}{\begin{tikzpicture}[color=black]
			\begin{scope}[xshift=-1.75cm]
				\node[anchor=south] (s12) at (1.7,1.75) {\small \textbf{Inequality by}...};
				\node[anchor=west] (s12) at (1.9,-0.50) {\small region \& industry};
				\node[anchor=west] (s1) at (1.9,  0.25) {\small industry};
				\node[anchor=west] (s2) at (1.9,-1.25) {\small region};
				\draw[draw=black,fill=my-green!40] (0,-1) circle (1);
				\draw[draw=black,fill=blue!40] (0,0) circle (1);
				\begin{scope}[draw=black]
					\clip (0,0) circle (1);
					\clip (0,-1) circle (1);
					\path [fill=orange!40] (0,-1) circle (1);
					\end{scope}
				\begin{scope}[draw=black,even odd rule]
					\clip (0,0) circle (1) (0,-0.5) ellipse (1.4 and 2.2);
					\clip (0,-1) circle (1) (0,-0.5) ellipse (1.4 and 2.2);
					\fill[magenta!40] (0,-0.5) ellipse (1.4 and 2.2);
				\end{scope}
				\path[-stealth]
				($(s12.west)$) edge ($(s12.west)+(-0.45,0)$)
				($(s1.west)$) edge ($(s1.west)+(-0.85,-0.25)$)
				($(s2.west)$) edge ($(s2.west)+(-0.85,+0.25)$);
				\draw (0,-0.5) ellipse (1.4 and 2.2) node[above,yshift=45,circle,fill=white,inner sep=2pt] {\footnotesize4};
				\node[circle,fill=white,inner sep=2pt, yshift=13] at (0,0) {\footnotesize 1};
				\node[circle,fill=white,inner sep=2pt] at (0,-0.5) {\footnotesize 3};
				\node[circle,fill=white,inner sep=2pt, yshift=-13] at (0,-1) {\footnotesize 2};
				\draw[draw=black] (0,-1) circle (1);
				\draw[draw=black] (0,0) circle (1);
			\draw[draw=black] (0,-0.5) ellipse (1.4 and 2.2);
			\end{scope}
			\node[anchor=west] at (0.0,-3.7) {\small \synergy~ synergetic};
			\node[anchor=west] at (-3.4,-3.2) {\small \uniqueA~ unique (industry)};
			\node[anchor=west] at (0.0,-3.2) {\small \redundancy~ redundant};
			\node[anchor=west] at (-3.4,-3.7) {\small \uniqueB~ unique (region)};
	\end{tikzpicture}}
	\caption{Proposed attribute decomposition}
	\label{fig:intro-proposal}
\end{subfigure}
	\caption{\textbf{Intuition for the relation between a subgroup decomposition and the proposed attribute decomposition.} Consider a set of companies with industry type and region as attributes: (\textbf{a}) A subgroup decomposition provides detailed insights for the possible values of an attribute, such as region 1 or 2. However, it does not provide insights into the dynamics between attributes. (\textbf{b}) An attribute decomposition provides detailed insights into the interaction between attributes, such as redundant and synergetic effects between industries and regions. However, it does not provide insights for particular attribute values, such as region 1 or 2. Therefore, subgroup and attribute decompositions complement each other.}
	\label{fig:intro-explanation}
\end{figure}

\subsection*{Related work}
The quantification and decomposition of inequality have mainly been driven by economic research~\cite{dalton1920,atkinson1970,theilBook,lerman1985,paul2004,centralbank2019}. An established framework within this area is the \textit{subgroup decomposition}~\cite{bhattacharya1967,bourguignon1979,shorrocks1984,dagum1998}: As indicated by \refEq{eq:intro-subgroup} and \refFig{fig:intro-indicator}, this framework considers a particular partitioning of the population into subgroups. The framework aims to decompose the total inequality into the inequality between subgroups and the inequality within subgroups.
\begin{equation}
	\text{Total Inequality} = \text{(Inequality between subgroups)} + \sum_{\text{subgroup } \in \text{ Partition}} \text{(Inequality within subgroup)}
	\label{eq:intro-subgroup}
\end{equation}
As visualized in \refFig{fig:intro-indicator}, this can provide detailed insights into the attribute values that characterize a subgroup. For example, we can see inequality with respect to the different regions and industries but do not clearly see the interactions between industries and regions.

This work presents a complementing \textit{partition decomposition} or \textit{attribute decomposition}, as visualized in \refFig{fig:intro-proposal}. We decompose inequality into different population partitionings to characterize the dynamics between attributes. The resulting decomposition provides insights into how inequality is constructed from redundant, unique, and synergetic effects between attributes, as indicated by \refEq{eq:intro-attributedecomp}.
\begin{equation}
	\begin{aligned}
		\text{Total Inequality} = & \text{ (redundant between attributes)} + \text{(unique to first attribute)} \\
		&~~+ \text{(unique to second attribute)} + \text{(synergetic between attributes)}
	\end{aligned}
	\label{eq:intro-attributedecomp}
\end{equation}
In summary, a subgroup decomposition studies the interactions between subgroups for a particular population partitioning. An attribute decomposition studies the interactions between possible partitionings of the population based on the attributes of individuals.

	\section{Background, preliminaries and examples}

	\begin{Remark}
		Throughout this work, we assume access to some empirical/estimated/known distribution of the indicator variable. All concepts within this work can be described by probability distributions or a finite set of given samples. Since common inequality measures are typically expressed in terms of the latter, we provide all definitions in the same format. This also enables the discussion of small and intuitive examples. However, all presented definitions can be adjusted for the computation on a given probability distribution rather than a given set of samples. Methods for estimating the relevant distributions are discussed among others in~\cite{estimateLorenz1990,estimateLorenz1993,estimateLorenz1999,estimateLorenz2021}.
	\end{Remark}

	\subsection{Definitions and notation}
	\label{sec:notation}
%
\begin{Notation}\hfill
	\begin{itemize}
		\item We use subscripts to distinguish variable names, such as $s_1,s_2\in\mathbb{R}_{\geq0}$.
		\item We notate the power set as $\mathcal{P}(\cdot)$ and the set of all multisets as $\mathcal{P}_M(\cdot)$.
		\item We notate the Cartesian product of two sets by $\mathbf{A}\times\mathbf{B}$.
		\item We notate the n-ary Cartesian product for a set of sets by $\mathcal{C}(\cdot)$.
		\item We notate the additive union of multisets as $\mathbf{A}\uplus\mathbf{B}$.
		\item We reserve the variable $n\geq1$ for the total number of features/attributes of each individual.
		\item We indicate the set of values for a categorical feature/attribute as $\mathbb{A}_i$ with $i\in\{1,..,n\}$.
		\item We write the function $\tau(i,\cdot)$ to access the $i$-th elements of a tuple starting from zero. \newline For example $\tau(i,(0,..,i))=i$. This notation only appears within Section \ref{sec:notation}.
		\item We indicate unused variables using an underscore, such as $\_\in\mathbf{A}$.
	\end{itemize}
\end{Notation}
\begin{Definition}\hfill
	\begin{itemize}
		\item An \textbf{individual} is a tuple $\rho \in (\mathbb{R}_{\geq 0} \times \mathbb{A}_1 \times \cdots\times\mathbb{A}_n)$. The first element $\tau(0,\rho) \in \mathbb{R}_{\geq 0}$  represents its non-negative indicator variable. The remaining elements ($0<i\leq n$) represent its categorical features/attributes $\tau(i,\rho)\in\mathbb{A}_i$.
		\item We define a \textbf{model} as multiset of individuals $\mathbf{M}\in\mathcal{P}_M(\mathbb{R}_{\geq 0}\times\mathbb{A}_1\times...\times \mathbb{A}_n)$. The distribution of indicator values and attributes may be obtained from empirical data and/or estimations. We reserve the symbol $\mathbf{M}$ throughout this work to indicate a model.
		\item We define a \textbf{population} $\mathbf{S}\in\mathcal{P}_M(\mathbb{R}_{\geq0})$ as multiset of indicator values. Throughout this work, we reserve the symbol $\mathbf{S}$ for multisets of indicator values and note the average indicator value (arithmetic mean) of $\mathbf{S}$ by $\overline{\textbf{S}}\definedAs\sfrac{1}{|\mathbf{S}|}\cdot\sum_{s\in\mathbf{S}}s$.
		\item We define a \textbf{subgroup} by a function $\vartheta(\mathbf{B},\mathbf{M})$ that takes a set of attribute indices and values $(i,a)\in\mathbf{B}$ with a model $\mathbf{M}$ and returns a population by selecting the indicator values of individuals that satisfy all given attributes.
		\begin{equation}
			\vartheta(\mathbf{B},\mathbf{M}) \definedAs \{\tau(0,\rho)\in\mathbf{M}~:~(\forall (i,a)\in\mathbf{B})[\tau(i,\rho)=a]\}
			\label{eq:subgroup}
		\end{equation}
		\item We define a \textbf{partitioning} of a model by a function $\Gamma(\mathbf{a},\mathbf{M})$ that takes a set of attribute indices $i\in\mathbf{a}$ and a model $\mathbf{M}$ and returns a population. Each distinct subgroup from the considered attributes shall be represented by its size and cumulative indicator value. As it can be seen from Section \ref{sec:bg-lorenz}, this is (Lorenz) equivalent to representing each individual $\rho\in\mathbf{M}$ by the average indicator value of its subgroup:
		\begin{equation}
			\Gamma(\mathbf{a},\mathbf{M}) \definedAs
				\biguplus_{\mathbf{B}\in\mathbf{C}}\quad \biguplus_{\_\in\vartheta(\mathbf{B},\mathbf{M})}\left\{\sum_{s\in\vartheta(\mathbf{B},\mathbf{M})} \frac{s}{|\vartheta(\mathbf{B},\mathbf{M})|}\right\} \qquad\text{ where: } \mathbf{C} = \mathcal{C}(\left\{\{i\}\times\mathbb{A}_i~:~i\in\mathbf{a}\right\})
			\label{eq:partition}
		\end{equation}
		\item We notate an inequality measure as function $I : \mathcal{P}_M(\mathbb{R}_{\geq 0}) \rightarrow \mathbb{R}_{\geq 0}$ that assigns a non-negative real value to any population.
	\end{itemize}
	\label{def:terminology}
\end{Definition}
\begin{Assumption}
	Throughout this work, we assume that indicator values are non-negative ($\forall \rho\in\mathbf{M}~:~\tau(0,\rho)\geq 0$) and that at least one individual has a non-zero indicator value ($\exists \rho\in\mathbf{M} : \tau(0,\rho) > 0$).
\end{Assumption}
\begin{Example}
	Consider a \ac{RPL}, where battery-powered devices form a tree for routing packets to a root node. For a comparison with other protocols, we are interested in how evenly the required energy is split between devices.
	\begin{itemize}
		\item Let each device in the network have two attributes (n=2): a device type $\mathbb{A}_1=\{A, B\}$ and rank $\mathbb{A}_2=\mathbb{N}_{\geq 0}$ that indicates the length of its shortest path to the root.
		\item Let the indicator variable be the average power consumption caused by the routing protocol relative to the device's battery size.
		\item For a network of four devices, let the system model be $\mathbf{M}=\{(0.01,A,0),\ (0.05,B,1),\ (0.03,B,1),\ (0.01,B,2)\}$.
		\item The population $\vartheta(\{(1,B)\},\mathbf{M}) = \{0.05,\ 0.03,\ 0.01\}$ is the subgroup of individuals with the first attribute (device type) having value $B$. This subgroup has an average indicator value of $0.3$.
		\item The partition on the first attribute gives the population $\Gamma(\{1\},\mathbf{M}) = \{0.01,\ 0.03,\ 0.03,\ 0.03\}$. The partition on both attributes gives the population $\Gamma(\{1,2\},\mathbf{M}) = \{0.01,\ 0.04,\ 0.04,\ 0.01\}$. The partition on no attribute gives a uniform distribution $\Gamma(\emptyset,\mathbf{M}) = \{0.025,\ 0.025,\ 0.025,\ 0.025\}$, since $\mathbf{C}=\{\emptyset\}$ in \refEq{eq:partition} and $\vartheta(\emptyset,\mathbf{M})$ returns the indicator value of all individuals.
		\item Note that we refer with 'total inequality' to the inequality between distinguishable individuals based on all given attributes $I(\Gamma(\{1,..,n\},\mathbf{M}))$. As it can be seen from the partition $\Gamma(\{1,2\},\mathbf{M})$ above, distinguishing all individuals may require a unique \ac{ID} which can be modeled as additional attribute $\mathbb{A}_3$.
	\end{itemize}
	\label{example:bg-1-notation}
\end{Example}

	\subsection{Measuring inequality}
	\label{sec:bg-inequality-measures}
%
\subsubsection{Inequality metric properties}
An inequality measure should satisfy the following properties:
\begin{PropertyM}[Label invariance~\cite{centralbank2019}]
	Inequality is invariant to the label of groups or individuals.
	\label{prop:ineq-1}
\end{PropertyM}
\begin{PropertyM}[Duplication invariance~\cite{dalton1920}]\hfill\newline
	Inequality is invariant when duplicating each individual in the population (size invariance).
	\begin{equation}
		I(\mathbf{S}) = I(\mathbf{S} \uplus \mathbf{S})
	\end{equation}
	\label{prop:ineq-2}
\end{PropertyM}
\begin{PropertyM}[Scale invariance~\cite{allison1978}]\hfill\newline
	Inequality is invariant under linear scaling of the indicator variable by a factor $k\in\mathbb{R}_{> 0}$ (unit invariance).
	\begin{equation}
		I(\mathbf{S}) = I(\{k\cdot s ~:~ s \in \mathbf{S}\})
	\end{equation}
	\label{prop:ineq-3}
\end{PropertyM}
\begin{PropertyM}[Pigou-Dalton transfer principle~\cite{dalton1920,pigou1912}]\hfill\newline
	Consider a population $\mathbf{S} = \mathbf{G}\uplus\{s_1,s_2\}$ and a population  $\mathbf{S}' = \mathbf{G}\uplus\{s'_1,s'_2\}$, where $s_1\neq s_2$ and $(s'_1,s'_2)$ is a convex combination of $(s_1,s_2)$ with $q\in(0,0.5]$ as shown in \refEq{eq:transfer}. We say $\mathbf{S}'$ represents the population $\mathbf{S}$ after a Pigou-Dalton transfer between $s_1$ and $s_2$. 
	\begin{equation}
		\begin{aligned}
			s'_1 &= (1-q)\cdot s_1 + q\cdot s_2\\
			s'_2 &= q\cdot s_1+ (1-q)\cdot s_2
		\end{aligned}
		\label{eq:transfer}
	\end{equation}
	\begin{itemize}
		\item weak version: A non-zero Pigou-Dalton transfer ($q\in(0,0.5]$) can only reduce inequality $I(S) \geq I(S')$.
		\item strict version: A non-zero Pigou-Dalton transfer ($q\in(0,0.5]$) must reduce inequality $I(S) > I(S')$.
	\end{itemize}
	Note that satisfying Property \ref{prop:ineq-1} directly extends the range of $q\in(0,0.5]$ to $q\in(0,1)$, since $q>0.5$ equals a transfer with relabeling.
	\label{prop:ineq-4}
\end{PropertyM}
\begin{PropertyM}[Non-Negativity with zero at uniform distribution~\cite{bourguignon1979}]\hfill
\begin{itemize}
	\item Inequality is non-negative: $I(\mathbf{S})\geq 0$.
	\item Inequality is zero if all individuals have an identical indicator value.
	\begin{equation}
		\begin{aligned}
			(\forall s_1,s_2\in\mathbf{S}~:~ s_1 = s_2) &\Longrightarrow I(\mathbf{S}) = 0
		\end{aligned}
	\end{equation}
\end{itemize}
\label{prop:ineq-5}
\end{PropertyM}
\begin{Definition}
	\item An inequality measure satisfies the `weak Property \ref{prop:ineq-1}-\ref{prop:ineq-5}` when considering the weak version of Property \ref{prop:ineq-4}.
	\item An inequality measure satisfies the `strict Property \ref{prop:ineq-1}-\ref{prop:ineq-5}` when considering the strict version of Property \ref{prop:ineq-4}.
\end{Definition}

\subsubsection{Measures of inequality}
Several inequality measures are known to satisfy the weak or strict Property \ref{prop:ineq-1}-\ref{prop:ineq-5}. The following measures are commonly used in the literature~\cite{centralbank2019}:
\begin{itemize}
	\item Gini coefficient~\cite{gini1912}:
	\begin{equation}
		G(\mathbf{S}) = \frac{1}{2 \overline{\mathbf{S}} |\mathbf{S}|^2 }\sum_{s_1\in\mathbf{S}}\sum_{s_2\in\mathbf{S}}|s_1-s_2|
	\end{equation}
	\item Pietra index~\cite{pietra1915}, also known as Ricci-, Schutz- or Hoover index:
	\begin{equation}
		R(\mathbf{S}) = \frac{1}{2 |\mathbf{S}|}\sum_{s\in\mathbf{S}}\frac{|s-\overline{\mathbf{S}}|}{\overline{\mathbf{S}}}
		\label{eq:pietra}
	\end{equation}
	\item Generalized Entropy index~\cite{shorrocks1980}: a parameterized family, the special case of $c=1$ is known as Theil index~\cite{theilBook}. The parameter range ($c$) varies with restrictions on the indicator value.
	\begin{equation}
		\begin{aligned}
			&c \in \mathbb{R}\setminus\{0,1\}  &\Longrightarrow~&~ \text{GE}_{c}(\mathbf{S}) = \frac{1}{c(c-1)}\frac{1}{|\mathbf{S}|}\sum_{s\in\mathbf{S}}\left(\left(\frac{s}{\overline{\mathbf{S}}}\right)^c-1\right)\\
			&c = 1 &\Longrightarrow~&~  \text{GE}_{1}(\mathbf{S}) = \frac{1}{|\mathbf{S}|}\sum_{s\in\mathbf{S}} \frac{s}{\overline{\mathbf{S}}}\ln\left(\frac{s}{\overline{\mathbf{S}}}\right)\\
			&c = 0 &\Longrightarrow~&~ \text{GE}_{0}(\mathbf{S}) = -\frac{1}{|\mathbf{S}|}\sum_{s\in\mathbf{S}} \ln\left(\frac{s}{\overline{\mathbf{S}}}\right)\\
		\end{aligned}
		\label{eq:gec-1}
	\end{equation}
	\item Atkinson indexes~\cite{atkinson1970}: a parameterized family, designed with the properties of a welfare function in mind. It can be represented as transformation of the Generalized entropy index~\cite{inequalityBook}. 
	\begin{equation}
		\begin{aligned}
			&d \in (0,1) &\Longrightarrow~&~ A_{d}(\mathbf{S}) = 1-\left[d(d-1)\text{GE}_{1-d}(\mathbf{S}) + 1\right]^\frac{1}{1-d} &&= 	1- \frac{1}{\overline{\mathbf{S}}}\left(\frac{1}{|\mathbf{S}|} \sum_{s\in\mathbf{S}} s^{1-d}\right)^{\frac{1}{1-d}}
			\\
			&d = 1 &\Longrightarrow~&~  A_{1}(\mathbf{S}) = 1-e^{-\text{GE}_{0}(\mathbf{S})}&&= 1- \frac{1}{\overline{\mathbf{S}}}\left(\prod_{s\in\mathbf{S}}s\right)^{\frac{1}{|\mathbf{S}|}}
		\end{aligned}
		\label{eq:atkinson-index}
	\end{equation}
\end{itemize}

	\subsection{Lorenz curves and their ordering}
	\label{sec:bg-lorenz}
%
The Lorenz curve represents the minimal concentration of wealth in a subgroup of a particular size and is typically defined through the Quantile function~\cite{inequalityBook,centralbank2019}. However, we can equivalently define the Lorenz curve as the boundary of a zonogon~\cite{Mosler1996,Mosler2007}. This directly highlights the well-known relation~\cite{theilBook,atkinson1970} between the Lorenz curve~\cite{lorenz1905}, the Neyman-Pearson region from hypothesis testing~\cite[p. 278]{polyanskiy2022information}, and the (pointwise) Blackwell order from information theory~\cite{blackwell-no-lattice,mages2024}. Moreover, the definition through zonogons highlights additional properties of the inequality measures defined in Section \ref{sec:method}.

\subsubsection{Zonogons and their partial order}

\begin{Definition}[Stochastic matrix]
	A (row) stochastic matrix $\lambda\in\mathbb{R}_{\geq0}^{a\times b}$ of dimension $a \times b$ is a matrix, where all entries are non-negative real values and each row sums to one. In a double stochastic matrix, all entries are non-negative, and each row and column sums to one.
\end{Definition}
\begin{Definition}[Normalized population matrix]
	We define a function $\kappa: \mathcal{P}_M(\mathbb{R}_{\geq0}) \rightarrow \mathbb{R}_{\geq0}^{2\times m}$ as shown in \refEq{eq:population-to-matrix}. The function maps a population $\mathbf{S}$  to a $2\times |\mathbf{S}|$ row stochastic matrix by normalizing both, the population size and indicator value. The ordering of columns can be arbitrary (discussed below Definition \ref{def:zonogon}).
	\begin{equation}
		\kappa(\mathbf{S}) \definedAs \frac{1}{|\mathbf{S}|} \begin{bmatrix}~
			\begin{matrix}
				1\\ \sfrac{s}{\overline{\textbf{S}}}
			\end{matrix} &:~ s \in \mathbf{S}
		\end{bmatrix}
		\label{eq:population-to-matrix}
	\end{equation}
	\label{def:population-to-matrix}
\end{Definition}
\begin{Example}
A normalized population matrix for $\Gamma(\{1,2\},\mathbf{M})$ from Example \ref{example:bg-1-notation} is shown in \refEq{eq:matrix-example} (columns can be permuted).
\begin{equation}
	\kappa(\Gamma(\{1,2\},\mathbf{M})) = \begin{bmatrix}
		\sfrac{1}{4} & \sfrac{1}{4} & \sfrac{1}{4} & \sfrac{1}{4} \\
		\sfrac{1}{10} & \sfrac{4}{10} & \sfrac{4}{10} & \sfrac{1}{10} \end{bmatrix}
	\label{eq:matrix-example}
\end{equation}
\end{Example}
\begin{Notation}
	We access a vector within a normalized population matrix as $\vec{v}_i\in\kappa(\mathbf{S})$, such as $\vec{v}_i=\frac{1}{|\mathbf{S}|}\left(\begin{smallmatrix} 1 \\ \sfrac{s_i}{\overline{\textbf{S}}}\end{smallmatrix}\right)$.
\end{Notation}

\begin{Definition}[Zonogon~\cite{Mosler1996,Mosler2007,blackwell-no-lattice}]
	The function $Z: \mathbb{R}_{\geq0}^{2\times n} \rightarrow \mathcal{P}([0,1]^2)$ transforms a normalized population matrix into a zonogon. A zonogon (\refEq{eq:defZonogon}) is a set of two-dimensional points constructed from the Minkowski sum of line segments from its generating vectors $\vec{v}_i \in \kappa(\mathbf{S})$. 
	\begin{equation}
		Z(\kappa(\mathbf{S})) \definedAs \left\{\sum_{i=1}^{|\mathbf{S}|} x_i \vec{v}_i ~:~ x_i\in[0,1],~ \vec{v}_i \in \kappa(\mathbf{S})\right\} = \left\{\kappa(\mathbf{S}) a ~:~ a \in [0,1]^{|\mathbf{S}|}\right\}
		\label{eq:defZonogon}
	\end{equation}
	\label{def:zonogon}
\end{Definition}
The zonogon can be defined equivalently as image of the unit-cube $[0,1]^{|\mathbf{S}|}$ under the linear transformation of the given matrix and provides the following basic properties~\cite{blackwell-no-lattice}:
\begin{itemize}
	\item The zonogon of a stochastic matrix is a centrally symmetric convex polygon.
	\item The zonogon is invariant to permuting the order of matrix vectors:\newline$Z([\begin{smallmatrix}
		\vec{v}_1 ~&~ \dots ~&~ \vec{v}_2 ~&~ \dots
	\end{smallmatrix}]) = Z([\begin{smallmatrix}\vec{v}_2 ~&~ \dots ~&~ \vec{v}_1 ~&~ \dots\end{smallmatrix}])$.
	\item The zonogon is invariant to splitting/merging matrix vectors of identical slope:\newline $Z([\begin{smallmatrix}(1+\ell)\vec{v}_1 ~&~ \dots\end{smallmatrix}]) = Z([\begin{smallmatrix}\vec{v}_1 ~&~ \ell \vec{v}_1 ~&~ \dots\end{smallmatrix}])$.
	\item Ordering the matrix vectors $\vec{v}_i\in\kappa(\mathbf{S})$ by increasing/decreasing slope provides the zonogon perimeter (visualized in \refSubFig{fig:example-zonogon}{a}).
\end{itemize}
\begin{Notation}
	For abbreviation, we use the notation $Z_\kappa(\mathbf{S})\definedAs Z(\kappa(\mathbf{S}))$.
\end{Notation}
Zonogon examples and their interpretation are discussed in Example \ref{example:zonogon-lorenz} of Section \ref{subsec:lorenz-atkinson}.
\begin{Definition}[Zonogon order~\cite{blackwell-no-lattice}]
	The subset relation (\refEq{eq:zonogon-order-1}) is a partial order of zonogons from $2\times \_$ row stochastic matrices that forms a (non-distributive) lattice with unique meet and join elements. Under this ordering relation, the meet of two zonogons corresponds to their intersection and their join corresponds to the convex hull of their union.
	\begin{subequations}
		\begin{align}
			Z_\kappa(\mathbf{S}_1) ~&\subseteq Z_\kappa(\mathbf{S}_2)\quad \label{eq:zonogon-order-1}\\
			\Longleftrightarrow \qquad\kappa(\mathbf{S}_1) ~&= \kappa(\mathbf{S}_2)\ \lambda \quad \text{for some row stochastic matrix $\lambda$}\label{eq:zonogon-order-2}
			\end{align}
			\label{eq:zonogon-order}
	\end{subequations}
	\label{def:zonogon-order}
\end{Definition}

A zonogon is a subset of another $Z_\kappa(\mathbf{S}_1) \subseteq Z_\kappa(\mathbf{S}_2)$ if and only if there exists a row stochastic matrix $\lambda$ such that $\kappa(\mathbf{S}_1) = \kappa(\mathbf{S}_2) \lambda$ (\refEq{eq:zonogon-order-2})~\cite{blackwell-no-lattice}. 
This relation leads to \refEq{eq:product-subset}, which is useful since any sequence of Pigou-Dalton transfers corresponds to a multiplication by some stochastic matrix (see Appendix \ref{apsub:prop-representation}).
\begin{equation}
	Z(\kappa(\mathbf{S})\ \lambda) \subseteq Z_\kappa(\mathbf{S})
	\label{eq:product-subset}
\end{equation}
We can use the lattice of zonogons to define a lattice of population equivalence classes.
\begin{Definition}[Population equivalence]
	We say two populations $(\mathbf{S}_1,\mathbf{S}_2)$ are equivalent ($\cong$) if and only if they generate the same zonogon.
	\begin{equation}
		(\mathbf{S}_1\cong \mathbf{S}_2) \definedAs (Z_\kappa(\mathbf{S}_1) = Z_\kappa(\mathbf{S}_2))
	\end{equation}
	\label{def:pop-equivalence}
\end{Definition}
\begin{Notation}\hfill
	\begin{itemize}
	\item We notate the equivalence class of a population as $\langle\mathbf{S}_1\rangle \definedAs \{\mathbf{S}_2\in\mathcal{P}(\mathcal{P}_M(\mathbb{R}_{\geq 0}))~:~\mathbf{S}_2 \cong \mathbf{S}_1\}$.
	\item We extend the notation for zonogons to equivalence classes $Z_\kappa(\langle\mathbf{S}\rangle) \definedAs Z_\kappa(\mathbf{S})$.
	\end{itemize}
	\label{not:pop-classes}
\end{Notation}
\begin{Definition}[Lattice of population equivalence classes]
	The lattice of zonogons provides a lattice for the equivalence classes of populations. We notate their ordering as $\langle\mathbf{S}_1\rangle \sqsubseteq \langle\mathbf{S}_2\rangle$, their meet as $\langle\mathbf{S}_1\rangle \sqcap \langle\mathbf{S}_2\rangle$ and join as $\langle\mathbf{S}_1\rangle \sqcup \langle\mathbf{S}_2\rangle$. We notate a top and bottom population for the lattice as $\top_\textbf{S}=\{0,1\}$ and $\bot_\textbf{S}=\{1\}$ respectively. $\textnormal{Conv}(\cdot)$ indicates the convex hull in \refEq{eq:pop-join}.
	\begin{subequations}
		\begin{align}
			(\langle\mathbf{S}_1\rangle \sqsubseteq \langle\mathbf{S}_2\rangle) &\definedAs (Z_\kappa(\mathbf{S}_1) \subseteq Z_\kappa(\mathbf{S}_2))\\
			(\langle\mathbf{S}_1\rangle \sqsubset \langle\mathbf{S}_2\rangle) &\definedAs (Z_\kappa(\mathbf{S}_1) \subset Z_\kappa(\mathbf{S}_2))\\
			Z_\kappa(\langle\mathbf{S}_1\rangle \sqcap \langle\mathbf{S}_2\rangle) &=
			Z_\kappa(\mathbf{S}_1) \cap Z_\kappa(\mathbf{S}_2)\\
			Z_\kappa(\langle\mathbf{S}_1\rangle \sqcup \langle\mathbf{S}_2\rangle) &=
			\textnormal{Conv}\left(Z_\kappa(\mathbf{S}_1) \cup Z_\kappa(\mathbf{S}_2)\right)\label{eq:pop-join}
		\end{align}
		\label{eq:pop-class-lattice}
	\end{subequations}
	\label{def:pop-class-lattice}
\end{Definition}

\begin{Notation}
	The equivalence class of the `joint` distribution for two attributes is $\langle\Gamma(\{1,2\},\mathbf{M})\rangle$, while the `join` of both attributes is $\langle\Gamma(\{1\},\mathbf{M})\rangle \sqcup \langle\Gamma(\{2\},\mathbf{M})\rangle$.
\end{Notation}

To obtain a set-theoretic behavior of inequality measures, we have to understand the inclusion-exclusion relation between the defined lattice operations. For an example of this concept, we can first use the standard set-theoretic inclusion-exclusion relation ($|A \cup B| = |A| + |B| - |A \cap B|$) to obtain \refEq{eq:basic-in-ex}: For a non-empty set of populations ($\emptyset \neq \mathbf{A}$), computing an inclusion-exclusion principle on the zonogon area of the meet (zonogon intersection) gives the area of their union, which is a lower bound on the area of their join (convex hull of the union). We can separate terms based on their sign (\refEq{eq:basic-in-ex-2}) to recognize another inclusion-exclusion principle below. 
\begin{subequations}
		\begin{align}
			\area{Z_\kappa\left(\bigsqcup_{\mathbf{S}\in\mathbf{A}} \langle\mathbf{S}\rangle\right)} \geq \area{\bigcup_{\mathbf{S}\in\mathbf{A}} Z_\kappa(\mathbf{S})} = \sum_{\emptyset \neq \mathbf{B} \subseteq \mathbf{A}}(-1)^{|\mathbf{B}|-1} \area{Z_\kappa\left(\bigsqcap_{\mathbf{S}\in\mathbf{B}}\langle\mathbf{S}\rangle\right)}\label{eq:basic-in-ex}\\
			\area{Z_\kappa\left(\bigsqcup_{\mathbf{S}\in\mathbf{A}} \langle\mathbf{S}\rangle\right)}+ 
			\sum_{\stackrel{\emptyset\neq\mathbf{B}\subseteq\mathbf{A}}{|\mathbf{B}| \text{ even}}} \area{Z_\kappa\left(\bigsqcap_{\mathbf{S}\in\mathbf{B}}\langle\mathbf{S}\rangle\right)}
			\geq\sum_{\stackrel{\mathbf{B}\subseteq\mathbf{A}}{|\mathbf{B}| \text{ odd}}} \area{Z_\kappa\left(\bigsqcap_{\mathbf{S}\in\mathbf{B}}\langle\mathbf{S}\rangle\right)}\label{eq:basic-in-ex-2}
		\end{align}
		\label{eq:basic-in}
\end{subequations}

Instead of measuring the area, we will define a class of inequality measures in Section \ref{subsec:generalized-ineq-measure} that is additive with the zonogon sum:
\begin{Definition}[Zonogon sum]
	The addition of two zonogons corresponds to their Minkowski sum:
	\begin{subequations}
		\begin{align}
			Z_\kappa(\langle\mathbf{S}_1\rangle) + Z_\kappa(\langle\mathbf{S}_2\rangle) \definedAs Z_\kappa(\mathbf{S}_1) + Z_\kappa(\mathbf{S}_2) ~&\definedAs \left\{a+b ~:~ a\in Z_\kappa(\mathbf{S}_1),~ b\in Z_\kappa(\mathbf{S}_2)\right\} \label{eq:def-zono-sum-1}\\
			~&\ = Z\left(\begin{bmatrix}\kappa(\mathbf{S}_1)&\kappa(\mathbf{S}_2)\end{bmatrix}\right)
			\label{eq:def-zono-sum-2}
		\end{align}
		\label{eq:def-zono-sum}
	\end{subequations}
	\label{def:zono-sum}
\end{Definition}
The defined operators provide the following inclusion-exclusion relation at the zonogon sum~\cite[Lemma A5]{mages2024}.
\begin{equation}
	Z_\kappa\left(\bigsqcap_{\mathbf{S} \in \mathbf{A}} \langle\mathbf{S}\rangle\right) + \sum_{\stackrel{\emptyset\neq\mathbf{B}\subseteq\mathbf{A}}{|\mathbf{B}| \text{ even}}} Z_\kappa\left(\bigsqcup_{\mathbf{S}\in\mathbf{B}}\langle\mathbf{S}\rangle\right) \subseteq \sum_{\stackrel{\mathbf{B}\subseteq\mathbf{A}}{|\mathbf{B}| \text{ odd}}} Z_\kappa\left(\bigsqcup_{\mathbf{S}\in\mathbf{B}}\langle\mathbf{S}\rangle\right)
	\label{eq:zonogon-relation}
\end{equation}

\subsubsection{Operational meaning of zonogons}
\label{subsec:lorenz-atkinson}
\begin{Definition}[Lorenz Curve~\cite{lorenz1905}]
	The Lorenz curve maps a fraction of the population (x-axis) to the minimal fraction of the indicator value  (y-axis) concentrated in any subgroup of this size. The Lorenz curve is the lower boundary of the zonogon (Definition \ref{def:zonogon}, visualized in \refSubFig{fig:example-zonogon}{a})~\cite{Mosler1996,Mosler2007}.
	\label{def:lorenz-cuve}
\end{Definition}

\begin{Definition}[Atkinson criterion]
	Assume two populations ($\mathbf{S}_1,\mathbf{S}_2$) with identical indicator mean ($\overline{\textbf{S}}_1 =\overline{\textbf{S}}_2$) and let the welfare of a population be the expected value of an increasing concave function $w(\cdot)$. Some populations can be compared without agreeing on the specific function $w(\cdot)$, which leads to the Atkinson criterion shown in \refEq{eq:at-crit-1}~\cite{atkinson1970}.
	\begin{subequations}
		\begin{align}
		& \forall w~:~\frac{1}{|\mathbf{S}_2|}\sum_{s\in\mathbf{S}_2} w(s)\leq \frac{1}{|\mathbf{S}_1|} \sum_{s\in\mathbf{S}_1} w(s) &&\text{where $w(\cdot)$ is increasing and concave}\label{eq:at-crit-1}\\
		\Longleftrightarrow  \quad&\forall v~:~\frac{1}{|\mathbf{S}_1|}\sum_{s\in\mathbf{S}_1}v(s) \leq \frac{1}{|\mathbf{S}_2|}\sum_{s\in\mathbf{S}_2}v(s)  &&\text{where $v(t)\definedAs -w(t)$ is decreasing and convex}\label{eq:at-crit-4}\\
		\Longleftrightarrow \quad& \exists \lambda~:~\kappa(\mathbf{S}_1) = \kappa(\mathbf{S}_2)\ \lambda &&\text{where $\lambda$ is a row stochastic matrix}\label{eq:at-crit-2}\\
		\Longleftrightarrow \quad& Z_\kappa(\mathbf{S}_1) \subseteq Z_\kappa(\mathbf{S}_2)\label{eq:at-crit-3}\\
		\Longleftrightarrow \quad& \quad\ \langle\mathbf{S}_1\rangle \sqsubseteq \langle\mathbf{S}_2\rangle
		\end{align}
		\label{eq:atkinson-crit}
	\end{subequations}
	\label{def:atkinson-crit}
\end{Definition}
For the context of this work, we can change the perspective from higher welfare to lower inequality (\refEq{eq:at-crit-4}), where the convexity of $v(\cdot)$ leads to Jensen's inequality. \citet{atkinson1970} showed that a population has a higher welfare (in this context: lower inequality) for any $w$ (in this context: $v$) if and only if there exists a sequence of Pigou-Dalton transfers from $\mathbf{S}_2$ to $\mathbf{S}_1$ (\refEq{eq:at-crit-2}). This equals the condition of non-intersecting Lorenz curves and the zonogon order (\refEq{eq:at-crit-3})~\cite{atkinson1970}.
\begin{Remark}
	The condition of an identical indicator mean in Definition \ref{def:atkinson-crit} has no further importance if the inequality measure is invariant to the population size and scaling of the indicator variable (Property \ref{prop:ineq-1}-\ref{prop:ineq-3}). In this case, normalizing the population size and indicator variable always results in the same mean without affecting the inequality measure.
\end{Remark}

\begin{Example}
	Consider the model $\textbf{M}$ obtained from Table \ref{tbl:example-zonogonmodel} with the two attributes $\mathbb{A}_1 = \{A,B\}$ and $\mathbb{A}_2 = \{C,D\}$.
	\begin{table}[h!]\centering
		\begin{tabular}{r | c c | l}\toprule
			Indicator value & $\mathbb{A}_1$ &  $\mathbb{A}_2$ & Number of individuals\\\hline
			\sfrac{1}{6} & A & D & 6\\
			\sfrac{2}{3} & A & C & 3\\
			7 & B & C & 1 \\\midrule
			Total sum:\quad 10 & / & / & 10\\\bottomrule
		\end{tabular}
		\caption{Example population model}
		\label{tbl:example-zonogonmodel}
	\end{table}
	
	We can construct the following three partitions based on the given attributes. To abbreviate the notation, we can sum columns with an identical slope without affecting the underlying zonogon, as discussed above. The order of columns is arbitrary.
	\begin{subequations}
		\begin{align}
			Z_\kappa(\Gamma(\{1,2\},\mathbf{M}))  &= 
			Z\left(\begin{bmatrix}
				0.1 & 0.1 & 0.1 & 0.1 & 0.1 &0.1 & 0.1 & 0.1 & 0.1 & 0.1 \\
				\sfrac{1}{60}& \sfrac{1}{60} & \sfrac{1}{60} & \sfrac{1}{60} & \sfrac{1}{60} & \sfrac{1}{60} & \sfrac{2}{30} & \sfrac{2}{30} & \sfrac{2}{30} & \sfrac{7}{10} \\ 
			\end{bmatrix}\right)\\ &= Z\left( \begin{bmatrix}
				0.6 & 0.3 & 0.1 \\ 0.1 &  0.2 & 0.7
			\end{bmatrix}\right)\\
			Z_\kappa(\Gamma(\{1\},\mathbf{M})) &= 
				Z\left(\begin{bmatrix}
					0.1 & 0.1 & 0.1 & 0.1 & 0.1 &0.1 & 0.1 & 0.1 & 0.1 & 0.1 \\
					\sfrac{7}{10}& \sfrac{1}{30} & \sfrac{1}{30} & \sfrac{1}{30} & \sfrac{1}{30} & \sfrac{1}{30} & \sfrac{1}{30} & \sfrac{1}{30} & \sfrac{1}{30} & \sfrac{1}{30} \\ 
				\end{bmatrix}\right)\\
				& = Z\left(\begin{bmatrix}
					0.1 & 0.9 \\ 0.7 & 0.3
				\end{bmatrix}\right)\\
			Z_\kappa(\Gamma(\{2\},\mathbf{M}))  &= 
			Z\left(\begin{bmatrix}
				0.1 & 0.1 & 0.1 & 0.1 & 0.1 &0.1 & 0.1 & 0.1 & 0.1 & 0.1 \\
				\sfrac{1}{60}& \sfrac{1}{60} & \sfrac{1}{60} & \sfrac{1}{60} & \sfrac{1}{60} & \sfrac{1}{60} & \sfrac{9}{40} & \sfrac{9}{40} & \sfrac{9}{40} & \sfrac{9}{40} \\ 
			\end{bmatrix}\right)\\
			&= Z\left(\begin{bmatrix}
				0.6 & 0.4 \\ 0.1 & 0.9
			\end{bmatrix}\right)
		\end{align}
	\end{subequations}
	The zonogon for each partition is shown in \refFig{fig:example-zonogon}. The lower boundary of the zonogon is the Lorenz curve (\refSubFig{fig:example-zonogon}{a}). Each edge segment of the zonogon corresponds to one subgroup of the partition, and its slope represents the expected normalized indicator value of the individuals within it (\refSubFig{fig:example-zonogon}{b}). As shown in \refSubFig{fig:example-zonogon}{b}, the partitions $\Gamma(\{1\},\mathbf{M})$ and $\Gamma(\{2\},\mathbf{M})$ are incomparable since neither zonogon is a subset of the other. This means that there does not exist a sequence of Pigou-Dalton transfers to convert one population into the other and that the preferable partition depends on the considered inequality measure. However, we can always construct Pigou-Dalton transfers to eliminate an attribute, which leads to the relation of \refEq{eq:partition-subset}.
	\begin{equation}
		\mathbf{a} \subseteq \mathbf{b} \Longrightarrow Z_\kappa(\Gamma(\mathbf{a},\mathbf{M})) \subseteq Z_\kappa(\Gamma(\mathbf{b},\mathbf{M}))
		\label{eq:partition-subset}
	\end{equation}
	The numbers of this particular example (Table \ref{tbl:example-zonogonmodel}) were chosen such that the joint attribute distribution corresponds to the join of partitioning on the individual attributes: $\langle\Gamma(\{1,2\},\mathbf{M})\rangle = \langle\Gamma(\{1\},\mathbf{M})\rangle \sqcup \langle\Gamma(\{2\},\mathbf{M})\rangle$. The join plays an important role since it represents the dependence between attributes ($\mathbb{A}_1$, $\mathbb{A}_2$) that leads to a zonogon that is unique and a subset of any other dependence. Thus, we can construct Pigou-Dalton transfers from all other attribute dependencies to arrive at the join population. The attribute dependence of the join provides minimal inequality under any measure satisfying Property \ref{prop:ineq-1}-\ref{prop:ineq-5}. Therefore, it represents a notion of ideal attribute dependence and demonstrates that the ideal dependence between attributes is measure independent.
	\begin{figure}[h]
		\centering	
		\begin{subfigure}[t]{0.48\textwidth}\centering
%
\begin{tikzpicture}[color=black]
	\draw[draw=none,fill=gray!20,fill opacity=0.7] (0,0) -- (4,4) -- (4,0) -- (0,0);
	\draw[mark=none,black] (0,0) -- (4,4);
	\draw[gray,fill=magenta!40,fill opacity=0.5] (0,0) -- (4*0.1,4*0.7) -- (4*0.4,4*0.9) -- (4,4) -- (4-4*0.1,4-4*0.7) -- (4-4*0.4,4-4*0.9) -- (0,0);
	\begin{axis}[height=5.58cm, width=5.58cm,xmin=0, xmax=1,ymin=0, ymax=1,  
		ylabel={\small cumulative indicator}, xlabel={\small proportion with lowest indicator},minor x tick num=1,minor y tick num=1,clip=false,
		legend style={at={(1.1,0.5)},anchor=west,draw=none,legend plot pos=left,legend cell align=left}]
		\draw[red,line width=.5mm] (1,1) -- (1-0.1,1-0.7) -- (1-0.4,1-0.9) -- (0,0);
		\path[draw=none] (1-0.1,1-0.7) -- node[midway,sloped,below] {\color{red}\footnotesize Lorenz curve} (1-0.4,1-0.9);
		\addplot[only marks,mark=*,mark size=1] coordinates {(0,0) (0.1,0.7) (0.4,0.9) (1,1) (1-0.1,1-0.7) (1-0.4,1-0.9) (0,0)};
		\path[-latex, draw=black] (0,0) -- node[pos=0.55,anchor=south] {\small$\vec{v}_{1}$} (0.6,0.1);
		\path[-latex, draw=black] (0.6,0.1) -- node[pos=0.5,anchor=south east] {\small$\vec{v}_{2}$} (0.6+0.3,0.1+0.2);
		\path[-latex, draw=black] (0.6+0.3,0.1+0.2) -- node[pos=0.55,anchor=east] {\small$\vec{v}_{3}$} (0.6+0.3+0.1,0.1+0.2+0.7);
		\addplot[only marks,mark=*,mark size=1] coordinates {(0.6+0.3+0.1,0.1+0.2+0.7)};
		\path[latex-, draw=black] (1-0,1-0) -- node[pos=0.55,anchor=north] {\small$\vec{v}_{1}$} (1-0.6,1-0.1);
		\path[latex-, draw=black] (1-0.6,1-0.1) -- node[pos=0.5,anchor=north west] {\small$\vec{v}_{2}$} (1-0.6-0.3,1-0.1-0.2);
		\path[latex-, draw=black] (1-0.6-0.3,1-0.1-0.2) -- node[pos=0.55,anchor=west]{\small$\vec{v}_{3}$} (1-0.6-0.3-0.1,1-0.1-0.2-0.7);
	\end{axis};
	\foreach \col/\label [count=\i] in {{magenta/{$Z_\kappa(\Gamma(\{1,2\},\mathbf{M})) = Z\left( \begin{bmatrix}\vec{v}_1&\vec{v}_2&\vec{v}_3\end{bmatrix}\right)$}}}{
		\node[anchor=west] at ($(-0.35-0.2,-1.85) - \i*(0,-0.4)$) {\footnotesize \label\vphantom{q}};%
		\draw[fill=\col!30] ($(-0.8,-1.8) - \i*(0,-0.4) + (0.2,0.1)$) rectangle ($(-0.8,-1.8) - \i*(0,-0.4) - (0.2,0.1)$);}
	\node at (3.2,-2.2) {\footnotesize$= Z\left( \begin{bmatrix}
			0.6 & 0.3 & 0.1 \\ 0.1 &  0.2 & 0.7
		\end{bmatrix}\right)$};
	\node at (0,-2.5) {};
\end{tikzpicture}
			\caption{Zonogon boundary}
		\end{subfigure}%
		~ 
		\begin{subfigure}[t]{0.48\textwidth}\centering
%
\begin{tikzpicture}[color=black]
	\begin{scope}[xshift=-0.5mm,yshift=0.5mm]
	\path[clip] (-1.5,-1.5) -- (5,5) -- (-1.5,5) -- cycle;
	\draw[draw=none,fill=gray!20,fill opacity=0.7] (0,0) -- (4,4) -- (4,0) -- (0,0);
	\draw[fill=my-green!40,fill opacity=0.5, opacity=0.5] (0,0) --  (4*0.4,4*0.9) -- (4,4) -- (0,0);
	\draw[fill=blue!40,fill opacity=0.5, opacity=0.5] (0,0) -- (4*0.1,4*0.7) -- (4,4) -- (0,0);
	\draw[fill=blue!40,fill opacity=0.5,opacity=0.5] (0,0) -- (4,4) -- (4-4*0.1,4-4*0.7) -- (0,0);
	\draw[fill=my-green!40,fill opacity=0.5,opacity=0.5] (0,0) -- (4,4) -- (4-4*0.4,4-4*0.9) -- (0,0);
	\begin{axis}[height=5.58cm, width=5.58cm,xmin=0, xmax=1,ymin=0, ymax=1,  
		ylabel={\small cumulative indicator}, xlabel={\small proportion with lowest indicator},minor x tick num=1,minor y tick num=1,clip=false,
		legend style={at={(1.1,0.5)},anchor=west,draw=none,legend plot pos=left,legend cell align=left}]
		\addplot[only marks,mark=*,mark size=1] coordinates {(0,0) (0.1,0.7)(0.4,0.9) (1,1) (1-0.1,1-0.7) (1-0.4,1-0.9) (0,0)};
		\path[draw=black,thick] (0,0)-- node[pos=0.6,sloped,below] {\footnotesize $\vartheta({(1,B)},M)$} (0.1,0.7);
		\path[draw=black,thick] (0.1,0.7) -- node[pos=0.4,sloped,below] {\footnotesize $\vartheta({(1,A)},M)$} (1,1);
		\path[draw=black,thick] (0,0) -- node[pos=0.6,sloped,above] {\footnotesize $\vartheta({(2,D)},M)$} (1-0.4,1-0.9);
		\path[draw=black,thick] (1-0.4,1-0.9) -- node[pos=0.4,sloped,above] {\footnotesize $\vartheta({(2,C)},M)$} (1,1);
	\end{axis};
	\end{scope}
	
	\begin{scope}[xshift=0.5mm,yshift=-0.5mm]
		\path[clip] (-1.5,-1.5) -- (5,5) -- (5,-1.5) -- cycle;
		\draw[draw=none,fill=gray!20,fill opacity=0.7] (0,0) -- (4,4) -- (4,0) -- (0,0);
		\draw[fill=my-green!40,fill opacity=0.5, opacity=0.5] (0,0) --  (4*0.4,4*0.9) -- (4,4) -- (0,0);
		\draw[fill=blue!40,fill opacity=0.5, opacity=0.5] (0,0) -- (4*0.1,4*0.7) -- (4,4) -- (0,0);
		\draw[fill=blue!40,fill opacity=0.5,opacity=0.5] (0,0) -- (4,4) -- (4-4*0.1,4-4*0.7) -- (0,0);
		\draw[fill=my-green!40,fill opacity=0.5,opacity=0.5] (0,0) -- (4,4) -- (4-4*0.4,4-4*0.9) -- (0,0);
		\begin{axis}[height=5.58cm, width=5.58cm,xmin=0, xmax=1,ymin=0, ymax=1,  
			ylabel={\small cumulative indicator}, xlabel={\small proportion with lowest indicator},minor x tick num=1,minor y tick num=1,clip=false,
			legend style={at={(1.1,0.5)},anchor=west,draw=none,legend plot pos=left,legend cell align=left}]
			\addplot[only marks,mark=*,mark size=1] coordinates {(0,0) (0.1,0.7)(0.4,0.9) (1,1) (1-0.1,1-0.7) (1-0.4,1-0.9) (0,0)};
			\path[draw=black,thick] (0,0)-- node[pos=0.6,sloped,below] {\footnotesize $\vartheta({(1,B)},M)$} (0.1,0.7);
			\path[draw=black,thick] (0.1,0.7) -- node[pos=0.4,sloped,below] {\footnotesize $\vartheta({(1,A)},M)$} (1,1);
			\path[draw=black,thick] (0,0) -- node[pos=0.6,sloped,above] {\footnotesize $\vartheta({(2,D)},M)$} (1-0.4,1-0.9);
			\path[draw=black,thick] (1-0.4,1-0.9) -- node[pos=0.4,sloped,above] {\footnotesize $\vartheta({(2,C)},M)$} (1,1);
		\end{axis};
	\end{scope}	
	
	\foreach \col/\label [count=\i] in {{my-green/{$Z_\kappa(\Gamma(\{2\},\mathbf{M}))$}},{blue/{$Z_\kappa(\Gamma(\{1\},\mathbf{M}))$}}}{
		\node[anchor=west] at ($(1.25,-2.25) - \i*(0,-0.4)$) {\footnotesize \label\vphantom{q}};%
		\draw[fill=\col!30] ($(1.0,-2.2) - \i*(0,-0.4) + (0.2,0.1)$) rectangle ($(1.0,-2.2) - \i*(0,-0.4) - (0.2,0.1)$);}
		
	\node at (0,-2.5) {};
\end{tikzpicture}
			\caption{Zonogon relation and partition subgroups}
		\end{subfigure}
		\caption{\textbf{Zonogon construction, the meaning of its boundary and their ordering.} (\textbf{a}) The zonogon of a population is a symmetric convex polygon containing the line from (0,0) to (1,1). Sorting the vectors of a normalized population matrix by increasing slope provides the lower boundary of the zonogon, which is the Lorenz curve. (\textbf{b}) Each zonogon boundary segment corresponds to one subgroup of the partition, and its slope is the expected normalized indicator value of its individuals. The subgroups for the partition on attribute $\mathbb{A}_1$ are labeled in the upper triangle and those for the partition on $\mathbb{A}_2$ are labeled in the lower triangle. The example was constructed such that the join of both attributes (\refSubFig{fig:example-zonogon}{b}) equals their joint distribution (\refSubFig{fig:example-zonogon}{a}). For any other attribute dependence, the zonogon of their joint distribution is a superset of \refSubFig{fig:example-zonogon}{a}.}
		\label{fig:example-zonogon}
	\end{figure}
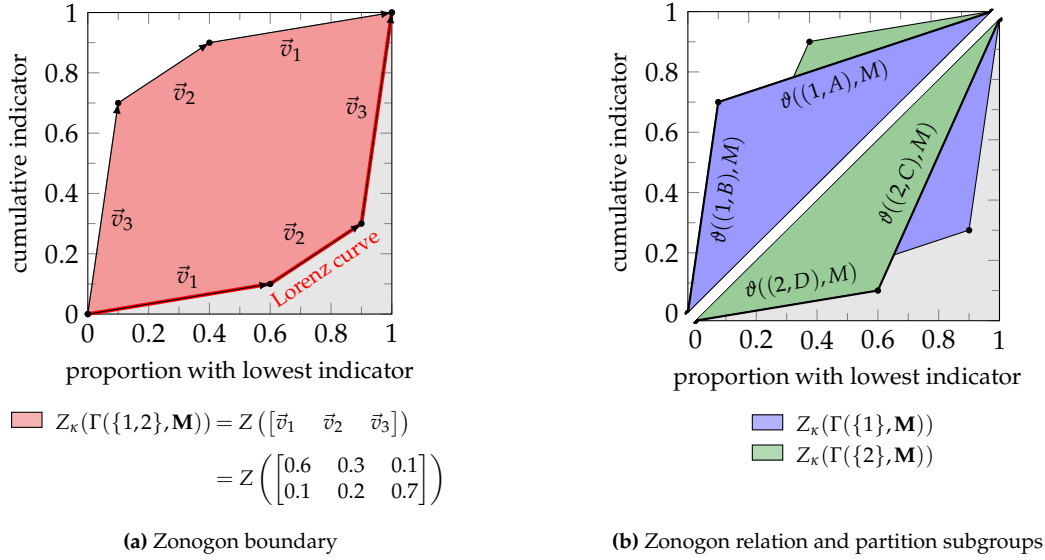
	\label{example:zonogon-lorenz}
\end{Example}

\subsubsection{From ordering to quantification}
\label{subsubsec:order-prop-relation}
We can simplify the required properties of inequality measures for the remaining context of this work by using the ordering of population equivalence classes:
\begin{WeakPropertyStar}
	The inequality measure $I(\cdot)$ shall maintain the zonogon order and quantify a bottom population ($\bot_{\mathbf{S}}$) to zero (\refEq{eq:mainain-weak-zonogon-order}). 
	\begin{subequations}
		\begin{align}
			\langle\mathbf{S}_1\rangle \sqsubseteq \langle\mathbf{S}_2\rangle &\Longrightarrow I(\mathbf{S}_1) \leq I(\mathbf{S}_2),\label{eq:mainain-weak-zonogon-order-1}\\
			I(\bot_{\mathbf{S}}) &= 0.\label{eq:mainain-weak-zonogon-order-2}
		\end{align}
		\label{eq:mainain-weak-zonogon-order}
	\end{subequations}
	\label{prop:weak-propertyStar}
\end{WeakPropertyStar}
\begin{StrictPropertyStar}
	The inequality measure $I(\cdot)$ shall maintain the strict zonogon order and quantify a bottom population ($\bot_{\mathbf{S}}$) to zero (\refEq{eq:mainain-strict-zonogon-order}).
	\begin{subequations}
		\begin{align}
			\langle\mathbf{S}_1\rangle = \langle\mathbf{S}_2\rangle &\Longrightarrow I(\mathbf{S}_1) = I(\mathbf{S}_2),\label{eq:mainain-strict-zonogon-order-1}\\
			\langle\mathbf{S}_1\rangle \sqsubset \langle\mathbf{S}_2\rangle &\Longrightarrow I(\mathbf{S}_1) < I(\mathbf{S}_2),\label{eq:mainain-strict-zonogon-order-2}\\
			I(\bot_{\mathbf{S}}) &= 0.\label{eq:mainain-strict-zonogon-order-3}
		\end{align}
		\label{eq:mainain-strict-zonogon-order}
	\end{subequations}
	\label{prop:strict-propertyStar}
\end{StrictPropertyStar}
\begin{Lemma}
	Satisfying Property \ref{prop:weak-propertyStar} implies that the inequality measure satisfies the weak Property \ref{prop:ineq-1}-\ref{prop:ineq-5}.
	\label{thm:weak-star-implies-1-5}
\end{Lemma}
\begin{Lemma}
	Satisfying Property \ref{prop:strict-propertyStar} implies that the inequality measure satisfies the strict Property \ref{prop:ineq-1}-\ref{prop:ineq-5}.
	\label{thm:strict-star-implies-1-5}
\end{Lemma}
The proof of Lemma \ref{thm:weak-star-implies-1-5} and \ref{thm:strict-star-implies-1-5} is shown in Appendix \ref{subap:proofs-order-prop}.
The relation between some inequality measures and the Lorenz curve (and thus their zonogons) is well established: 
The Gini coefficient is known to equal twice the area between the Lorenz curve and diagonal~\cite[p. 121]{theilBook}. Thus, the Gini coefficient equals the zonogon area ($G(\mathbf{S}) = \area{Z_\kappa(\mathbf{S})}$) and satisfies Property \ref{prop:strict-propertyStar}. The Pietra index is known to equal the maximal vertical distance between the Lorenz curve and diagonal~\cite[p. 17]{centralbank2019} and thus satisfies Property \ref{prop:weak-propertyStar}.

	\section{Methodology}
	\label{sec:method}

	We begin by defining a family of inequality measures (Section \ref{subsec:generalized-ineq-measure}) that are additive under the zonogon sum and demonstrate that several established measures are its special case. Section \ref{subsec:intuitionandshapley} provides an intuition for the concepts of redundancy and synergy and highlights the limitation of analyses with Shapley values in this setting. With this motivation, we explain the decomposition lattice and desired properties for a set-theoretic intuition (Section \ref{subsec:decompositoin-lattice}). Section \ref{subsec:decomposition-ineq-measure} defines a decomposition that satisfies the desired properties and provides a suitable operational interpretation. We demonstrate how the decomposition results can be transformed to other inequality measures, such as an Atkinson index (Section \ref{subsec:transforming-ineq-measure}). Finally, Section \ref{subsec:layered-inequality} discusses multi-layered inequality and Section \ref{sec:relation} highlights the relation between decomposing measures of information and inequality.

	\subsection{Defining f-inequality}
	\label{subsec:generalized-ineq-measure}
%
If a zonogon is a subset of another, then it shall obtain a smaller inequality score to obtain Property \ref{prop:ineq-1}-\ref{prop:ineq-5} from Lemma \ref{thm:weak-star-implies-1-5} and \ref{thm:strict-star-implies-1-5}. For a first intuition, consider quantifying the length of the zonogon boundary (Lorenz curve): All zonogons are convex and have a common start and end point. Therefore, if a zonogon is a subset of another (Atkinson criterion), then its boundary is shorter. 

For turning this conceptual idea into a family of inequality measures, we can follow a simple strategy: (1) Define the inequality measure as sum of quantifying each vector in the normalized population matrix (zonogon boundary segment) by a function $r$. This could be re-phrased to a sum of quantifying each individual of the population. (2) The function $r$ shall satisfy three properties: (a) quantify any vector of slope one to zero, (b) scale linearly, and (c) be convex. Quantifying any vector of slope one to a score of zero ensures quantifying the bottom element ($\bot_{\mathbf{S}}$) correctly. The linear scaling and convexity provide a triangle inequality on the zonogon boundary, which then reflects their subset relation on the inequality measure. Interestingly, we previously studied a function that satisfies exactly these properties for decomposing information measures~\cite{mages2024}.
\begin{Notation}
	We reserve the name $f$ for generator functions of an $f$-divergence~\cite{csiszar1967}: Let $f :(0,\infty) \rightarrow \mathbb{R}$ be a function that satisfies the following three properties. By convention we understand that $f(0) = \lim_{t\rightarrow 0^+}f(t)$ and $0 f\left(\tfrac{0}{0}\right) = 0$:
	\begin{itemize}
		\item $f$ is convex,
		\item $f(1) = 0$,
		\item $f(t)$ is finite for all $t > 0$.
	\end{itemize}
	\label{not:f-div}
\end{Notation}

\begin{Definition}[$f$-inequality]\hfill
	\begin{itemize}
		\item Define a function $r_{f,p}$ as shown in \refEq{eq:f-pseudo-distance} to quantify a vector $\vec{v}=\left[\begin{smallmatrix}x\\y\end{smallmatrix}\right]$ of the zonogon boundary with $p,x,y\in[0,1]$.
		\item Define a parameterized class of $f$-inequality measures ($p\in[0,1]$) as shown in \refEq{eq:zonogon-perimeter} to be the sum of all segments from the Lorenz curve for a populations $\mathbf{S}$.
	\end{itemize}
	\begin{subequations}
		\begin{align}
			r_{f,p}\left(\left[\begin{smallmatrix}x\\y\end{smallmatrix}\right]\right) &\definedAs \left(px+(1-p)y\right) \cdot f\left(\frac{x}{px+(1-p)y}\right)	\label{eq:f-pseudo-distance}\\
			I_{f,p}\left(\mathbf{S}\right) &\definedAs \sum_{\vec{v}\in\kappa(\mathbf{S})} r_{f,p}(\vec{v}) \label{eq:zonogon-perimeter}\\
			&= \frac{1}{|\mathbf{S}|}\sum_{s\in\mathbf{S}}\frac{f\left(g(p,\overline{\mathbf{S}},s)\right)}{g(p,\overline{\mathbf{S}},s)}\qquad \text{where: ~}~ g(p,\overline{\mathbf{S}},s)=\frac{\overline{\mathbf{S}}}{p\overline{\mathbf{S}}+(1-p)s}\label{eq:new-measure}
		\end{align}
	\end{subequations}
	\label{def:generalized-ineq}
\end{Definition}
\begin{Notation}
	\item We say an $f$-inequality measure is `strict` if and only if its generator function $f$ is strictly convex.
	\item We say an $f$-inequality measure is `weak` if and only if its generator function $f$ is not strictly convex.
\end{Notation}

\begin{Theorem}[Properties of $r_{f,p}$ and $I_{f,p}$]
	For a constant $p\in[0,1]$:
	\begin{enumerate}
		\item the function $r_{f,p}(\vec{v})$:
		\begin{enumerate}
			\item \makebox[7.2cm][l]{quantifies any vector of slope one to zero:} $r_{f,p}\left(\left[\begin{smallmatrix}\ell\\\ell\end{smallmatrix}\right]\right)=0$
			\item \makebox[7.2cm][l]{quantifies the zero vector to zero:} $r_{f,p}\left(\left[\begin{smallmatrix}0\\0\end{smallmatrix}\right]\right)=0$
			\item \makebox[7.2cm][l]{scales linearly in $\vec{v}$ where $\ell\in\mathbb{R}$:} $r_{f,p}(\ell\vec{v})=\ell r_{f,p}(\vec{v})$
			\item {is convex in $\vec{v}$:}\begin{itemize}
				\item \makebox[6.5cm][l]{$f$-inequality $\ell\in\{0,1\}$:} $r_{f,p}(\ell\vec{v}_1+(1-\ell)\vec{v}_2) = \ell r_{f,p}(\vec{v}_1)+(1-\ell)r_{f,p}(\vec{v}_2)$
				\item\makebox[6.5cm][l]{weak $f$-inequality $\ell\in(0,1)$:} $r_{f,p}(\ell\vec{v}_1+(1-\ell)\vec{v}_2) \leq \ell r_{f,p}(\vec{v}_1)+(1-\ell)r_{f,p}(\vec{v}_2)$
				\item \makebox[6.5cm][l]{strict $f$-inequality $\ell\in(0,1)$:} $r_{f,p}(\ell\vec{v}_1+(1-\ell)\vec{v}_2) < \ell r_{f,p}(\vec{v}_1)+(1-\ell)r_{f,p}(\vec{v}_2)$
			\end{itemize} 
			\item {satisfies a triangle inequality in $\vec{v}$:}\begin{itemize}
				\item \makebox[6.5cm][l]{$f$-inequality $\textnormal{Slope}(\vec{v}_1) = \textnormal{Slope}(\vec{v}_2)$:} $r_{f,p}(\vec{v}_1+\vec{v}_2) = r_{f,p}(\vec{v}_1)+r_{f,p}(\vec{v}_2)$
				\item\makebox[6.5cm][l]{weak $f$-inequality $\textnormal{Slope}(\vec{v}_1) \neq \textnormal{Slope}(\vec{v}_2)$:} $r_{f,p}(\vec{v}_1+\vec{v}_2) \leq r_{f,p}(\vec{v}_1)+r_{f,p}(\vec{v}_2)$
				\item \makebox[6.5cm][l]{strict $f$-inequality $\textnormal{Slope}(\vec{v}_1) \neq \textnormal{Slope}(\vec{v}_2)$:} $r_{f,p}(\vec{v}_1+\vec{v}_2) < r_{f,p}(\vec{v}_1)+r_{f,p}(\vec{v}_2)$
			\end{itemize}  
		\end{enumerate}
		\item the function $I_{f,p}(\mathbf{S})$:
		\begin{enumerate}
			\item \makebox[6cm][l]{quantifies the bottom element to zero:} $I_{f,p}(\bot_\mathbf{S})=0$
			\item \makebox[6cm][l]{maintains the zonogon order:} 
			\begin{itemize}
			\item \makebox[5.2cm][l]{$f$-inequality:} $\langle\mathbf{S}_1\rangle = \langle\mathbf{S}_2\rangle \Longrightarrow I_{f,p}(\mathbf{S}_1) = I_{f,p}(\mathbf{S}_2)$
			\item\makebox[5.2cm][l]{weak $f$-inequality:} $\langle\mathbf{S}_1\rangle \sqsubseteq \langle\mathbf{S}_2\rangle \Longrightarrow I_{f,p}(\mathbf{S}_1) \leq I_{f,p}(\mathbf{S}_2)$
			\item \makebox[5.2cm][l]{strict $f$-inequality:} $\langle\mathbf{S}_1\rangle \sqsubset \langle\mathbf{S}_2\rangle \Longrightarrow I_{f,p}(\mathbf{S}_1) < I_{f,p}(\mathbf{S}_2)$
			\end{itemize} 
		\end{enumerate}
	\end{enumerate} 
	\label{thm:rf-props}
\end{Theorem}
The proof of Theorem \ref{thm:rf-props} is shown in Appendix \ref{subap:f-ineq-properties}.
\begin{Corollary}\hfil
	\begin{itemize}
	\item Any weak $f$-inequality satisfies Property \ref{prop:weak-propertyStar} and the weak Property \ref{prop:ineq-1}-\ref{prop:ineq-5}.
	\item Any strict $f$-inequality satisfies Property \ref{prop:strict-propertyStar} and the strict Property \ref{prop:ineq-1}-\ref{prop:ineq-5}.
	\end{itemize}
\end{Corollary}
\begin{proof}Follows directly from Theorem \ref{thm:rf-props} with Lemma \ref{thm:weak-star-implies-1-5} and Lemma \ref{thm:strict-star-implies-1-5}.
\end{proof}
\begin{Notation}
	Since $f$-inequality is constant for all populations within an equivalence class (Theorem \ref{thm:rf-props} nr. 2b), we can quantify an equivalence class by any population that it contains: $I_{f,p}(\langle\mathbf{S}\rangle)\definedAs I_{f,p}(\mathbf{S})$.
\end{Notation}

The intended attribute decomposition will require an interpretation for the addition of inequality from multiple populations. Therefore, it will be helpful that the Minkowski sum of the underlying zonogons directly corresponds to the addition of $f$-inequality from their generating populations.
\begin{Lemma}
	Consider two non-empty sets of populations with equal cardinality ($|\mathbf{A}| = |\mathbf{B}|$), then:
	\begin{subequations}
		\begin{align}
			\text{$f$-inequality:}&&\sum_{\mathbf{S}\in \mathbf{A}}Z_\kappa(\mathbf{S}) = \sum_{\mathbf{S}\in \mathbf{B}}Z_\kappa(\mathbf{S}) &\Longrightarrow \sum_{\mathbf{S}\in \mathbf{A}}I_{f,p}(\mathbf{S}) = \sum_{\mathbf{S}\in \mathbf{B}}I_{f,p}(\mathbf{S})\\
			\text{weak $f$-inequality:}&&\sum_{\mathbf{S}\in \mathbf{A}}Z_\kappa(\mathbf{S}) \subseteq \sum_{\mathbf{S}\in \mathbf{B}}Z_\kappa(\mathbf{S}) &\Longrightarrow \sum_{\mathbf{S}\in \mathbf{A}}I_{f,p}(\mathbf{S}) \leq \sum_{\mathbf{S}\in \mathbf{B}}I_{f,p}(\mathbf{S})\\
			\text{strict $f$-inequality:}&&\sum_{\mathbf{S}\in \mathbf{A}}Z_\kappa(\mathbf{S}) \subset \sum_{\mathbf{S}\in \mathbf{B}}Z_\kappa(\mathbf{S}) &\Longrightarrow \sum_{\mathbf{S}\in \mathbf{A}}I_{f,p}(\mathbf{S}) < \sum_{\mathbf{S}\in \mathbf{B}}I_{f,p}(\mathbf{S})
		\end{align}
		\label{eq:minkowski-sum-subsets}
	\end{subequations}
	\label{lem:minkowski-sum-subsets}
\end{Lemma}
The proof of Lemma \ref{lem:minkowski-sum-subsets} is shown in Appendix \ref{subap:f-ineq-additivity}.
\begin{Corollary}
	Any $f$-inequality satisfies the following inclusion-exclusion relation:
	\begin{equation}
		\begin{aligned}
			I_{f,p}\left(\bigsqcap_{\mathbf{S} \in \mathbf{A}} \langle\mathbf{S}\rangle\right) &\leq \sum_{\emptyset\neq\mathbf{B}\subseteq\mathbf{A}}(-1)^{|\mathbf{B}|-1} I_{f,p}\left(\bigsqcup_{\mathbf{S}\in\mathbf{B}}\langle\mathbf{S}\rangle\right)
		\end{aligned}
	\end{equation}
	\label{col:f-ineq-inc-enc}
\end{Corollary}
\begin{proof}
	Follows directly from Lemma \ref{lem:minkowski-sum-subsets} and \refEq{eq:zonogon-relation}.
\end{proof}

\begin{Theorem}
	The Pietra index and Generalized Entropy index are special cases of $f$-inequality:
	\begin{subequations}
		\begin{align}
			R(\mathbf{S}) &= I_{f,p}(\mathbf{S}) &\text{where:~} & p=0~\text{ and }~ f(t) = \frac{|t-1|}{2}\\
			\text{GE}_{c}(\mathbf{S}) &= I_{f,p}(\mathbf{S}) &\text{where:~} & p=0~\text{ and }~f(t) = \frac{t^{1-c} - t}{c(c - 1)}\label{eq:ge-func}\\
			\text{GE}_1(\mathbf{S}) &= I_{f,p}(\mathbf{S}) &\text{where:~} & p=0~\text{ and }~f(t) = -\ln\left(t\right)\\
			\text{GE}_0(\mathbf{S}) &= I_{f,p}(\mathbf{S}) &\text{where:~} & p=0~\text{ and }~f(t) =  t\ln\left(t\right)
		\end{align}
	\end{subequations}
	\label{thm:special-cases}
\end{Theorem}
The proof of Theorem \ref{thm:special-cases} is shown separately in Appendix \ref{subap:f-ineq-cases}. This section presented the construction of inequality measures from any $f$-divergence.

	\subsection{Preliminary intuition for redundancy and synergy}
	\label{subsec:intuitionandshapley}
%
\subsubsection{Intuition examples}
\label{subsec:intuition}
Before constructing the desired decomposition, this section shall give an intuition for the concepts of redundant, unique, and synergetic contributions in the context of inequality measures. For this, we adopt typical examples~\cite{finn-ppid-2} that are fully determined by Property \ref{prop:ineq-1}-\ref{prop:ineq-5}:
\begin{Example}[Redundant contributions]
	\refFig{fig:redundancy-example} provides a fully redundant model $\mathbf{M}$. Attribute $\mathbb{A}_2$ is a re-labeled copy of attribute $\mathbb{A}_1$, and re-labeling groups shall not affect inequality (Property \ref{prop:ineq-1}). Therefore, the partitionings $\Gamma(\{1\},\mathbf{M})$, $\Gamma(\{2\},\mathbf{M})$, and $\Gamma(\{1,2\},\mathbf{M})$ must obtain the same inequality index as visualized by the Venn diagram in \refFig{fig:redundancy-example}. Since all regions fully intersect, we say the inequality is contributed redundantly by attributes $\mathbb{A}_1$ and $\mathbb{A}_2$.
	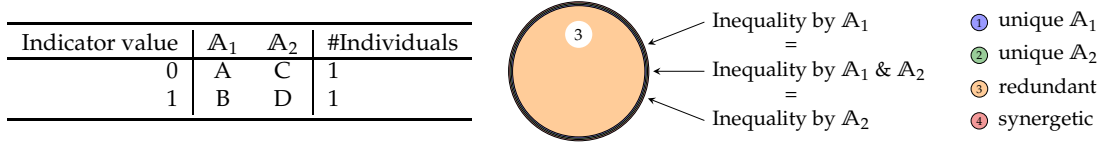
\begin{figure}[h]
		\centering
%
\scalebox{0.87}{\begin{tikzpicture}[color=black]\centering
		\begin{scope}[local bounding box=scope1]
			\node at (0,0) {
				\begin{tabular}{r | c c | l}\toprule
					Indicator value & $\mathbb{A}_1$ &  $\mathbb{A}_2$ & {\#}Individuals\\\hline
					0 & A & C & 1\\
					1 & B & D & 1\\\bottomrule
			\end{tabular}};
		\end{scope}
		\begin{scope}[shift={($(scope1.east)+(1.5cm,0)$)}]
			\node[anchor=west] (s1) at (1.9,0.75) {\small Inequality by $\mathbb{A}_1$};
			\node[anchor=west] (s2) at (1.9,-0.75) {\small Inequality by $\mathbb{A}_2$};
			\node[anchor=west] (s12) at (1.9,0) {\small Inequality by $\mathbb{A}_1$ \& $\mathbb{A}_2$};
			\node[anchor=north,yshift=1mm] at (s1.south) {\small =};
			\node[anchor=south,yshift=-1mm] at (s2.north) {\small =};
			\path[-stealth,shorten >=32pt]
			($(s12.west)$) edge (0,0)
			($(s1.west)$) edge (0,0)
			($(s2.west)$) edge (0,0);
			\path [line width=0.05mm,draw=black,fill=magenta!40] (0,0) circle (1.06);
			\path [line width=0.05mm,draw=black,fill=my-green!40] (0,0) circle (1.04);
			\path [line width=0.05mm,draw=black,fill=blue!40] (0,0) circle (1.02);
			\path [line width=0.05mm,draw=black,fill=orange!40] (0,0) circle (1);
			\draw (0,0.35) node[above,circle,fill=white,inner sep=2pt] {\scriptsize 3};
			\node[anchor=west] at (5.8,0.75) {\small \uniqueA~ unique $\mathbb{A}_1$};
			\node[anchor=west] at (5.8,0.25) {\small \uniqueB~ unique $\mathbb{A}_2$};
			\node[anchor=west] at (5.8,-0.25) {\small \redundancy~ redundant};
			\node[anchor=west] at (5.8,-0.75) {\small \synergy~ synergetic};
		\end{scope}
\end{tikzpicture}}
		\caption{\textbf{Redundancy example.} Fully redundant contribution by both attributes.}
		\label{fig:redundancy-example}
	\end{figure}
\end{Example}
\begin{Example}[Unique contributions]
	\refFig{fig:unique-example} provides a fully unique model $\mathbf{M}$ to attribute $\mathbb{A}_1$. The partitioning $\Gamma(\{2\},\mathbf{M})$ provides a uniform distribution and, thus, an inequality index of zero (Property \ref{prop:ineq-5}). The partitioning on attribute $\mathbb{A}_1$ provides the same population as partitioning on both attributes: $\Gamma(\{1,2\},\mathbf{M}) = \Gamma(\{1\},\mathbf{M})$. Therefore, both partitionings must obtain the same inequality index. This results in the Venn diagram in \refFig{fig:redundancy-example} and we conclude that inequality is contributed uniquely by attribute $\mathbb{A}_1$.
	\begin{figure}[h]
		\centering
%
\scalebox{0.87}{\begin{tikzpicture}[color=black]\centering
		\begin{scope}[local bounding box=scope1]
			\node at (0,0) {
				\begin{tabular}{r | c c | l}\toprule
					Indicator value & $\mathbb{A}_1$ &  $\mathbb{A}_2$ & {\#}Individuals\\\hline
					0 & A & C & 1 \\
					0 & A & D & 1 \\
					1/2 & B & C & 1 \\
					1/2 & B & D  & 1 \\ \bottomrule
			\end{tabular}};
		\end{scope}
		\begin{scope}[shift={($(scope1.east)+(1.5cm,0)$)}]
			\node[anchor=west] (s1) at (1.9,0.75) {\small Inequality by $\mathbb{A}_1$};
			\node[anchor=west] (s2) at (1.9,-0.75) {\small Inequality by $\mathbb{A}_2$ is zero};
			\node[anchor=west] (s12) at (1.9,0) {\small Inequality by $\mathbb{A}_1$ \& $\mathbb{A}_2$};
			\node[anchor=north,yshift=1mm] at (s1.south) {\small =};
			\path [line width=0.05mm,draw=black,fill=magenta!40] (0,0) circle (1.02);
			\path [line width=0.05mm,draw=black,fill=blue!40] (0,0) circle (1.00);
			\path [line width=0.05mm,draw=black,fill=my-green!40] (0,0) circle (0.04);
			\path [line width=0.05mm,draw=black,fill=orange!40] (0,0) circle (0.02);
			\draw (0,0.35) node[above,circle,fill=white,inner sep=2pt] {\scriptsize 1};
			\path[-stealth,shorten >=30pt]
			($(s12.west)$) edge (0,0)
			($(s1.west)$) edge (0,0);
			\path[-stealth,shorten >=2pt]($(s2.west)$) edge (0,0);
			\node[anchor=west] at (5.8,0.75) {\small \uniqueA~ unique $\mathbb{A}_1$};
			\node[anchor=west] at (5.8,0.25) {\small \uniqueB~ unique $\mathbb{A}_2$};
			\node[anchor=west] at (5.8,-0.25) {\small \redundancy~ redundant};
			\node[anchor=west] at (5.8,-0.75) {\small \synergy~ synergetic};
		\end{scope}
\end{tikzpicture}}
		\caption{\textbf{Unique example.} Fully unique contribution by attribute $\mathbb{A}_1$.}
		\label{fig:unique-example}
	\end{figure}
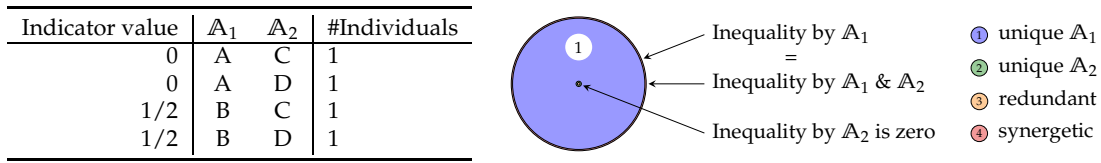
\end{Example}
\begin{Example}[Synergetic contributions]
 \refFig{fig:synergy-example} provides a fully synergetic model $\mathbf{M}$. Partitioning on either attribute individually ($\Gamma(\{1\},\mathbf{M})$ and $\Gamma(\{2\},\mathbf{M})$) provides a uniform distribution and thus an inequality index of zero (Property \ref{prop:ineq-5}). Non-zero inequality can only be measured when partitioning on both attributes ($\Gamma(\{1,2\},\mathbf{M})$), which results in the Venn diagram in \refFig{fig:synergy-example}. Therefore, we say the inequality is contributed synergetically by both attributes. 
	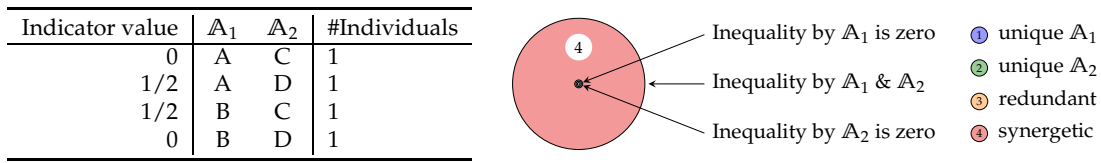
\begin{figure}[h]
		\centering
%
\scalebox{0.87}{\begin{tikzpicture}[color=black]\centering
		\begin{scope}[local bounding box=scope1]
			\node at (0,0) {
				\begin{tabular}{r | c c | l}\toprule
					Indicator value & $\mathbb{A}_1$ &  $\mathbb{A}_2$ & {\#}Individuals\\\hline
					0 & A & C & 1\\
					1/2 & A & D & 1\\
					1/2 & B & C & 1\\
					0 & B & D & 1\\\bottomrule
			\end{tabular}};
		\end{scope}
		\begin{scope}[shift={($(scope1.east)+(1.5cm,0)$)}]
			\node[anchor=west] (s1) at (1.9,0.75) {\small Inequality by $\mathbb{A}_1$ is zero};
			\node[anchor=west] (s2) at (1.9,-0.75) {\small Inequality by $\mathbb{A}_2$ is zero};
			\node[anchor=west] (s12) at (1.9,0) {\small Inequality by $\mathbb{A}_1$ \& $\mathbb{A}_2$};
			\path [line width=0.05mm,draw=black,fill=magenta!40] (0,0) circle (1.0);
			\path [line width=0.05mm,draw=black,fill=blue!40] (0,0) circle (0.06);
			\path [line width=0.05mm,draw=black,fill=my-green!40] (0,0) circle (0.04);
			\path [line width=0.05mm,draw=black,fill=orange!40] (0,0) circle (0.02);
			\draw (0,0.35) node[above,circle,fill=white,inner sep=2pt] {\scriptsize 4};
			\path[-stealth,shorten >=30pt]
			($(s12.west)$) edge (0,0);
			\path[-stealth,shorten >=2pt]
			($(s2.west)$) edge (0,0)
			($(s1.west)$) edge (0,0);
			\node[anchor=west] at (5.8,0.75) {\small \uniqueA~ unique $\mathbb{A}_1$};
			\node[anchor=west] at (5.8,0.25) {\small \uniqueB~ unique $\mathbb{A}_2$};
			\node[anchor=west] at (5.8,-0.25) {\small \redundancy~ redundant};
			\node[anchor=west] at (5.8,-0.75) {\small \synergy~ synergetic};
		\end{scope}
\end{tikzpicture}}
		\caption{\textbf{Synergetic example.} Fully synergetic contribution by both attributes.}
		\label{fig:synergy-example}
	\end{figure}
\end{Example}

\subsubsection{Game theoretic synergy is insufficient}
\label{subsubset:insufficient}
As it could already be seen (\refEq{eq:intro-attributedecomp} in Section \ref{sec:intro} and the previous examples), the desired attribute decomposition builds on Assumption \ref{as:pid-assumption}:
\begin{Assumption}
	Inequality can be decomposed into non-negative redundant, unique, and synergetic contributions as indicated by \refEq{eq:pid-assumption} and \refFig{fig:intro-proposal} for the case of two attributes.
	\begin{subequations}
		\begin{align}
		&\begin{aligned}
			\text{inequality by $\mathbb{A}_1$ and $\mathbb{A}_2$}~ = &\ \redundancy~ \text{redundant inequality by attribute $\mathbb{A}_1$ and $\mathbb{A}_2$}\\
			&+ \uniqueA~ \text{unique inequality by attribute $\mathbb{A}_1$}\\
			&+ \uniqueB~ \text{unique inequality by attribute $\mathbb{A}_2$}\\
			&+ \synergy~ \text{synergetic inequality by attribute $\mathbb{A}_1$ and $\mathbb{A}_2$}\\
		\end{aligned}\vspace{3mm}\\
		&\qquad\quad\begin{aligned}
			\text{inequality by $\mathbb{A}_1$}~ = &\ \redundancy~ \text{redundant inequality by attribute $\mathbb{A}_1$ and $\mathbb{A}_2$}\\
			&+ \uniqueA~ \text{unique inequality by attribute $\mathbb{A}_1$}
		\end{aligned}\vspace{3mm}\\
		&\qquad\quad\begin{aligned}
			\text{inequality by $\mathbb{A}_2$}~ = &\ \redundancy~ \text{redundant inequality by attribute $\mathbb{A}_1$ and $\mathbb{A}_2$}\\
			&+ \uniqueB~ \text{unique inequality by attribute $\mathbb{A}_2$}
		\end{aligned}
	\end{align}
	\label{eq:pid-assumption}
	\end{subequations}
	\label{as:pid-assumption}
\end{Assumption}
The corresponding decomposition for Assumption \ref{as:pid-assumption} is challenging since it requires quantifying four partial contributions, while only three cumulative contributions can be measured ($\Gamma(\{1\},\mathbf{M})$, $\Gamma(\{2\},\mathbf{M})$, and $\Gamma(\{1,2\},\mathbf{M})$). The resulting system of equations is under-determined, which causes the necessity of extending the inequality measure to either a notion of intersection or union. The examples in Section \ref{subsec:intuition} avoided this issue by only discussing special cases where Property \ref{prop:ineq-1}-\ref{prop:ineq-5} imply that the redundant or synergetic contribution must be zero.

A (different) notion of synergy is already well established in game theory and the computation of Shapley values. Since Shapley values can be applied to inequality measure~\cite{deutsch2008shapley}, it raises the question of how the challenges mentioned above have been addressed in this setting:
\begin{Definition}[Game synergy and Shapley values~\cite{shapley1951notes,shapleySynergy}]
	Game synergy is a function $\textnormal{GS} : \mathcal{P}(\{1,..,n\})\rightarrow \mathbb{R}$, that takes a set of attribute indices and quantifies their synergy as shown in \refEq{eq:gamesyn} in its direct application to this setting.
	\begin{equation}
		\textnormal{GS}(\mathbf{a}) \definedAs \sum_{\emptyset\neq \mathbf{b}\subseteq \mathbf{a}} (-1)^{|\mathbf{a}|-|\mathbf{b}|} I_{f,p}(\Gamma(\mathbf{b},\mathbf{M}))
		\label{eq:gamesyn}
	\end{equation} 
	Game synergy can be used to compute Shapley values $\varphi : \{1,..,n\}\rightarrow \mathbb{R}$ (\refEq{eq:shapley}), which shall quantify the contribution of attribute $i$.
	\begin{equation}
		\varphi(i) \definedAs \sum_{\mathbf{a}\in\mathcal{P}(\{1,..,n\}\setminus\{i\})} \frac{\textnormal{GS}(\mathbf{a}\cup\{i\})}{|A|+1}
	\end{equation}
	\label{def:Shaley}
\end{Definition}
Game synergy does not consider the concept of redundancy and thus fails to separate it from synergy, as shown in \refEq{eq:shapley}. From our perspective, `game synergy` is the difference between synergy and redundancy.
\begin{subequations}
	\begin{align}
	&\qquad\qquad\qquad\quad \textnormal{GS}(\{1,2\}) = I_{f,p}(\Gamma(\{1,2\},\mathbf{M})) - I_{f,p}(\Gamma(\{1\},\mathbf{M})) - I_{f,p}(\Gamma(\{2\},\mathbf{M}))\\
	& \begin{aligned}
		\text{game synergy of $\mathbb{A}_1$ \& $\mathbb{A}_2$} 
		= &\ \redundancy\uniqueA\uniqueB\synergy~ \text{ inequality by attribute $\mathbb{A}_1$ \& $\mathbb{A}_2$}\\
		&- \redundancy\uniqueA~ \text{inequality by attribute $\mathbb{A}_1$}\\
		&- \redundancy\uniqueB~ \text{inequality by attribute $\mathbb{A}_2$}\\
		=&\ \synergy~\text{our notion of synergy} - \redundancy~\text{our notion of redundancy}
	\end{aligned}
	\end{align}
	\label{eq:shapley}
\end{subequations}
The interpretation of \refEq{eq:shapley} can be used to explain the negativity of game synergy and its consequent meaning:
we can interpret positive `game synergy` as indication of dominant synergetic interactions between attributes, while negative `game synergy` indicates dominant redundant interactions. Since both components may be present simultaneously (visualized in \refFig{fig:intro-proposal}) and in a canceling direction (highlighted in \refEq{eq:shapley}), it would be desirable to separate them. This would enable more detailed analyses and a more practical operational interpretation, as shown in Section \ref{subsec:decomposition-ineq-measure}.
\begin{Remark}
	An equivalent argument was made by \citet{williams-beer} for interaction information and motivated the research area of \aclp{PID}.
\end{Remark}
With this interpretation of game synergy, we can also provide an interpretation of Shapley values, as shown in \refEq{eq:shapley-2}. At two attributes, the shapely value of each attribute corresponds to its unique contribution plus half of their redundancy and synergy. As a result, the Shapley values of each attribute sum to the total amount: $I_{f,p}(\Gamma(\{1,2\},\mathbf{M})) = \varphi(1) + \varphi(2)$.
\begin{subequations}
	\begin{align}
	&\qquad\qquad\qquad\quad \varphi(1) = \frac{\textnormal{GS}(\{1\})}{1} + \frac{\textnormal{GS}(\{1,2\})}{2}	\\
	&\begin{aligned}
		\text{shapley value of } \mathbb{A}_1 
		= &\  \frac{\text{ game synergy of $\mathbb{A}_1$}}{1} + \frac{\text{ game synergy of $\mathbb{A}_1$ \& $\mathbb{A}_2$}}{2}\\
		=& \frac{\redundancy~ \text{ redundancy by $\mathbb{A}_1$ \& $\mathbb{A}_2$}}{2} + \uniqueA~ \text{unique by $\mathbb{A}_1$} + \frac{\synergy~ \text{ synergy by $\mathbb{A}_1$ \& $\mathbb{A}_2$}}{2}
	\end{aligned}
\end{align}
\label{eq:shapley-2}
\end{subequations}
This section provided an intuition for the desired concepts of redundancy and synergy. We highlighted the necessity of extending inequality measures to a notion of union or intersection since the decomposition is otherwise under-determined. We also explained our interpretation of game synergy and Shapley values and why we consider them insufficient for studying the interactions between attributes in this setting. Finally, both game synergy and Shapley values can be computed by combining the partial contributions of the following attribute decomposition, as indicated by \refEq{eq:shapley} and \refEq{eq:shapley-2}.

	\subsection{Decomposition lattice and required properties}
	\label{subsec:decompositoin-lattice}
%
This section presents the considered framework for an attribute decomposition and follows the general methodology of~\citet{williams-beer} from \aclp{PID}: we consider a lattice that captures the desired subset relation for a set-theoretic intuition and discuss the required properties for a cumulative measure on this lattice. The partial contributions are then obtained from the Möbius inverse, which enforces an inclusion-exclusion relation between them.

\begin{Definition}[Sources, atoms and union lattice \cite{williams-beer,Rosas_2020}]\hfill
	\begin{itemize}
		\item An attribute set $\mathbf{a} \in \mathcal{P}(\{1,..,n\})$ is a subset of all attribute indices that is used to construct a partition.\newline  For example: $\Gamma(\mathbf{a},\mathbf{M})$.
		\item An atom $\alpha\in\mathcal{A}(n)$ is a non-empty set of attribute sets defined by \refEq{eq:atom-set}. The cardinality of $\mathcal{A}(n)$ is one less than the $n$-th Dedekind number~\cite{gutknecht2023babel}. In this work, we use atoms to represent a notion of union. For example, the atom $\alpha=\{\{1\},\{2,3\}\}$ shall represent the union of inequality when partitioning on attribute $\mathbb{A}_1$ and $(\mathbb{A}_2,\mathbb{A}_3)$. 
		\begin{subequations}
			\begin{align}
			\mathcal{A}(n) &\definedAs \{\alpha \in \mathcal{P}_1(\mathcal{P}(\{1,..,n\}))~:~(\forall\mathbf{a},\mathbf{b}\in\alpha)[\neg(\mathbf{a}\subset\mathbf{b})]\}\label{eq:atom-set}\\
			(\alpha \preceq \beta) &\definedAs (\forall \mathbf{a}\in\alpha.~\exists \mathbf{b}\in\beta.~ \mathbf{a} \subseteq \mathbf{b})
			\label{eq:union-order}\\
			(\alpha \prec \beta) &\definedAs (\alpha \preceq \beta ~\text{ and }~ \neg(\beta\preceq\alpha))
			\end{align}
		\end{subequations}
		\item The set of atoms form a distributive lattice with the ordering of \refEq{eq:union-order}. We refer to the resulting lattice $(\mathcal{A}(n),\preceq)$ as union lattice~\cite{gutknecht2023babel,Kolchinsky,Gomes}.
	\end{itemize} 
\end{Definition}
\begin{Remark}
	We treat the union lattice as reversed synergy lattice. This enables the direct application of our results from \cite{mages2024}.
\end{Remark}
\begin{Notation}\hfill
	\begin{itemize}
		\item We notate the meet and join on the union lattice as $\alpha\curlywedge\beta$ and $\alpha\curlyvee\beta$ respectively.
		\item We notate the bottom and the top of the union lattice as $\bot_\cup=\{\emptyset\}$ and $\top_\cup=\{\{1,..,n\}\}$ respectively.
		\item We notate the upset and strict upset of on the union lattice as $\uparrow\alpha$ and $\dot\uparrow\alpha$  respectively.
	\end{itemize}
\end{Notation}
The union lattice for two and three attributes is visualized in \refFig{fig:lattice-examples}.
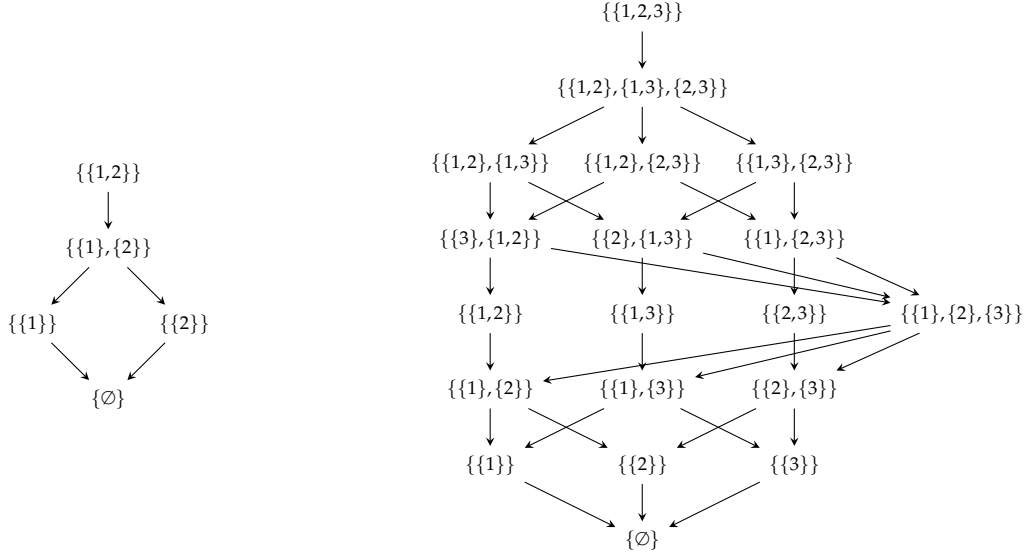
\begin{figure}[h]
	\centering
%
\begin{subfigure}[t]{0.4\linewidth}\centering
	\newcommand{\latSpace}{1}
	\begin{tikzpicture}[color=black]
		\node (Zs12) at (0,-1) {\scriptsize $\{\emptyset\}$};
		\node (Zu1) at (-\latSpace,0) {\scriptsize $\{\{1\}\}$};
		\node (Zu2) at (\latSpace,0) {\scriptsize $\{\{2\}\}$};
		\node (Zr12) at (0,1) {\scriptsize $\{\{1\},\{2\}\}$};
		\node (bot) at (0,2) {\scriptsize $\{\{1,2\}\}$};
		\foreach \fr/\to in {{Zs12/Zu2}, {Zs12/Zu1}, {Zu1/Zr12}, {Zu2/Zr12}, {Zr12/bot}}
		\path[-stealth]	(\to) edge (\fr);
		\node at (0,-3) {};
	\end{tikzpicture}
	\caption{Union lattice for two attributes ($n=2$).}
\end{subfigure}%
~
\begin{subfigure}[t]{0.6\linewidth}\centering
	\newcommand{\latSpace}{2}
	\begin{tikzpicture}[color=black]
		\node (Zs123) at (0,-3) { \scriptsize $\{\emptyset\}$};
		\node (Zs12) at (-\latSpace,-2) {\scriptsize $\{\{1\}\}$};
		\node (Zs13) at (0,-2) {\scriptsize $\{\{2\}\}$};
		\node (Zs23) at (\latSpace,-2) {\scriptsize $\{\{3\}\}$};
		\node (Zrs12s13) at (-\latSpace,-1) {\scriptsize $\{\{1\},\{2\}\}$};
		\node (Zrs12s23) at (0,-1) {\scriptsize $\{\{1\},\{3\}\}$};
		\node (Zrs13s23) at (\latSpace,-1) {\scriptsize $\{\{2\},\{3\}\}$};
		\node (Zu1) at (-\latSpace,0) {\scriptsize $\{\{1,2\}\}$};
		\node (Zu2) at (0,0) {\scriptsize $\{\{1,3\}\}$};
		\node (Zu3) at (\latSpace,0) {\scriptsize $\{\{2,3\}\}$};
		\node (Zrs12s13s23) at (2*\latSpace+0.2,0) {\scriptsize $\{\{1\},\{2\},\{3\}\}$};
		\node (Zrs1s23) at (-\latSpace,1) {\scriptsize $\{\{3\},\{1,2\}\}$};
		\node (Zrs2s13) at (0,1) {\scriptsize $\{\{2\},\{1,3\}\}$};
		\node (Zrs3s12) at (\latSpace,1) {\scriptsize $\{\{1\},\{2,3\}\}$};
		\node (Zr12) at (-\latSpace,2) {\scriptsize $\{\{1,2\},\{1,3\}\}$};
		\node (Zr13) at (0,2) {\scriptsize $\{\{1,2\},\{2,3\}\}$};
		\node (Zr23) at (\latSpace,2) {\scriptsize $\{\{1,3\},\{2,3\}\}$};
		\node (Zr123) at (0,3) {\scriptsize $\{\{1,2\},\{1,3\},\{2,3\}\}$};
		\node (bot) at (0,4) {\scriptsize $\{\{1,2,3\}\}$};
		\foreach \fr/\to in {{Zs123/Zs12}, {Zs123/Zs13}, {Zs123/Zs23}, {Zs12/Zrs12s13}, {Zs12/Zrs12s23}, {Zs13/Zrs12s13}, {Zs13/Zrs13s23}, {Zs23/Zrs12s23}, {Zs23/Zrs13s23}, {Zrs12s13/Zu1}, {Zrs12s13/Zrs12s13s23}, {Zrs12s23/Zu2}, {Zrs12s23/Zrs12s13s23}, {Zrs13s23/Zu3}, {Zrs13s23/Zrs12s13s23}, {Zrs12s13s23/Zrs1s23}, {Zrs12s13s23/Zrs2s13}, {Zrs12s13s23/Zrs3s12}, {Zu1/Zrs1s23}, {Zu2/Zrs2s13}, {Zu3/Zrs3s12}, {Zrs1s23/Zr12}, {Zrs1s23/Zr13}, {Zrs2s13/Zr12}, {Zrs2s13/Zr23}, {Zrs3s12/Zr13}, {Zrs3s12/Zr23}, {Zr12/Zr123}, {Zr13/Zr123}, {Zr23/Zr123}, {Zr123/bot}}
		\path[-stealth]	(\to) edge (\fr);
	\end{tikzpicture}
	\caption{Union lattice for three attributes ($n=3$).}
\end{subfigure}
	\caption{\textbf{Decomposition lattice visualization.} The union lattice for (\textbf{a}) two attributes and (\textbf{b}) three attributes.}
	\label{fig:lattice-examples}
\end{figure}

Using an inequality measure $I(\cdot)$, we can already quantify the inequality for an attribute set $\mathbf{a}\in\alpha$ as $I(\Gamma(\mathbf{a},\mathbf{M}))$. However, this provides fever equations than free variables when calculating partial contributions (under-determined) as discussed in Section \ref{subsubset:insufficient}. To fully determine the system, we have to extend inequality measures from attribute sets to atoms (cumulative measure) and can express partial contributions as computation on the decomposition lattice (partial measure)~\cite{williams-beer}. We first introduce both definitions and then discuss the required properties for achieving the desired set-theoretic analogy:
\begin{Definition}[Cumulative measure: union inequality  $I^\cup$]
	The union inequality $I^\cup(\alpha,\mathbf{M})$ is a function that assigns a real value to every atom of the union lattice. It is a cumulative measure that shall satisfy Property \ref{prop:union-1}-\ref{prop:union-4} defined below.
\end{Definition}
\begin{Definition}[Partial inequality contributions $I^\delta$]
	The partial inequality (redundant, unique, synergetic) contributions $I^\delta(\alpha,\mathbf{M})$ are defined by the Möbius inverse~\cite{rota1964foundations,williams-beer} on the reversed lattice~\cite{mages2024} as shown in \refEq{eq:union-möbius}.
	\begin{subequations}
		\begin{align}
			I^\cup(\top_\cup,\mathbf{M}) - I^\cup(\alpha,\mathbf{M}) &\definedAs \sum_{\beta \in \uparrow\alpha}I^\delta(\beta,\mathbf{M})\\
			I^\delta(\alpha,\mathbf{M}) &= I^\cup(\top_\cup,\mathbf{M}) - I^\cup(\alpha,\mathbf{M}) - \sum_{\beta \in \dot\uparrow\alpha}I^\delta(\beta,\mathbf{M})\label{subeq:möbius-inv}
		\end{align}
		\label{eq:union-möbius}
	\end{subequations}
	\label{def:union-möbius}
\end{Definition}
The following properties for a cumulative measure are typically presented as axioms in the context of \aclp{PID}~\cite{williams-beer,dualdecomposition} and can directly be transferred to inequality measures: 
\begin{PropertyU}[Commutativity~\cite{williams-beer,dualdecomposition}]
	A notion of union inequality is invariant to the order of attribute sets. Let $\sigma:\alpha\rightarrow \alpha$ permute the order of attribute sets in an atom. 
	\begin{subequations}
		\begin{align}
			\makebox[5.5cm][l]{$\forall\alpha\in\mathcal{A}(n):$}&I^\cup(\alpha,\mathbf{M}) = I^\cup(\{\sigma(\mathbf{a})~:~\mathbf{a}\in\alpha\},\mathbf{M})\\
			\makebox[5.5cm][l]{\textnormal{Example:}}&I^\cup(\{\{1\},\{2,3\}\},\mathbf{M}) = I^\cup(\{\{2,3\},\{1\}\},\mathbf{M})
		\end{align}
	\end{subequations}
	\label{prop:union-1}
\end{PropertyU}
\begin{PropertyU}[Monotonicity~\cite{williams-beer,dualdecomposition}]
	Adding an attribute set to an atom can only increase their union inequality:
	\begin{subequations}
		\begin{align}
		\makebox[5.5cm][l]{$\forall\alpha\in\mathcal{A}(n),~\forall \mathbf{a}\in\mathcal{P}(\{1,..,n\}):$} &I^\cup(\alpha,\mathbf{M}) \leq I^\cup(\alpha\cup\{\mathbf{a}\},\mathbf{M})\\
		\makebox[5.5cm][l]{\textnormal{Example:}}&I^\cup(\{\{2,3\}\},\mathbf{M}) \leq I^\cup(\{\{2,3\},\{1\}\},\mathbf{M})
		\end{align}
	\end{subequations}
	\label{prop:union-2}
\end{PropertyU}
\begin{PropertyU}[Self-inequality~\cite{williams-beer,dualdecomposition}]
The union of a single attribute set equals the desired inequality measure.
\begin{subequations}
	\begin{align}
	\makebox[5.5cm][l]{$\forall \mathbf{a}\in\mathcal{P}(\{1,..,n\}):$}&I^\cup(\{\mathbf{a}\},\mathbf{M}) = I(\Gamma(\mathbf{a},\mathbf{M}))\\
	\makebox[5.5cm][l]{\textnormal{Example:}}&I^\cup(\{\{2,3\}\},\mathbf{M}) = I(\Gamma(\{2,3\},\mathbf{M}))
	\end{align}
\end{subequations}
\label{prop:union-3}
\end{PropertyU}
\begin{PropertyU}[Non-negativity~\cite{williams-beer,dualdecomposition}]
The partial inequality contributions are non-negative.
\begin{equation}
	\makebox[5.5cm][l]{$\forall\alpha\in\mathcal{A}(n):$} \makebox[5.5cm][l]{$I^\delta(\alpha,\mathbf{M}) \geq 0$}
\end{equation}
\label{prop:union-4}
\end{PropertyU}
The combination of Property \ref{prop:union-2} and the union lattice ensures the expected subset relation. Property \ref{prop:union-3} binds the union measure to the desired inequality measure. Property \ref{prop:union-4} ensures the interpretability of results by enabling the analogy from a population's inequality to a set's cardinality.
Finally, \refFig{fig:lattice-intution} visualizes the relation between a Venn diagram and the used decomposition lattice at the example of $n=2$. Except for the top element, each partial contribution on the union lattice $I^\delta(\cdot,\mathbf{M})$ corresponds to a partial region of the Venn diagram.
\begin{figure}[h]
	\centering
%
\begin{subfigure}[t]{0.4\linewidth}
	\scalebox{0.87}{\begin{tikzpicture}[color=black]
			\begin{scope}[xshift=-1.75cm]
				\node[anchor=west] (s12) at (1.9,-0.50) {\small $I(\Gamma(\{1,2\},\mathbf{M}))$};
				\node[anchor=west] (s1) at (1.9,  0.25) {\small $I(\Gamma(\{1\},\mathbf{M}))$};
				\node[anchor=west] (s2) at (1.9,-1.25) {\small $I(\Gamma(\{2\},\mathbf{M}))$};
				\draw[draw=black,fill=my-green!40] (0,-1) circle (1);
				\draw[draw=black,fill=blue!40] (0,0) circle (1);
				\begin{scope}[draw=black]
					\clip (0,0) circle (1);
					\clip (0,-1) circle (1);
					\path [fill=orange!40] (0,-1) circle (1);
				\end{scope}
				\begin{scope}[draw=black,even odd rule]
					\clip (0,0) circle (1) (0,-0.5) ellipse (1.4 and 2.2);
					\clip (0,-1) circle (1) (0,-0.5) ellipse (1.4 and 2.2);
					\fill[magenta!40] (0,-0.5) ellipse (1.4 and 2.2);
				\end{scope}
				\path[-stealth]
				($(s12.west)$) edge ($(s12.west)+(-0.45,0)$)
				($(s1.west)$) edge ($(s1.west)+(-0.85,-0.25)$)
				($(s2.west)$) edge ($(s2.west)+(-0.85,+0.25)$);
				\draw (0,-0.5) ellipse (1.4 and 2.2) node[above,yshift=45,circle,fill=white,inner sep=2pt] {\footnotesize 4};
				\node[circle,fill=white,inner sep=2pt, yshift=13] at (0,0) {\footnotesize 1};
				\node[circle,fill=white,inner sep=2pt] at (0,-0.5) {\footnotesize 3};
				\node[circle,fill=white,inner sep=2pt, yshift=-13] at (0,-1) {\footnotesize 2};
				\draw[draw=black] (0,-1) circle (1);
				\draw[draw=black] (0,0) circle (1);
				\draw[draw=black] (0,-0.5) ellipse (1.4 and 2.2);
			\end{scope}
			\node[anchor=west] at (-0.3,-3.7) {\small \synergy~ synergetic};
			\node[anchor=west] at (-3.0,-3.2) {\small \uniqueA~ unique $\mathbb{A}_1$};
			\node[anchor=west] at (-0.3,-3.2) {\small \redundancy~ redundant};
			\node[anchor=west] at (-3.0,-3.7) {\small \uniqueB~ unique $\mathbb{A}_2$};
	\end{tikzpicture}}
	\caption{Representation as Venn diagram}
\end{subfigure}%
\begin{subfigure}[t]{0.4\linewidth}
	\newcommand{\latSpace}{2}
	\begin{tikzpicture}[color=black,anode/.style={text width=3.2cm,align=center,font=\scriptsize}]
		\node[anode] (Zs12) at (0,-1.5) {\baselineskip=10pt $I^\cup(\{\emptyset\}, \mathbf{M}) = 0$ \\$I^\delta(\{\emptyset\}, \mathbf{M}) = \redundancy$};
		\node[anode] (Zu1) at (-\latSpace,0) {\baselineskip=12pt $I^\cup(\{\{1\}\}, \mathbf{M}) = \redundancy\uniqueA$ \\$I^\delta(\{\{1\}\}, \mathbf{M}) = \uniqueB$};
		\node[anode] (Zu2) at (\latSpace,0) {\baselineskip=12pt $I^\cup(\{\{2\}\}, \mathbf{M}) = \redundancy\uniqueB$ \\$I^\delta(\{\{2\}\}, \mathbf{M}) = \uniqueA$};
		\node[anode] (Zr12) at (0,1.5) {\baselineskip=12pt $I^\cup(\{\{1\},\{2\}\}, \mathbf{M}) = \redundancy\uniqueA\uniqueB$ \\$I^\delta(\{\{1\},\{2\}\}, \mathbf{M}) = \synergy$};
		\node[anode] (bot) at (0,3) {\baselineskip=12pt $I^\cup(\{\{1,2\}\}, \mathbf{M}) = \redundancy\uniqueA\uniqueB\synergy$ \\$I^\delta(\{\{1,2\}\}, \mathbf{M}) =  0$};
		\foreach \fr/\to in {{Zs12/Zu2}, {Zs12/Zu1}, {Zu1/Zr12}, {Zu2/Zr12}, {Zr12/bot}}
		\path[-stealth]	(\to) edge (\fr);
	\end{tikzpicture}
	\caption{Representation as union lattice}
\end{subfigure}%
	\caption{\textbf{Visualization for the relation between Venn diagrams and the union lattice at two attributes ($n=2$).} Representation of partial contributions as (\textbf{a}) Venn diagram and (\textbf{b}) union lattice. The partial contribution of the top element is always zero. All other partial contributions of an atom on the union lattice ($I^\delta$) correspond to exactly one partial region in the Venn diagram.}
	\label{fig:lattice-intution}
\end{figure}
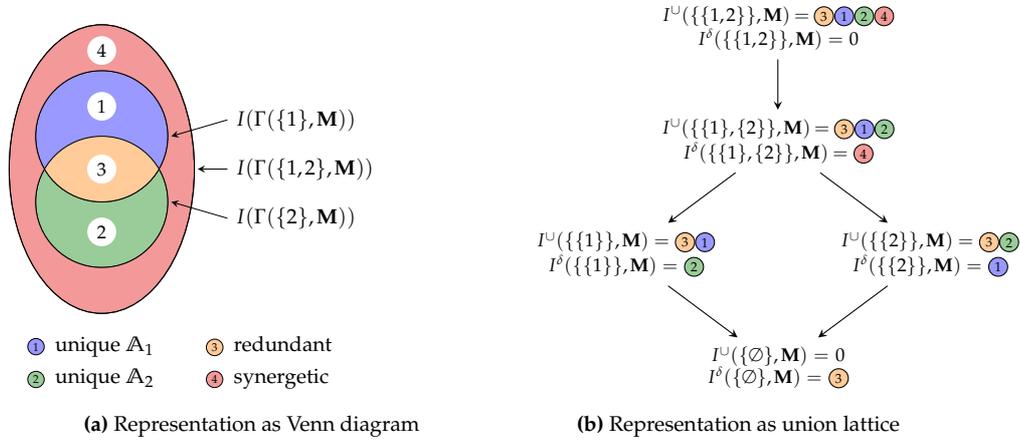

This section discussed the considered decomposition framework based on the union lattice and the necessary properties for the required cumulative measure.

\begin{Remark}
	If desired, the decomposition on the union lattice can be transformed into a decomposition on the redundancy lattice~\cite{williams-beer} as described in \cite[Section 3.4]{mages2024}.
\end{Remark}

	\subsection{Decomposing f-inequality}
	\label{subsec:decomposition-ineq-measure}
%
With the decomposition framework of Section \ref{subsec:decompositoin-lattice}, we only have to define a union inequality measure ($I^\cup_{f,p}$) for $f$-inequality to obtain its decomposition. To achieve the required properties and a practical operational interpretation, we use the join of the zonogon order (convex-hull), as shown in Definition \ref{def:def-f-union}. Intuitively, this appears suitable since it reflects the unique and measure independent optimal dependence between attributes as a notion of their union.
\begin{Definition}[$f$-inequality union]
	We define the union of two partitions by their join under the zonogon order.
	\begin{equation}
		I^\cup_{f,p}(\alpha,\mathbf{M}) \definedAs I_{f,p}\left(\bigsqcup_{\mathbf{a}\in\alpha}\langle\Gamma(\mathbf{a},\mathbf{M})\rangle\right)
	\end{equation}
	\label{def:def-f-union}
\end{Definition}
\begin{Theorem}
	Definition \ref{def:def-f-union} satisfies Property \ref{prop:union-1}-\ref{prop:union-4}.
	\label{thm:decom-proofs}
\end{Theorem}

The proof of Theorem \ref{thm:decom-proofs} is shown in Appendix \ref{ap:decomposition-properties}. 
\begin{Remark}
	Appendix \ref{ap:implementation} shows that we can compute partial contributions in a practical implementation using \refEq{eq:imp-partial}, where $\mathcal{C}(\cdot)$ is the n-ary Cartesian product. We recommend caching the cumulative measure $I^\cup(\cdot,\mathbf{M})$ to avoid repeated computations. This implementation is advantageous by computing the Möbius inverse without having to identify and visit each element in the strict upset of an atom ($\dot{\uparrow}\alpha$) as the lattice $|\mathcal{A}(n)|$ grows rapidly in $n$.
	\begin{subequations}
		\begin{align}
			\textnormal{reduce}(\sim,\alpha)&= \{\mathbf{a}\in\alpha~:~\neg(\exists\mathbf{b}\in\alpha)[\mathbf{a}\sim\mathbf{b}]\}\\
			\textnormal{dual}(\alpha) &= \textnormal{reduce}(\supset,~\mathcal{C}(\{\{1,..,n\}\setminus \mathbf{a}~:~\mathbf{a}\in\alpha\}))\\
			I^\delta_{f,p}(\alpha,\mathbf{M}) &=  \begin{cases}
				0 & \text{if } \alpha=\top_\cup\\ \sum_{\beta\in\mathcal{P}(\textnormal{dual}(\alpha))}(-1)^{|\beta|-1}~I^\cup_{f,p}(\textnormal{reduce}(\subset,~\alpha\cup\beta),\mathbf{M})& \text{otherwise}
			\end{cases}
		\end{align}
		\label{eq:imp-partial}
	\end{subequations}
\end{Remark}
The resulting operational interpretation depends on the type of $f$-inequality:
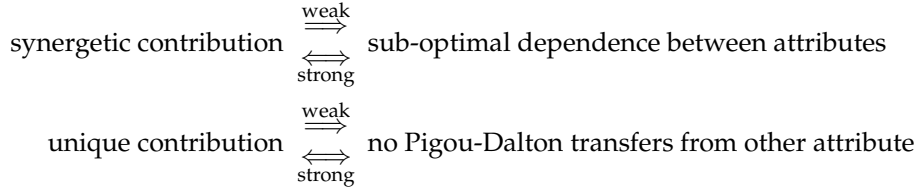
\begin{equation*}
	\begin{aligned}
		\text{synergetic contribution} ~~&\begin{matrix}
			\overset{\text{weak}}{\large  \Longrightarrow}\\ \underset{\text{strong}}{\normalsize\Longleftrightarrow}
		\end{matrix}~~\text{sub-optimal dependence between attributes} \\
		\text{unique contribution} ~~&\begin{matrix}
		\overset{\text{weak}}{\large  \Longrightarrow}\\ \underset{\text{strong}}{\normalsize\Longleftrightarrow}
		\end{matrix}~~\text{no Pigou-Dalton transfers from other attribute}\\
	\end{aligned}
\end{equation*}
Synergetic contributions indicate that inequality can be reduced by re-distributing the indicator variable based on the dependence between attributes or suitably increasing the dependence between attributes. Unique contributions can be reduced by re-distributing the indicator variable based on the specific attribute or changing the distribution of this attribute. As it can be seen from Corollary \ref{col:f-ineq-inc-enc}, the resulting notion of redundancy is lower bound by the quantification of the zonogon meet (intersection).
\begin{Example}
		Consider the model $\textbf{M}$ obtained from Table \ref{tbl:example-decomp} with the two attributes $\mathbb{A}_1 = \{A,B\}$ and $\mathbb{A}_2 = \{C,D\}$.
	\begin{table}[h!]\centering
		\begin{tabular}{r | c c | l}\toprule
			Indicator value & $\mathbb{A}_1$ &  $\mathbb{A}_2$ & Number of individuals\\\midrule
			0 & A & C & 170\\
			\sfrac{3}{150} & A & D & 150\\
			\sfrac{1}{30}   & B & C & 30 \\
			\sfrac{6}{50} & B & D & 50 \\\bottomrule
		\end{tabular}
		\caption{Example population model}
		\label{tbl:example-decomp}
	\end{table}
	The corresponding population matrices and zonogons for model $\mathbf{M}$ are visualized in \refFig{fig:decomp-zonogons}.
	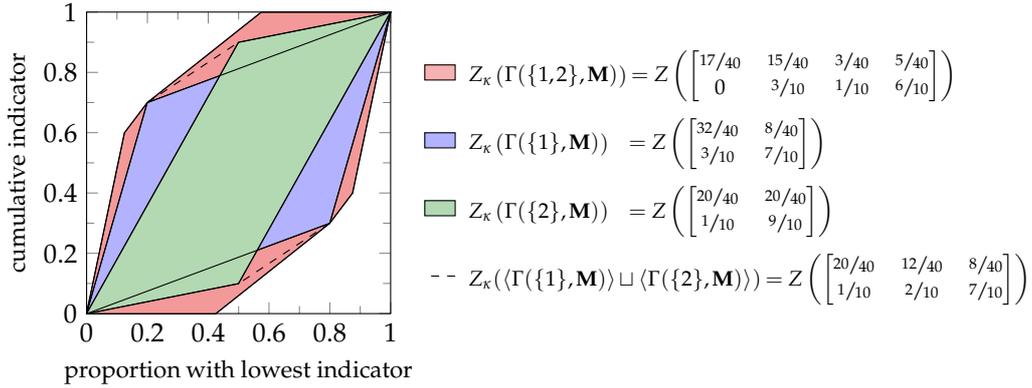
\begin{figure}[h]
		\centering
%
\begin{tikzpicture}[color=black]
	\begin{axis}[height=5.58cm, width=5.58cm,xmin=0, xmax=1,ymin=0, ymax=1, ylabel={\small cumulative indicator}, xlabel={\small proportion with lowest indicator},minor x tick num=1,minor y tick num=1,clip=false,
		legend style={at={(1.1,0.5)},anchor=west,draw=none,legend plot pos=left,legend cell align=left}]
	\end{axis};
	\draw[draw=black,fill=magenta!40,fill opacity=0.5] (0,0) -- (4*17/40,4*0/10) -- (4*32/40,4*3/10) -- (4*35/40,4*4/10) -- (1*4,1*4) -- (4-4*17/40,4-4*0/10) -- (4-4*32/40,4-4*3/10) -- (4-4*35/40,4-4*4/10) -- (0,0);
	\draw[draw=black,fill=blue!30,fill opacity=0.5] (0,0) -- (4*32/40,4*3/10) -- (4,4) -- (4-4*32/40,4-4*3/10) -- (0,0);
	\draw[draw=black,fill=my-green!30,fill opacity=0.5] (0,0) -- (4*20/40,4*1/10) -- (4,4)-- (4-4*20/40,4-4*1/10) --(0,0);
	\foreach \col/\label [count=\i] in {{my-green/{$Z_\kappa\left(\Gamma(\{2\},\mathbf{M})\right)~~~= Z\left(\begin{bmatrix}
					\sfrac{20}{40} & \sfrac{20}{40} \\
					\sfrac{1}{10} & \sfrac{9}{10} 
				\end{bmatrix}\right)$}},{blue/{$Z_\kappa\left(\Gamma(\{1\},\mathbf{M})\right)~~~= Z\left(\begin{bmatrix}
					\sfrac{32}{40} & \sfrac{8}{40}\\
					\sfrac{3}{10} & \sfrac{7}{10} 
				\end{bmatrix}\right)$}},{magenta/$Z_\kappa\left(\Gamma(\{1,2\},\mathbf{M})\right) = Z\left(\begin{bmatrix}
				\sfrac{17}{40} & \sfrac{15}{40} & \sfrac{3}{40} & \sfrac{5}{40} \\
				0 & \sfrac{3}{10} & \sfrac{1}{10} & \sfrac{6}{10}
			\end{bmatrix}\right)$}}{
		\node[anchor=west] at ($(4.3,3) + (0.6,-2.55) - \i*(0,-0.9)$) {\footnotesize \label\vphantom{q}};%
		\draw[fill=\col!30] ($(4.3,3) + (0.35,-2.5) - \i*(0,-0.9) + (0.2,0.1)$) rectangle ($(4.3,3) + (0.35,-2.5) - \i*(0,-0.9) - (0.2,0.1)$);}
		\draw[draw=black] (0,0) -- (4*17/40,4*0/10) -- (4*32/40,4*3/10) -- (4*35/40,4*4/10) -- (1*4,1*4) -- (4-4*17/40,4-4*0/10) -- (4-4*32/40,4-4*3/10) -- (4-4*35/40,4-4*4/10) -- (0,0);
		\draw[draw=black] (0,0) -- (4*32/40,4*3/10) -- (4,4) -- (4-4*32/40,4-4*3/10) -- (0,0);
		\draw[draw=black] (0,0) -- (4*20/40,4*1/10) -- (4,4)-- (4-4*20/40,4-4*1/10) --(0,0);
	\draw[dashed] (4*32/40,4*3/10) -- (4*20/40,4*1/10);
	\draw[dashed] (4-4*32/40,4-4*3/10) -- (4-4*20/40,4-4*1/10);
	\draw[dashed] ($(4.3,3) + (0.35,-2.5) + (0.2,0.0)$) -- ($(4.3,3) + (0.35,-2.5) - (0.2,0.0)$);
	\node[anchor=west] at ($(4.3,3) + (0.6,-2.55)$) {\footnotesize $Z_\kappa(\langle\Gamma(\{1\},\mathbf{M})\rangle\sqcup\langle\Gamma(\{2\},\mathbf{M})\rangle) = Z\left(\begin{bmatrix}
			\sfrac{20}{40} & \sfrac{12}{40}& \sfrac{8}{40} \\
			\sfrac{1}{10} & \sfrac{2}{10} & \sfrac{7}{10} 
		\end{bmatrix}\right)$};
\end{tikzpicture}
		\caption{\textbf{Visualization of the model from Table \ref{tbl:example-decomp}.}}
		\label{fig:decomp-zonogons}
	\end{figure}
	
	To analyze the model, we first define an inequality measure that suitably captures the required properties for the specific application.
	This is important since it determines how (zonogon) incomparable populations shall be ranked. For $f$-inequality, this is determined by the $(f,p)$ combination. Assume we consider the population $\Gamma(\{2\},\mathbf{M})$ preferable over $\Gamma(\{1\},\mathbf{M})$ and thus want to assign it a smaller inequality index.
	Without further information, we arbitrarily choose the inequality measure obtained from Definition \ref{def:generalized-ineq} using the $\chi^2$-divergence $f(t)= (t-1)^2$ with $p=0.4$, as shown in \refEq{eq:example-ineq-measure}. 
	\begin{equation}
		\begin{aligned}
		I_{\chi^2,0.4}(\mathbf{S}) 
		 &= \sum_{\vec{v}\in\kappa(\mathbf{S})} r(\vec{v}) &&\text{where: ~}~ r\left(\left[\begin{smallmatrix}x\\y\end{smallmatrix}\right]\right) = \frac{0.9\cdot (x - y)^2}{x + 1.5\cdot y}\\
		 &=  \frac{1}{\overline{\mathbf{S}}|\mathbf{S}|}\sum_{s\in\mathbf{S}}\frac{0.9\cdot\left(\overline{\mathbf{S}}-s\right)^2}{\overline{\mathbf{S}}+1.5\cdot s}
		\end{aligned}
		\label{eq:example-ineq-measure}
	\end{equation}
	We can compute the attribute decomposition using Definition \ref{def:def-f-union} and \refEq{eq:union-möbius} or \refEq{eq:imp-partial}.
	The results are visualized in \refFig{fig:decomp-results}.
	\begin{figure}[h]
		\centering
%
\begin{subfigure}[t]{0.5\linewidth}\centering
	\newcommand{\latSpace}{2}
	\begin{tikzpicture}[color=black,anode/.style={text width=3.5cm,align=center,font=\scriptsize}]
		\node[anode] (Zs12) at (0,-1.5) {\baselineskip=10pt $I^\cup_{\chi^2,0.4}(\{\emptyset\}, \mathbf{M}) = 0$ \\$I^\delta_{\chi^2,0.4}(\{\emptyset\}, \mathbf{M}) = 0.2429$};
		\node[anode] (Zu1) at (-\latSpace,0) {\baselineskip=12pt $I^\cup_{\chi^2,0.4}(\{\{1\}\}, \mathbf{M}) = 0.3600$ \\$I^\delta_{\chi^2,0.4}(\{\{1\}\}, \mathbf{M}) = 0.0565$};
		\node[anode] (Zu2) at (\latSpace,0) {\baselineskip=12pt $I^\cup_{\chi^2,0.4}(\{\{2\}\}, \mathbf{M}) = 0.2994$ \\$I^\delta_{\chi^2,0.4}(\{\{2\}\}, \mathbf{M}) = 0.1171$};
		\node[anode] (Zr12) at (0,1.5) {\baselineskip=12pt $I^\cup_{\chi^2,0.4}(\{\{1\},\{2\}\}, \mathbf{M}) = 0.4165$ \\$I^\delta_{\chi^2,0.4}(\{\{1\},\{2\}\}, \mathbf{M}) = 0.1727$};
		\node[anode] (bot) at (0,3) {\baselineskip=12pt $I^\cup_{\chi^2,0.4}(\{\{1,2\}\}, \mathbf{M}) =0.5892$ \\$I^\delta_{\chi^2,0.4}(\{\{1,2\}\}, \mathbf{M}) =  0$};
		\foreach \fr/\to in {{Zs12/Zu2}, {Zs12/Zu1}, {Zu1/Zr12}, {Zu2/Zr12}, {Zr12/bot}}
		\path[-stealth]	(\to) edge (\fr);
	\end{tikzpicture}
	\caption{Decomposition lattice}
\end{subfigure}%
\begin{subfigure}[t]{0.4\linewidth}\centering
	\scalebox{0.9}{\begin{tikzpicture}[color=black]
			\begin{axis}[
				set layers=axis on top,
				width=7cm, height=2.5cm, xbar stacked, axis lines=left,
				xlabel={\small Inequality $I_{\chi^2,0.4}$},
				ymin=0.5, ymax=1.5, 
				xmin=0, xmax=0.625,clip=false,
				ytick={1}, yticklabels={},
				]
				\addplot [fill=orange!40] coordinates {(0.2429, 1)};
				\addplot [fill=my-green!40] coordinates {(0.0565, 1)};
				\addplot [fill=blue!40] coordinates {(0.1171, 1)};
				\addplot [fill=magenta!40] coordinates {(0.1727, 1)};
			\end{axis}
			\node[circle,fill=white,inner sep=1pt] at (1.1,0.45) {\tiny 3};
			\node[circle,fill=white,inner sep=1pt] at (2.35,0.45) {\tiny 2};
			\node[circle,fill=white,inner sep=1pt] at (3.1,0.45) {\tiny 1};
			\node[circle,fill=white,inner sep=1pt] at (4.4,0.45) {\tiny 4};
			\begin{scope}[xshift=3cm,yshift=1.75cm]
				\node[anchor=west] at (-0.3,-3.7) {\small \synergy~ synergetic};
				\node[anchor=west] at (-3.0,-3.2) {\small \uniqueA~ unique $\mathbb{A}_1$};
				\node[anchor=west] at (-0.3,-3.2) {\small \redundancy~ redundant};
				\node[anchor=west] at (-3.0,-3.7) {\small \uniqueB~ unique $\mathbb{A}_2$};
			\end{scope}
			\node at (0,-3.7) {};
	\end{tikzpicture}}
	\caption{Result visualization}
\end{subfigure}%
		\caption{\textbf{Result visualization (rounded).} Considering the model from Table \ref{tbl:example-decomp} with the inequality measure of \refEq{eq:example-ineq-measure}.}
		\label{fig:decomp-results}
	\end{figure}
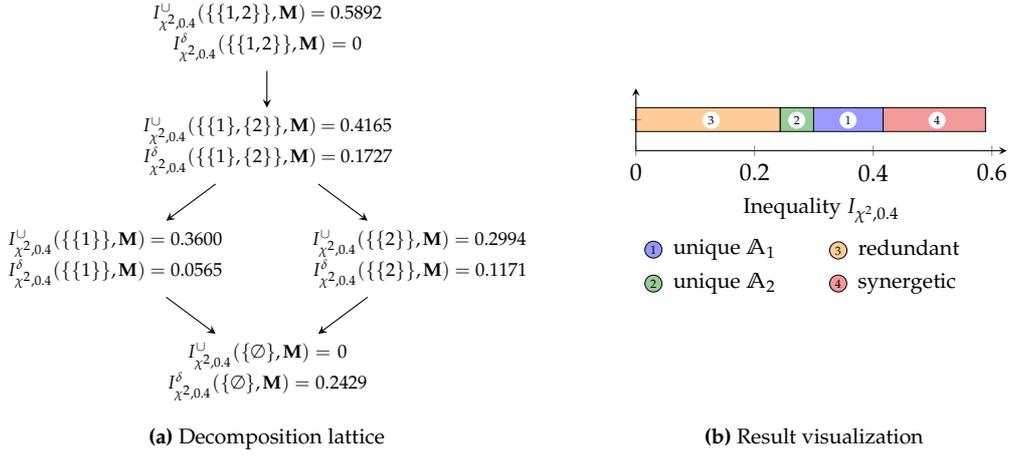
\end{Example}

	\subsection{Decomposing the Atkinson index}
	\label{subsec:transforming-ineq-measure}
%
The presented attribute decomposition can be extended to any invertible transformation of an  $f$-inequality. We demonstrate the approach using the Atkinson index (\refEq{eq:atkinson-index} with $d\in(0,1]$), which is a transformation of an $f$-inequality (Generalized Entropy index, Theorem \ref{thm:special-cases} and \refEq{eq:atkinson-index}) as shown in \refEq{eq:atkinson-transform}.
\begin{subequations}
	\begin{align}
		v_d(m) &=\begin{cases}
			1-e^{-m} & \text{if } d = 1\\
			1-[d(d-1)m+1]^{\frac{1}{1-d}}& \text{if } 0 < d < 1\\
		\end{cases} \label{eq:atkinson-transform-1}\\
		A_d(\mathbf{S}) &= v_d(\text{GE}_{1-d}(\mathbf{S}))
	\end{align}
	\label{eq:atkinson-transform}
\end{subequations}
The transformation function $v_d$ correctly maintains the bottom element at zero ($v_d(0) = 0$) and is invertible: the case of $d=1$ is invertible, and the case of $0 < d < 1$ is invertible for the required domain shown in \refEq{eq:transform-range}. Therefore, we can view the Atkinson index as a re-graduation~\cite{Knuth2019} on the decomposition lattice.
\begin{subequations}
	\begin{align}
			0 &< d < 1\\
			\text{GE}_{1-d}(\bot_\mathbf{S}) = \text{GE}_{1-d}(\{1\}) = 0 &\leq m \leq \frac{1 - 2^{-d}}{d - d^2}  = \text{GE}_{1-d}(\{1,0\})= \text{GE}_{1-d}(\top_\mathbf{S})
	\end{align}
	\label{eq:transform-range}
\end{subequations}
The attribute decomposition of the Atkinson index is obtained by \refEq{eq:Atkinson-decomp}~\cite[Section 3.6]{mages2024}, which maintains the operational interpretation of the decomposition from the Generalized Entropy index.
\begin{subequations}
	\begin{align}
		f_c(t) &\definedAs \begin{cases}
			t \ln(t)  & \text{if } c = 0\\
			\frac{t^{1-c} - t}{(c - 1) c}  \qquad& \text{if } 0<c<1
		\end{cases}\\
		A_d^\cup(\alpha,\mathbf{M}) &\definedAs v_d(I^\cup_{f_{1-d},0}(\alpha,\mathbf{M}))\\
		A_d^\delta(\alpha,\mathbf{M}) &\definedAs v_d(I^\delta_{f_{1-d},0}(\alpha,\mathbf{M}))
	\end{align}
	\label{eq:Atkinson-decomp}
\end{subequations}
The resulting re-graduation of the decomposition lattice satisfies the inclusion-exclusion relation (Möbius inverse, \refEq{eq:union-möbius}) under a transformed definition of addition shown in Definition \ref{def:Atkinson-addition}. This corresponds to the addition under which the partial contributions sum to the total amount. Moreover, the transformation maintains the desired Properties  \ref{prop:union-1}-\ref{prop:union-4}~\cite[Section 3.6]{mages2024}.
\begin{Definition}[Addition of Atkinson index inequality]
	We define the addition $\oplus_d$ and subtraction $\ominus_d$ on the cumulative and partial Atkinson index ($A_d^\diamond(\mathbf{S})$ where $\diamond\in\{\delta,\cup\}$)  by:
	\begin{subequations}
		\begin{align}
			A_d^\diamond(\mathbf{S}_1) \oplus_d A_d^\diamond(\mathbf{S}_2) &~\definedAs~ v_d(v_d^{-1}(A_d^\diamond(\mathbf{S}_1)) + v_d^{-1}(A_d^\diamond(\mathbf{S}_2)))\\
			A_d^\diamond(\mathbf{S}_1) \ominus_d A_d^\diamond(\mathbf{S}_2) &~\definedAs~ v_d(v_d^{-1}(A_d^\diamond(\mathbf{S}_1)) - v_d^{-1}(A_d^\diamond(\mathbf{S}_2)))
		\end{align}
	\end{subequations}
	\label{def:Atkinson-addition}
\end{Definition}
An interesting observation is that combining the resulting partial contributions into a Shapley value (\refEq{eq:shapley-2} using the addition of Definition \ref{def:Atkinson-addition}) is equivalent to computing the Shapley value directly from Definition \ref{def:Shaley} and Definition \ref{def:Atkinson-addition}. This indicates that the addition used when computing Shapley values should depend on the used inequality measure to maintain consistency between their results.

	\subsection{Multi-layered inequality}
	\label{subsec:layered-inequality}
%
In practical scenarios, inequality may appear in different layers, such as income, healthcare, or housing. Reducing inequality in complex systems forms a multi-objective optimization problem, which can be approached by scalarization (weighting method)~\cite[p. 10 ff.]{branke2008multiobjective}. We can construct an overall inequality measure as weighted sum of $f$-inequalities from each layer. 
Since we can use different $f$-inequalities on each layer, we can precisely specify which distributions are more desirable on each layer. While a Pigou-Dalton transfer on the resulting measure may no longer be practically meaningful, the partial contributions for each layer are additive. This can be used to construct an attribute decomposition of layered inequality that restricts Pigou-Dalton transfers within each layer.
\begin{Definition}
	A multi-layered inequality model is a set of models $\mathcal{M} = \{\mathbf{M}_i : 1 \leq i \leq k\}$ which share the same attributes and individuals. This provides $k$ indicator values for each individual.
\end{Definition}
\begin{Definition}[Layered $f$-inequality]
	Let $\mathcal{S}= \{\mathbf{S}_i ~:~ 1 \leq i \leq k\} = \{\Gamma(\mathbf{a},\mathbf{M}_i) ~:~  \mathbf{M}_i\in\mathcal{M}\}$ be  a set of $k$ populations obtained from a set of attributes $\mathbf{a}$ and a multi-layered inequality model $\mathcal{M}$. We define a layered inequality measure $\mathcal{I}(\mathcal{S})$ as weighted sum ($w_i\geq0$) of $f$-inequality measures. Like indicated by the subscripts of $f$ and $p$ in \refEq{eq:layered-inequality}, the considered $f$-inequality can vary between layers to emphasize important characteristics of inequality on the respective layer.
	\begin{equation}
		\mathcal{I}(\{\mathbf{S}_1,...,\mathbf{S}_k\}) \definedAs \sum_{i=1}^k w_i I_{f_i,p_i}(\mathbf{S}_i)
		\label{eq:layered-inequality}
	\end{equation}
	\label{def:layered-inequality}
\end{Definition}
\begin{Definition}[Layered $f$-inequality decomposition]
	Since the partial contributions of any $f$-inequality consider the same notion of addition, their decompositions are additive as shown in \refEq{eq:layered-inequality-decomposition}. The resulting decomposition restricts Pigou-Dalton transfer to each layer and satisfies an inclusion-exclusion relation (Möbius inverse).
	\begin{subequations}
		\begin{align}
			\mathcal{I}^\cup(\alpha,\mathcal{M}) &\definedAs \sum_{i=1}^k w_i I^\cup_{f_i,p_i}(\alpha,\mathbf{M}_i)\\
			\mathcal{I}^\delta(\alpha,\mathcal{M}) &\definedAs \sum_{i=1}^k w_i I^\delta_{f_i,p_i}(\alpha,\mathbf{M}_i
		\end{align}
		\label{eq:layered-inequality-decomposition}
	\end{subequations}
	\label{def:layered-inequality-decomposition}
\end{Definition}
\begin{Remark}
	The ideal attribute dependence can vary between indicator values. Therefore, there may not exist an attribute dependence that leads to zero synergy for the layered measure.
\end{Remark}
The notion of addition for the partial contributions has to be identical on each layer to obtain an inclusion-exclusion relation on the layered measure. To maintain this consistency when transforming layered inequality, the same transformation has to be applied to all layers. 
\begin{Notation}[Transformed addition and multiplication]
	Consider a specific invertable transformation $v(\cdot)$, then we obtain the following notion of addition and multiplication on transformed inequality measures:
	\begin{subequations}
		\begin{align}
			a \oplus b &~\definedAs~ v(v^{-1}(a) + v^{-1}(b))\\
			a \odot b &~\definedAs~ v(v^{-1}(a) ~\cdot~ v^{-1}(b))
		\end{align}
	\end{subequations}
	\label{not:op-transforms}
\end{Notation}
Transforming a layered inequality measure equals transforming each $f$-inequality and their corresponding operators as shown in \refEq{eq:transform-layers} by Definition \ref{def:layered-inequality-decomposition} and Notation \ref{not:op-transforms}. As a result, the inclusion-exclusion relation (Möbius inverse) of partial contributions from the transformed measure is maintained under the transformed addition operation.
\begin{subequations}
	\begin{align}
		v(\mathcal{I}^\cup(\alpha,\mathcal{M})) &= v(w_1) \odot v(I^\cup_{f_1,p_1}(\alpha,\mathbf{M}_1)) \quad\oplus~ ... ~\oplus\quad  v(w_k) \odot v(I^\cup_{f_k,p_k}(\alpha,\mathbf{M}_k))\\
		v(\mathcal{I}^\delta(\alpha,\mathcal{M})) &= v(w_1) \odot v(I^\delta_{f_1,p_1}(\alpha,\mathbf{M}_1)) \quad\oplus~ ... ~\oplus\quad  v(w_k) \odot v(I^\delta_{f_k,p_k}(\alpha,\mathbf{M}_k))
	\end{align}
	\label{eq:transform-layers}
\end{subequations}

	\subsection{The relation of inequality and information}
	\label{sec:relation}
%
This section brings the results from decomposing inequality into the context of decomposing information. Their relation is intuitive since both information and inequality measures aim to quantify a notion of distance from the uniform distribution. We think studying their relation provides valuable insights and can encourage the transfer of results between areas. 

A \acl{PID}~\cite{williams-beer,special-issue-intro} aims to understand how a set of source variables provides information redundantly, uniquely or synergistically about a target. We use the following notation:
\begin{Notation}\hfill
	\begin{itemize}
	\item We notate a discrete visible/source variable $V$ with state $v$ in the state space $\mathcal{V}=\{v_1,...,v_m\}$.
	\item We notate a discrete target variable $T$ with state $t$ in the state space $\mathcal{T}$.
	\item We notate an indicator variable for state $t$ of variable $T$ as $\mathbf{1}_T(t)$.
	\end{itemize}
\end{Notation}
We can define discrete $f$-information as shown in \refEq{eq:info}~\cite[Section 3.1]{mages2024}. Notice that \refEq{eq:relation-ineq} is $f$-inequality as defined in Section \ref{subsec:generalized-ineq-measure} and that $f$-information is its expected value (\refEq{eq:relation-ineq-2}). Therefore, discrete $f$-information is a layered inequality measure by Definition \ref{def:layered-inequality}. We refer to $i_{f,p}(\kappa)$ as specific or target-pointwise information. The ordering of populations by the Atkinson criterion is identical to the ordering of random variables by the Blackwell order with respect to an indicator variable~\cite{blackwell1953equivalent,blackwell-no-lattice}, which is the binary target $\mathbf{1}_T(t)$. In this context, the zonogon $Z_\kappa(P(V\mid \mathbf{1}_T(t)))$ represents the achievable trade-off between the type-I and type-II error for predicting the state $t\in\mathcal{T}$~\cite{mages2024} and its lower boundary is known as Neyman-Pearson boundary~\cite[Section 14.1]{polyanskiy2022information}. The transformation of measures is also used in both areas: just like the Atkinson index is an invertible transformation on an $f$-inequality, so is Rényi-information an invertible transformation of an $f$-information~\cite{mages2024}. Due to these relations, the presented methodology in this work can directly be applied to obtain non-negative \aclp{PID} with practical operational interpretation, as shown in~\cite{mages2024}.
\begin{subequations}
	\begin{align}
		P(V\mid \mathbf{1}_T(t)) &\definedAs \begin{bmatrix}
			p(V=v_1\mid T=t) & \dots & p(V=v_m\mid T=t)\\
			p(V=v_1\mid T\neq t) & \dots & p(V=v_m\mid T\neq t)\\
		\end{bmatrix}\\
		i_{f,p}(\kappa) &\definedAs \sum_{\vec{v}\in\kappa}r_{f,p}(\vec{v})\label{eq:relation-ineq}\\
		\text{discrete } f\text{-information of }(V;T) &\definedAs \mathbb{E}_{t\in T}[i_{f,p(T=t)}(P(V\mid \mathbf{1}_T(t)))]\label{eq:relation-ineq-2}
	\end{align}
	\label{eq:info}
\end{subequations}

This creates a relation between some commonly used information and inequality measures, as shown in Table \ref{tbl:relation-measures}. It may be desirable to survey existing inequality measures in the future to see if they are (invertible transformations of) an $f$-inequality and identify the equivalent (transformation of an)  $f$-information.
\begin{table}[h]\centering
	\begin{tabular}{r c l}\toprule
		\textbf{$f$-information measure} & \textbf{Generator}& \textbf{$f$-inequality measure}\\\midrule 
		Total Variation & $f(t) = 0.5|t-1|$ & Pietra index\\
		Reversed \acl{KL}-Information & $f(t) = -\ln\left(t\right)$  & Theil index/Generalized Entropy ($c=1$)\\
		Mutual Information & $f(t) =  t\ln\left(t\right)$&  Generalized Entropy ($c=0$)\\		\bottomrule
	\end{tabular}
	\caption{Relation of commonly used inequality and information measure.}
	\label{tbl:relation-measures}
\end{table}

Some key relations between both areas are summarized in Table \ref{tbl:relation}. Consequently, we see further opportunities to apply concepts and insights from one area to the other. In particular, we are curious about the resulting interpretation when applying subgroup decompositions from inequality measures to specific information.
\begin{table}[h]\centering
	\begin{tabular}{r c l}\toprule
		\textbf{Information Measures} && \textbf{Inequality Measures}\\\midrule  
		Blackwell order with respect to an indicator variable &$\Longleftrightarrow$& Atkinson criterion\\
		Neyman-Pearson boundary &$\Longleftrightarrow$& Lorenz curve\\
		specific $f$-information &$\Longleftrightarrow$& $f$-inequality\\\bottomrule
	\end{tabular}
	\caption{Equivalences between measures of information and inequality.}
	\label{tbl:relation}
\end{table}

	\section{Discussion}
%
This work transferred the results from decomposing information measures~\cite{mages2024} to the decomposition of inequality measures. We defined a class of inequality measures and demonstrated a framework for attribute decompositions based on current research on \aclp{PID}.
We studied a family of inequality measures that can be described as length-like on the Lorenz curve since they satisfy a triangle inequality. These measures are particularly interesting due to their properties and relation to established information measures. 

We could similarly construct an attribute decomposition for area-like measures, such as the Gini coefficient: As it can be seen from \refEq{eq:basic-in}, a non-negative attribute decomposition can be achieved by using the meet of the Atkinson order as a notion of intersection on the redundancy lattice~\cite{williams-beer}. While this approach does not provide the same operational interpretation of synergy, it may provide a beneficial interpretation of redundancy for some applications.

When decomposing inequality, it is sometimes desired to identify the flow of partial contributions through Pigou-Dalton transfers over time~\cite{centralbank2019}. While it has not been discussed in this work, this can be achieved using the method described in \cite[Section 4.2]{mages2024}.

Typical decompositions currently assume categorical attributes for forming clear partitions and subgroups. It appears to be an open research question of extending these ideas to attributes with a notion of similarity (distance) between states. For example, a person's age in years is discrete but not categorical, which leads to a more fuzzy definition of subgroups. We think it would be desirable to better understand the treatment of such variables in both inequality and information decompositions.

We noted in Section \ref{subsec:transforming-ineq-measure} that it would be desirable to utilize different notions of addition when computing Shapley values from inequality measures to ensure the consistency of results between related measures. This highlights the difficulty of transferring concepts between areas. However, we are optimistic that such issues can be avoided between inequality and information measures since they share an identical underlying representation and ordering relation.

Finally, we used the initial examples in Section \ref{sec:intro} (energy/communication and privacy/utility) to contrast inequality measures with respect to the contribution and distribution of resources. This indicates how the presented inequality measures can be extended to attribute decomposable fairness measures as future work.

	\section{Conclusions}
	In this work, we presented a new family of inequality measures and a new type of inequality decomposition. The presented decomposition focuses on the interactions between attributes of an individual to identify how inequality is obtained from the redundant, unique, and synergetic interactions between them. We demonstrated that the analysis by game synergy and Shapley values cannot separate the desired components and that the decomposition requires an extension of the inequality measure. We defined an extension for the introduced family of inequality measures, which satisfies the required properties and provides a practical operational interpretation. This generates a decomposition for established measures, such as the Generalized Entropy and Atkinson index. Finally, we discussed the relation between measures of information and inequality to encourage the transfer of results between both areas. 
	\vspace{3mm}
	
	{\noindent\small \textbf{Acknowledgments:} We thank Miia Bask for the helpful discussion and suggestions.}\vspace{3mm}
	
	{\noindent\small \textbf{Data Availability Statement:} An implementation of the presented decomposition is available at:\newline \url{https://github.com/uu-core/pid-inequality}}\vspace{3mm}
	
	{\noindent\small \textbf{Funding:} This research was funded by Swedish Civil Contingencies Agency (MSB) through the project RIOT grant number MSB 2018-12526. The funders had no role in study design, data collection and analysis, decision to publish, or preparation of the manuscript.}
	
	\bibliography{bibliography.bib}

\begin{thebibliography}{42}
\providecommand{\natexlab}[1]{#1}
\providecommand{\url}[1]{\texttt{#1}}
\expandafter\ifx\csname urlstyle\endcsname\relax
  \providecommand{\doi}[1]{doi: #1}\else
  \providecommand{\doi}{doi: \begingroup \urlstyle{rm}\Url}\fi

\bibitem[Theil(1967)]{theilBook}
Henri Theil.
\newblock \emph{Economics and information theory}.
\newblock Studies in mathematical and managerial economics, 7. North-Holland
  Publishing Company, Amsterdam, 1967.

\bibitem[Shorrocks(1980)]{shorrocks1980}
Anthony~F Shorrocks.
\newblock The class of additively decomposable inequality measures.
\newblock \emph{Econometrica: Journal of the Econometric Society}, pages
  613--625, 1980.

\bibitem[Williams and Beer(2010)]{williams-beer}
Paul~L. Williams and Randall~D. Beer.
\newblock Nonnegative decomposition of multivariate information.
\newblock \href{https://arxiv.org/abs/1004.2515}{arXiv 1004.2515}, 2010.

\bibitem[Mages et~al.(2024)Mages, Anastasiadi, and Rohner]{mages2024}
Tobias Mages, Elli Anastasiadi, and Christian Rohner.
\newblock Non-negative decomposition of multivariate information: From minimum
  to blackwell-specific information.
\newblock \emph{Entropy}, 26\penalty0 (5), 2024.
\newblock ISSN 1099-4300.
\newblock \doi{10.3390/e26050424}.

\bibitem[Dalton(1920)]{dalton1920}
Hugh Dalton.
\newblock The measurement of the inequality of incomes.
\newblock \emph{The Economic Journal}, 30\penalty0 (119):\penalty0 348--361,
  1920.

\bibitem[Atkinson(1970)]{atkinson1970}
Anthony~B Atkinson.
\newblock On the measurement of inequality.
\newblock \emph{Journal of economic theory}, 2\penalty0 (3):\penalty0 244--263,
  1970.

\bibitem[Lerman and Yitzhaki(1985)]{lerman1985}
Robert~I Lerman and Shlomo Yitzhaki.
\newblock Income inequality effects by income source: A new approach and
  applications to the united states.
\newblock \emph{The review of economics and statistics}, pages 151--156, 1985.

\bibitem[Paul(2004)]{paul2004}
Satya Paul.
\newblock Income sources effects on inequality.
\newblock \emph{Journal of Development Economics}, 73\penalty0 (1):\penalty0
  435--451, 2004.

\bibitem[Costa and P{\'e}rez-Duarte(2019)]{centralbank2019}
Rita~Neves Costa and S{\'e}bastien P{\'e}rez-Duarte.
\newblock \emph{Not all inequality measures were created equal: The measurement
  of wealth inequality, its decompositions, and an application to European
  household wealth}.
\newblock Number~31 in Statistics Paper Series. ECB Statistics Paper, 2019.

\bibitem[Bhattacharya and Mahalanobis(1967)]{bhattacharya1967}
Nath Bhattacharya and Bimalendu Mahalanobis.
\newblock Regional disparities in household consumption in india.
\newblock \emph{Journal of the American Statistical Association}, 62\penalty0
  (317):\penalty0 143--161, 1967.

\bibitem[Bourguignon(1979)]{bourguignon1979}
Francois Bourguignon.
\newblock Decomposable income inequality measures.
\newblock \emph{Econometrica: Journal of the Econometric Society}, pages
  901--920, 1979.

\bibitem[Shorrocks(1984)]{shorrocks1984}
Anthony~F Shorrocks.
\newblock Inequality decomposition by population subgroups.
\newblock \emph{Econometrica: Journal of the Econometric Society}, pages
  1369--1385, 1984.

\bibitem[Dagum(1998)]{dagum1998}
Camilo Dagum.
\newblock \emph{A new approach to the decomposition of the Gini income
  inequality ratio}.
\newblock Springer, 1998.

\bibitem[Basmann et~al.(1990)Basmann, Hayes, Slottje, and
  Johnson]{estimateLorenz1990}
R.L Basmann, K.J Hayes, D.J Slottje, and J.D Johnson.
\newblock A general functional form for approximating the lorenz curve.
\newblock \emph{Journal of Econometrics}, 43\penalty0 (1):\penalty0 77--90,
  1990.
\newblock ISSN 0304-4076.
\newblock \doi{https://doi.org/10.1016/0304-4076(90)90108-6}.

\bibitem[Chotikapanich(1993)]{estimateLorenz1993}
Duangkamon Chotikapanich.
\newblock A comparison of alternative functional forms for the lorenz curve.
\newblock \emph{Economics Letters}, 41\penalty0 (2):\penalty0 129--138, 1993.
\newblock ISSN 0165-1765.
\newblock \doi{https://doi.org/10.1016/0165-1765(93)90186-G}.

\bibitem[Sarabia et~al.(1999)Sarabia, Castillo, and
  Slottje]{estimateLorenz1999}
J.-M. Sarabia, Enrique Castillo, and Daniel~J. Slottje.
\newblock An ordered family of lorenz curves.
\newblock \emph{Journal of Econometrics}, 91\penalty0 (1):\penalty0 43--60,
  1999.
\newblock ISSN 0304-4076.
\newblock \doi{https://doi.org/10.1016/S0304-4076(98)00048-7}.

\bibitem[Sitthiyot and Holasut(2021)]{estimateLorenz2021}
Thitithep Sitthiyot and Kanyarat Holasut.
\newblock A simple method for estimating the lorenz curve.
\newblock \emph{Humanities and Social Sciences Communications}, 8\penalty0
  (1):\penalty0 1--9, 2021.

\bibitem[Allison(1978)]{allison1978}
Paul~D Allison.
\newblock Measures of inequality.
\newblock \emph{American sociological review}, pages 865--880, 1978.

\bibitem[Pigou(1912)]{pigou1912}
Arthur~Cecil Pigou.
\newblock \emph{Wealth and welfare}.
\newblock Macmillan and Company, limited, 1912.

\bibitem[Gini(1912)]{gini1912}
Corrado Gini.
\newblock \emph{Variabilit{\`a} e mutabilit{\`a}: contributo allo studio delle
  distribuzioni e delle relazioni statistiche.[Fasc. I.]}.
\newblock Tipogr. di P. Cuppini, 1912.

\bibitem[Pietra(1915)]{pietra1915}
Gaetano Pietra.
\newblock \emph{Delle relazioni tra gli indici di variabilit{\=a}}.
\newblock C. Ferrari, 1915.

\bibitem[Hao and Naiman(2010)]{inequalityBook}
Lingxin Hao and Daniel~Q Naiman.
\newblock \emph{Assessing inequality}.
\newblock Sage Publications, 2010.

\bibitem[Koshevoy and Mosler(1996)]{Mosler1996}
Gleb Koshevoy and Karl Mosler.
\newblock The lorenz zonoid of a multivariate distribution.
\newblock \emph{Journal of the American Statistical Association}, 91\penalty0
  (434):\penalty0 873--882, 1996.
\newblock \doi{10.1080/01621459.1996.10476955}.

\bibitem[Koshevoy and Mosler(2007)]{Mosler2007}
Gleb~A Koshevoy and Karl Mosler.
\newblock Multivariate lorenz dominance based on zonoids.
\newblock \emph{AStA Advances in Statistical Analysis}, 91:\penalty0 57--76,
  2007.

\bibitem[Lorenz(1905)]{lorenz1905}
Max~O Lorenz.
\newblock Methods of measuring the concentration of wealth.
\newblock \emph{Publications of the American statistical association},
  9\penalty0 (70):\penalty0 209--219, 1905.

\bibitem[Polyanskiy and Wu(2023)]{polyanskiy2022information}
Yury Polyanskiy and Yihong Wu.
\newblock Information theory: From coding to learning.
\newblock \emph{Book draft}, Nov 2023.

\bibitem[Bertschinger and Rauh(2014)]{blackwell-no-lattice}
Nils Bertschinger and Johannes Rauh.
\newblock The blackwell relation defines no lattice.
\newblock In \emph{2014 IEEE International Symposium on Information Theory},
  pages 2479--2483, 2014.
\newblock \doi{10.1109/ISIT.2014.6875280}.

\bibitem[Csisz{\'a}r(1967)]{csiszar1967}
Imre Csisz{\'a}r.
\newblock On information-type measure of difference of probability
  distributions and indirect observations.
\newblock \emph{Studia Sci. Math. Hungar.}, 2:\penalty0 299--318, 1967.

\bibitem[Finn and Lizier(2018)]{finn-ppid-2}
Conor Finn and Joseph~T. Lizier.
\newblock Pointwise partial information decomposition using the specificity and
  ambiguity lattices.
\newblock \emph{Entropy}, 20\penalty0 (4), 2018.
\newblock ISSN 1099-4300.
\newblock \doi{10.3390/e20040297}.

\bibitem[Deutsch and Silber(2008)]{deutsch2008shapley}
Joseph Deutsch and Jacques Silber.
\newblock On the shapley value and the decomposition of inequality by
  population subgroups with special emphasis on the gini index.
\newblock In \emph{Advances on income inequality and concentration measures},
  pages 183--200. Routledge, 2008.

\bibitem[Shapley(1951)]{shapley1951notes}
Lloyd~S Shapley.
\newblock Notes on the n-person game—ii: The value of an n-person game.
\newblock 1951.

\bibitem[Grabisch(1997)]{shapleySynergy}
Michel Grabisch.
\newblock k-order additive discrete fuzzy measures and their representation.
\newblock \emph{Fuzzy Sets and Systems}, 92\penalty0 (2):\penalty0 167--189,
  1997.
\newblock ISSN 0165-0114.
\newblock \doi{https://doi.org/10.1016/S0165-0114(97)00168-1}.
\newblock Fuzzy Measures and Integrals.

\bibitem[Rosas et~al.(2020)Rosas, Mediano, Rassouli, and Barrett]{Rosas_2020}
Fernando~E Rosas, Pedro A~M Mediano, Borzoo Rassouli, and Adam~B Barrett.
\newblock An operational information decomposition via synergistic disclosure.
\newblock \emph{Journal of Physics A: Mathematical and Theoretical},
  53\penalty0 (48):\penalty0 485001, nov 2020.
\newblock \doi{10.1088/1751-8121/abb723}.

\bibitem[Gutknecht et~al.(2023)Gutknecht, Makkeh, and
  Wibral]{gutknecht2023babel}
Aaron~J Gutknecht, Abdullah Makkeh, and Michael Wibral.
\newblock From babel to boole: The logical organization of information
  decompositions.
\newblock \emph{arXiv preprint arXiv:2306.00734}, 2023.

\bibitem[Kolchinsky(2022)]{Kolchinsky}
Artemy Kolchinsky.
\newblock A novel approach to the partial information decomposition.
\newblock \emph{Entropy}, 24\penalty0 (3), 2022.
\newblock ISSN 1099-4300.
\newblock \doi{10.3390/e24030403}.

\bibitem[Gomes and Figueiredo(2024)]{Gomes}
André F.~C. Gomes and Mário A.~T. Figueiredo.
\newblock A measure of synergy based on union information.
\newblock \emph{Entropy}, 26\penalty0 (3), 2024.
\newblock ISSN 1099-4300.
\newblock \doi{10.3390/e26030271}.

\bibitem[Rota(1964)]{rota1964foundations}
Gian-Carlo Rota.
\newblock On the foundations of combinatorial theory: I. theory of m{\"o}bius
  functions.
\newblock In \emph{Classic Papers in Combinatorics}, pages 332--360. Springer,
  1964.

\bibitem[Chicharro and Panzeri(2017)]{dualdecomposition}
Daniel Chicharro and Stefano Panzeri.
\newblock Synergy and redundancy in dual decompositions of mutual information
  gain and information loss.
\newblock \emph{Entropy}, 19\penalty0 (2), 2017.
\newblock ISSN 1099-4300.
\newblock \doi{10.3390/e19020071}.

\bibitem[Knuth(2019)]{Knuth2019}
Kevin~H. Knuth.
\newblock Lattices and their consistent quantification.
\newblock \emph{Annalen der Physik}, 531\penalty0 (3):\penalty0 1700370, 2019.

\bibitem[Branke(2008)]{branke2008multiobjective}
J{\"u}rgen Branke.
\newblock \emph{Multiobjective optimization: Interactive and evolutionary
  approaches}, volume 5252.
\newblock Springer Science \& Business Media, 2008.

\bibitem[Lizier et~al.(2018)Lizier, Bertschinger, Jost, and
  Wibral]{special-issue-intro}
Joseph~T. Lizier, Nils Bertschinger, Jürgen Jost, and Michael Wibral.
\newblock Information decomposition of target effects from multi-source
  interactions: Perspectives on previous, current and future work.
\newblock \emph{Entropy}, 20\penalty0 (4), 2018.
\newblock ISSN 1099-4300.
\newblock \doi{10.3390/e20040307}.

\bibitem[Blackwell(1953)]{blackwell1953equivalent}
David Blackwell.
\newblock Equivalent comparisons of experiments.
\newblock \emph{The annals of mathematical statistics}, pages 265--272, 1953.

\end{thebibliography}

	\appendix
	\renewcommand*{\thesection}{\Alph{section}}

	\section[\appendixname~\thesection]{Relation of Property \ref{prop:ineq-1}-\ref{prop:ineq-5} to the zonogon order}
	\label{ap:prop-to-lorenz}
%
\subsection{Representation of Property \ref{prop:ineq-1}-\ref{prop:ineq-5}}
\label{apsub:prop-representation}	
\begin{enumerate}
	\item \textbf{Label invariance} (Property \ref{prop:ineq-1}): Re-labeling individuals and groups is the  re-ordering of columns in the population matrix. This is a column permutation and, therefore, equivalent to the multiplication with a permutation matrix $\mathbf{P}$, as shown in \refEq{eq:relabling}.
	\begin{equation}
		\kappa(\mathbf{S}_2) = \kappa(\mathbf{S}_1) \mathbf{P}
		\label{eq:relabling}
	\end{equation}
	\begin{Lemma}
		If population $\mathbf{S}_2$ is a relabeling of $\mathbf{S}_1$, then both populations are in the same equivalence class $\langle\mathbf{S}_1\rangle = \langle\mathbf{S}_2\rangle$.
		\label{lem:label-inv-perm}
	\end{Lemma}
	\begin{proof}
		The inverse of a permutation matrix is its transpose ($\mathbf{P}^{-1} = \mathbf{P}^T$). Since permutation matrices are double stochastic, $\mathbf{P}$ and $\mathbf{P}^T$ are row stochastic. For any population re-labeling:
		\begin{equation}
			\begin{aligned}
				\kappa(\mathbf{S}_2) &= \kappa(\mathbf{S}_1) \mathbf{P} \\
				\kappa(\mathbf{S}_1) &= \kappa(\mathbf{S}_2) \mathbf{P}^T
			\end{aligned}
		\end{equation}
		Thus, we obtain:
		\begin{equation*}
			 \begin{aligned}
			 	 Z_\kappa(\mathbf{S}_1) &= Z_\kappa(\mathbf{S}_2)\qquad& \textnormal{(by \refEq{eq:zonogon-order})}\\
			 	 \mathbf{S}_1 &\cong \mathbf{S}_2 & \textnormal{(by Definition \ref{def:pop-equivalence})}\\
			 	 \langle\mathbf{S}_1\rangle &= \langle\mathbf{S}_2\rangle & \textnormal{(by Notation \ref{not:pop-classes})}
			 \end{aligned}
		\end{equation*}
	\end{proof}
	\item \textbf{Duplication invariance} (Property \ref{prop:ineq-2}): Duplicating a population is equivalent to duplicating and normalizing the corresponding population matrix, as shown in \refEq{eq:dup-conc}.
	\begin{equation}
		Z\left(\kappa(\mathbf{S}\uplus\mathbf{S})\right) = Z\left(0.5 \begin{bmatrix}
			\kappa(\mathbf{S}) ~&~\kappa(\mathbf{S})
		\end{bmatrix}\right)
		\label{eq:dup-conc}
	\end{equation}
	\begin{Lemma}
		Duplicating a population does not affect its equality class: $\langle\mathbf{S}\rangle = \langle\mathbf{S}\uplus\mathbf{S}\rangle$.
		\label{lem:size-inv-perm}
	\end{Lemma}
	\begin{proof}
		Let $\textbf{P}_1$ and $\textbf{P}_2$ be permutation matrices to generate the desired ordering of columns for the relation of \refEq{eq:prop-2-rel-ordering}. By Lemma \ref{lem:label-inv-perm}, this does not affect their equivalence class. 
		\begin{equation}
			\kappa(\mathbf{S} \uplus \mathbf{S})\mathbf{P}_2 = 0.5 \begin{bmatrix}
				\kappa(\mathbf{S})\mathbf{P}_1 ~&~\kappa(\mathbf{S})\mathbf{P}_1\label{eq:prop-2-rel-ordering}
			\end{bmatrix}
		\end{equation} 
		Let $\mathbf{I}$ be an $|\mathbf{S}|\times|\mathbf{S}|$ identity matrix, then we can find that both population matrices can be represented as multiplication of the other by a row stochastic matrix as shown in \refEq{eq:prop-2-relatinon}. As discussed in Lemma \ref{lem:label-inv-perm}, this implies $\langle\mathbf{S}\rangle = \langle\mathbf{S}\uplus\mathbf{S}\rangle$ by \refEq{eq:zonogon-order}, Definition \ref{def:pop-equivalence}, and Notation \ref{not:pop-classes}.
		\begin{subequations}
			\begin{align}
				\kappa(\mathbf{S} \uplus \mathbf{S})\mathbf{P}_2 &=  \kappa(\mathbf{S}) \mathbf{P}_1 \cdot0.5\cdot \begin{bmatrix} \mathbf{I} & \mathbf{I}
				\end{bmatrix} \label{eq:prop-2-relatinon-1} \\
				\kappa(\mathbf{S})\mathbf{P}_1 &= \kappa(\mathbf{S} \uplus \mathbf{S}) \mathbf{P}_2 \begin{bmatrix} \mathbf{I} \\ \mathbf{I}
				\end{bmatrix}\label{eq:prop-2-relatinon-2}
			\end{align}
			\label{eq:prop-2-relatinon}
		\end{subequations}
	\end{proof}
	\item \textbf{Scale invariance} (Property \ref{prop:ineq-3}):
	\begin{Lemma}
		If population $\mathbf{S}_1 = \{k\cdot s_i ~:~ s_i \in \mathbf{S}_2\}$ is a linear scaling for the indicator value of $\mathbf{S}_2$ by $k\in \mathbb{R}_{>0}$, then both populations are in the same equivalence class: $\langle\mathbf{S}_1\rangle= \langle\mathbf{S}_2\rangle$.
		\label{lem:prop-3-helper}
	\end{Lemma}
	\begin{proof}
		\begin{subequations}
			\begin{align}
				\kappa(\mathbf{S}_1) &= \frac{1}{|\mathbf{S}_1|} \begin{bmatrix}~
					\begin{matrix}
						1\\ \sfrac{s}{\overline{\mathbf{S}_1}}
					\end{matrix} &:~ s \in \mathbf{S}_1
				\end{bmatrix} & \text{(by Definition \ref{def:population-to-matrix})}\\
				&= \frac{1}{|\mathbf{S}_2|} \begin{bmatrix}~
					\begin{matrix}
						1\\ \frac{ks}{k\overline{\mathbf{S}_2}}
					\end{matrix} &:~ s \in \mathbf{S}_2
				\end{bmatrix} & \text{(by assumption)}\\
				&= \frac{1}{|\mathbf{S}_2|} \begin{bmatrix}~
				\begin{matrix}
					1\\ \sfrac{s}{\overline{\mathbf{S}_2}}
				\end{matrix} &:~ s \in \mathbf{S}_2
				\end{bmatrix} & \text{(by $k\neq 0$)}\\
				&= \kappa(\mathbf{S}_2)& \text{(by Definition \ref{def:population-to-matrix})}\\
				\Longrightarrow \langle\mathbf{S}_1\rangle &= \langle\mathbf{S}_2\rangle & \text{(by Lemma \ref{lem:label-inv-perm})}
			\end{align}
		\end{subequations}
	\end{proof}
	\item \textbf{Pigou-Dalton transfers} (Property \ref{prop:ineq-4}):
	Let $\mathbf{S}'$ be the population after a Pigou-Dalton transfer on $\mathbf{S} = \mathbf{G}\uplus\{s_1,s_2\}$ and choose the permutation matrices $\mathbf{P}_1$ and $\mathbf{P}_2$ according to \refEq{eq:pdtransfer-matrix-0}.  Note that a Pigou-Dalton transfer does not affect the total or average indicator value $\overline{\mathbf{S}}=\overline{\mathbf{S}}'$.
	\begin{subequations}
		\begin{align}
			\kappa(\mathbf{S})\ \mathbf{P}_1 &= \frac{1}{|\mathbf{S}|} \begin{bmatrix}\begin{bmatrix}
					\begin{matrix}
						1 \\ \sfrac{s}{\overline{\mathbf{S}}}
					\end{matrix}&:~s\in\mathbf{G}
				\end{bmatrix} & \begin{matrix}
				1 & 1\\ \sfrac{s_1}{\overline{\mathbf{S}}} & \sfrac{s_2}{\overline{\mathbf{S}}}
				\end{matrix}
			\end{bmatrix}\label{eq:pdtransfer-matrix-1}	\\
			\kappa(\mathbf{S}')\ \mathbf{P}_2 &= \frac{1}{|\mathbf{S}|} \begin{bmatrix}\begin{bmatrix}
				\begin{matrix}
					1 \\ \sfrac{s}{\overline{\mathbf{S}}}
				\end{matrix}&:~s\in\mathbf{G}
			\end{bmatrix} & \begin{matrix}
				1 & 1\\ \sfrac{s'_1}{\overline{\mathbf{S}}} & \sfrac{s'_2}{\overline{\mathbf{S}}}
			\end{matrix}
			\end{bmatrix}\label{eq:pdtransfer-matrix-2}
		\end{align}
		\label{eq:pdtransfer-matrix-0}
	\end{subequations}
	We can represent a Pigou-Dalton transfer with $p\in(0,0.5]$ as multiplication by a double stochastic matrix, as shown in \refEq{eq:pdtransfer-matrix}, where $\mathbf{I}$ is an identity matrix.
	\begin{subequations}
		\begin{align}
			\begin{bmatrix}\begin{bmatrix}
					\begin{matrix}
						1 \\ \sfrac{s}{\overline{\mathbf{S}}}
					\end{matrix}&:~s\in\mathbf{G}
				\end{bmatrix} & \begin{matrix}
					1 & 1\\ \sfrac{s'_1}{\overline{\mathbf{S}}} & \sfrac{s'_2}{\overline{\mathbf{S}}}
				\end{matrix}
			\end{bmatrix}&= \begin{bmatrix}\begin{bmatrix}
				\begin{matrix}
					1 \\ \sfrac{s}{\overline{\mathbf{S}}}
				\end{matrix}&:~s\in\mathbf{G}
			\end{bmatrix} & \begin{matrix}
				1 & 1\\ \sfrac{s_1}{\overline{\mathbf{S}}} & \sfrac{s_2}{\overline{\mathbf{S}}}
			\end{matrix}
			\end{bmatrix} \begin{bmatrix}
			\mathbf{I} & 0 & 0\\
			0 & (1-q) & q \\
			0 & q & (1-q) \\
			\end{bmatrix}\\
			\kappa(\mathbf{S}')\ \mathbf{P}_2 &= \kappa(\mathbf{S})\ \mathbf{P}_1 \lambda \qquad \text{(where $\lambda$ is a double stochastic matrix)}
		\end{align}
		\label{eq:pdtransfer-matrix}
	\end{subequations}
	Since stochastic matrices are closed under multiplication, any sequence of Pigou-Dalton transfers corresponds to a multiplication by some stochastic matrix $\lambda$.
	\begin{Lemma}
		If there exist a sequence of (non-empty) Pigou-Dalton transfers on population $\mathbf{S}$ to arrive at population $\mathbf{S}'$, then $\langle\mathbf{S}'\rangle \sqsubset \langle\mathbf{S}\rangle$.
		\label{lem:transfer-inferior}
	\end{Lemma}
	\begin{proof}
		We obtain $Z_\kappa(\mathbf{S}') \subseteq Z_\kappa(\mathbf{S})$ from \refEq{eq:pdtransfer-matrix} and \refEq{eq:product-subset}. We obtain $Z_\kappa(\mathbf{S}) \not\subseteq Z_\kappa(\mathbf{S}')$ since the inverse of the transfer matrix $\lambda$ is not a valid stochastic matrix for $q\in(0,0.5]$. 
		It follows from Notation \ref{not:pop-classes} and Definition \ref{def:pop-class-lattice} that $\langle\mathbf{S}'\rangle \sqsubset \langle\mathbf{S}\rangle$.
	\end{proof}
	\item \textbf{Bottom element} (Property \ref{prop:ineq-5}):
	\begin{Lemma}
		The equivalence class of the bottom element $\langle\bot_{\mathbf{S}}\rangle$ contains all uniform distributions.
		\label{lem:prop-5-helper}
	\end{Lemma}
	\begin{proof}
		The equivalence class of the bottom element is a predecessor for all other populations, as shown in \refEq{eq:bot-inerior-all}:
		\begin{subequations}
			\begin{align}
				\kappa(\bot_{\mathbf{S}}) = \begin{bmatrix}
					1 \\ 1
				\end{bmatrix}&= \kappa(\mathbf{S}) \begin{bmatrix} 1 \\ \vdots \\ 1 \end{bmatrix}& \text{(by Definition \ref{def:pop-class-lattice})} \\
				\Longrightarrow \langle\bot_{\mathbf{S}}\rangle &\sqsubseteq  \langle\mathbf{S}\rangle & \text{(by \refEq{eq:zonogon-order} and Definition \ref{def:pop-class-lattice})} 
			\end{align}
			\label{eq:bot-inerior-all}
		\end{subequations}
		Let $\mathbf{S}$ be an arbitrary uniform distribution, then its equivalence class is also a predecessor to the bottom element, as shown in \refEq{eq:bot-inerior-all-2}.
		\begin{subequations}
			\begin{align}
				\kappa(\mathbf{S}) &= \begin{bmatrix} 1 \\ 1 \end{bmatrix} \begin{bmatrix}
					\sfrac{1}{|\mathbf{S}|} & \dots & \sfrac{1}{|\mathbf{S}|}
				\end{bmatrix}& \text{(by $\mathbf{S}$ being uniform)}\\
				&= \kappa(\bot_{\mathbf{S}}) \begin{bmatrix}
					\sfrac{1}{|\mathbf{S}|} & \dots & \sfrac{1}{|\mathbf{S}|}
				\end{bmatrix}& \text{(by Definition \ref{def:pop-class-lattice})} \\
				\Longrightarrow \langle\mathbf{S}\rangle &\sqsubseteq \langle\bot_{\mathbf{S}}\rangle & \text{(by \refEq{eq:zonogon-order} and Definition \ref{def:pop-class-lattice})} 
			\end{align}
			\label{eq:bot-inerior-all-2}
		\end{subequations}
		\refEq{eq:bot-inerior-all} and \refEq{eq:bot-inerior-all-2} imply $\langle\bot_{\mathbf{S}}\rangle = \langle{\mathbf{S}}\rangle$.
	\end{proof}
\end{enumerate}

\subsection{Proofs for Section \ref{subsubsec:order-prop-relation}}
\label{subap:proofs-order-prop}

\textbf{Lemma \ref{thm:weak-star-implies-1-5}}: \textit{Satisfying Property \ref{prop:weak-propertyStar} implies that the inequality measure satisfies the weak Property \ref{prop:ineq-1}-\ref{prop:ineq-5}.}
\begin{proof}\hfil
	\begin{itemize}
		\item \textbf{Property \ref{prop:ineq-1}}: Lemma \ref{lem:label-inv-perm} states that relabeling individuals/groups does not affect the equivalence class, and \refEq{eq:mainain-weak-zonogon-order-1} ensures that all populations within an equivalence class obtain the same inequality index. This ensures Property \ref{prop:ineq-1}.
		\item \textbf{Property \ref{prop:ineq-2}}: Lemma \ref{lem:size-inv-perm} states that population duplication does not affect the equivalence class, and \refEq{eq:mainain-weak-zonogon-order-1} ensures that all populations within an equivalence class obtain the same inequality index. This ensures Property \ref{prop:ineq-2}.
		\item \textbf{Property \ref{prop:ineq-3}}: Lemma \ref{lem:prop-3-helper} states that scaling the indicator variable does not affect the equivalence class, and \refEq{eq:mainain-weak-zonogon-order-1} ensures that all populations within an equivalence class obtain the same inequality index. This ensures Property \ref{prop:ineq-3}.
		\item \textbf{weak Property \ref{prop:ineq-4}}: Lemma \ref{lem:transfer-inferior} states that the equivalence class after a Pigou-Dalton transfer is a predecessor of the original population. \refEq{eq:mainain-weak-zonogon-order-1} ensures that predecessors obtain an inequality index that is less or equal. This ensures the \textit{weak} Property \ref{prop:ineq-4}.
		\item \textbf{Property \ref{prop:ineq-5}}: Lemma \ref{lem:prop-5-helper} states that all uniform distributions are in the equivalence class of the bottom element, and \refEq{eq:mainain-weak-zonogon-order} ensures that all populations within this equivalence class obtain the inequality index zero. The non-negativity is then obtained from \refEq{eq:mainain-weak-zonogon-order-1} since all populations are successors or equivalent to the bottom element. This ensures Property \ref{prop:ineq-5}. 
	\end{itemize}
\end{proof}

\textbf{Lemma \ref{thm:strict-star-implies-1-5}}: \textit{Satisfying Property \ref{prop:strict-propertyStar} implies that the inequality measure satisfies the strict Property \ref{prop:ineq-1}-\ref{prop:ineq-5}.}
\begin{proof}\hfil
	\begin{itemize}
		\item \textbf{Property \ref{prop:ineq-1}}: Lemma \ref{lem:label-inv-perm} states that relabeling individuals/groups does not affect the equivalence class, and \refEq{eq:mainain-strict-zonogon-order-1} ensures that all populations within an equivalence class obtain the same inequality index. This ensures Property \ref{prop:ineq-1}.
		\item \textbf{Property \ref{prop:ineq-2}}: Lemma \ref{lem:size-inv-perm} states that population duplication does not affect the equivalence class, and \refEq{eq:mainain-strict-zonogon-order-1} ensures that all populations within an equivalence class obtain the same inequality index. This ensures Property \ref{prop:ineq-2}.
		\item \textbf{Property \ref{prop:ineq-3}}: Lemma \ref{lem:prop-3-helper} states that scaling the indicator variable does not affect the equivalence class, and \refEq{eq:mainain-strict-zonogon-order-1} ensures that all populations within an equivalence class obtain the same inequality index. This ensures Property \ref{prop:ineq-3}.
		\item \textbf{strict Property \ref{prop:ineq-4}}: Lemma \ref{lem:transfer-inferior} states that the equivalence class after a Pigou-Dalton transfer is a predecessor of the original population. \refEq{eq:mainain-strict-zonogon-order-2} ensures that predecessors obtain a smaller inequality index. This ensures the \textit{strict} Property \ref{prop:ineq-4}.
		\item \textbf{Property \ref{prop:ineq-5}}: Lemma \ref{lem:prop-5-helper} states that all uniform distributions are in the equivalence class of the bottom element, and \refEq{eq:mainain-strict-zonogon-order} ensures that all populations within this equivalence class obtain the inequality index zero. The non-negativity is then obtained from \refEq{eq:mainain-strict-zonogon-order-2} since all populations are successors or equivalent to the bottom element. This ensures Property \ref{prop:ineq-5}. 
	\end{itemize}
\end{proof}

	\section[\appendixname~\thesection]{Properties and special cases of f-inequality}
	\label{ap:special-cases}
%
\subsection{Properties of f-inequality}
\label{subap:f-ineq-properties}

\textbf{Theorem \ref{thm:rf-props}}: \textit{
For a constant $0 \leq p \leq 1$:
\begin{enumerate}
	\item the function $r_{f,p}(\vec{v})$:
	\begin{enumerate}
		\item \makebox[7.2cm][l]{quantifies any vector of slope one to zero:} $r_{f,p}\left(\left[\begin{smallmatrix}\ell\\\ell\end{smallmatrix}\right]\right)=0$
		\item \makebox[7.2cm][l]{quantifies the zero vector to zero:} $r_{f,p}\left(\left[\begin{smallmatrix}0\\0\end{smallmatrix}\right]\right)=0$
		\item \makebox[7.2cm][l]{scales linearly in $\vec{v}$ where $\ell\in\mathbb{R}$:} $r_{f,p}(\ell\vec{v})=\ell r_{f,p}(\vec{v})$
		\item {is convex in $\vec{v}$:}\begin{itemize}
			\item \makebox[6.5cm][l]{$f$-inequality $\ell\in\{0,1\}$:} $r_{f,p}(\ell\vec{v}_1+(1-\ell)\vec{v}_2) = \ell r_{f,p}(\vec{v}_1)+(1-\ell)r_{f,p}(\vec{v}_2)$
			\item\makebox[6.5cm][l]{weak $f$-inequality $\ell\in(0,1)$:} $r_{f,p}(\ell\vec{v}_1+(1-\ell)\vec{v}_2) \leq \ell r_{f,p}(\vec{v}_1)+(1-\ell)r_{f,p}(\vec{v}_2)$
			\item \makebox[6.5cm][l]{strict $f$-inequality $\ell\in(0,1)$:} $r_{f,p}(\ell\vec{v}_1+(1-\ell)\vec{v}_2) < \ell r_{f,p}(\vec{v}_1)+(1-\ell)r_{f,p}(\vec{v}_2)$
		\end{itemize} 
		\item {satisfies a triangle inequality in $\vec{v}$:}\begin{itemize}
			\item \makebox[6.5cm][l]{$f$-inequality $\textnormal{Slope}(\vec{v}_1) = \textnormal{Slope}(\vec{v}_2)$:} $r_{f,p}(\vec{v}_1+\vec{v}_2) = r_{f,p}(\vec{v}_1)+r_{f,p}(\vec{v}_2)$
			\item\makebox[6.5cm][l]{weak $f$-inequality $\textnormal{Slope}(\vec{v}_1) \neq \textnormal{Slope}(\vec{v}_2)$:} $r_{f,p}(\vec{v}_1+\vec{v}_2) \leq r_{f,p}(\vec{v}_1)+r_{f,p}(\vec{v}_2)$
			\item \makebox[6.5cm][l]{strict $f$-inequality $\textnormal{Slope}(\vec{v}_1) \neq \textnormal{Slope}(\vec{v}_2)$:} $r_{f,p}(\vec{v}_1+\vec{v}_2) < r_{f,p}(\vec{v}_1)+r_{f,p}(\vec{v}_2)$
		\end{itemize}  
	\end{enumerate}
	\item the function $I_{f,p}(\mathbf{S})$:
	\begin{enumerate}
		\item \makebox[6cm][l]{quantifies the bottom element to zero:} $I_{f,p}(\bot_\mathbf{S})=0$
		\item \makebox[6cm][l]{maintains the zonogon order:} 
		\begin{itemize}
			\item \makebox[5.2cm][l]{$f$-inequality:} $\langle\mathbf{S}_1\rangle = \langle\mathbf{S}_2\rangle \Longrightarrow I_{f,p}(\mathbf{S}_1) = I_{f,p}(\mathbf{S}_2)$
			\item\makebox[5.2cm][l]{weak $f$-inequality:} $\langle\mathbf{S}_1\rangle \sqsubseteq \langle\mathbf{S}_2\rangle \Longrightarrow I_{f,p}(\mathbf{S}_1) \leq I_{f,p}(\mathbf{S}_2)$
			\item \makebox[5.2cm][l]{strict $f$-inequality:} $\langle\mathbf{S}_1\rangle \sqsubset \langle\mathbf{S}_2\rangle \Longrightarrow I_{f,p}(\mathbf{S}_1) < I_{f,p}(\mathbf{S}_2)$
		\end{itemize} 
	\end{enumerate}
\end{enumerate}}
\begin{proof}\hfill
	\begin{enumerate}
		\item Properties of $r_{f,p}(\vec{v})$:
		\begin{enumerate}
			\item Non-zero vectors of slope one $\vec{v}=\left[\begin{smallmatrix}\ell\\\ell
			\end{smallmatrix}\right]$:
			\begin{subequations}
				\begin{align}
					r_{f,p}\left(\left[\begin{smallmatrix}\ell\\\ell
					\end{smallmatrix}\right]\right) &= \left(p\ell+(1-p)\ell\right) \cdot f\left(\tfrac{\ell}{p\ell+(1-p)\ell}\right)= \ell\cdot f\left(\tfrac{\ell}{\ell}\right) & \text{(by Definition \ref{def:generalized-ineq})}\label{eq:sub-p1a}\\
					&= \ell\cdot f\left(1\right) = 0& \text{(by Notation \ref{not:f-div})}
				\end{align}
			\end{subequations}
			\item The zero vector $\vec{v}=\left[\begin{smallmatrix}0\\0
			\end{smallmatrix}\right]$:
			\begin{subequations}
			\begin{align}
				r_{f,p}\left(\left[\begin{smallmatrix}0\\0
				\end{smallmatrix}\right]\right) &= 0\cdot f\left(\tfrac{0}{0}\right)& \text{(by \refEq{eq:sub-p1a})}\\
				&= 0& \text{(by Notation \ref{not:f-div})}
			\end{align}
			\end{subequations}
			\item Linear scaling of vectors:
			\begin{subequations}
				\begin{align}
					 r_{f,p}\left(\ell \vec{v}\right) = r_{f,p}\left(\ell\left[\begin{smallmatrix}x\\y
					\end{smallmatrix}\right]\right) &= \left(p\ell x+(1-p)\ell y\right) \cdot f\left(\tfrac{\ell x}{p\ell x+(1-p)\ell y}\right)& \text{(by Definition \ref{def:generalized-ineq})}\\
					&= \ell\left(p x+(1-p) y\right) \cdot f\left(\tfrac{x}{px+(1-p)y}\right) \\
					&= \ell r_{f,p}\left(\left[\begin{smallmatrix}x\\y
					\end{smallmatrix}\right]\right) = \ell r_{f,p}\left(\vec{v}\right)
				\end{align}
			\end{subequations}
			\item Convexity in $\vec{v}$:
			\begin{itemize}
				\item Assume $\ell=\{0,1\}$, then \refEq{eq:sub-p1d} simplifies into a simple identity:
				\begin{equation}
					r_{f,p}(\ell\vec{v}_1+(1-\ell)\vec{v}_2) = \ell r_{f,p}(\vec{v}_1)+(1-\ell)r_{f,p}(\vec{v}_2) \label{eq:sub-p1d}
				\end{equation}
				\item Assume $\ell\in(0,1)$:  We use the following definitions as abbreviation:
				\begin{equation*}
					\begin{matrix}
						\vec{v}_1 \definedAs \left[\begin{smallmatrix}
							x_1\\y_1
						\end{smallmatrix}\right] \quad&\quad a_1 \definedAs x_1 p + y_1 (1-p) \quad&\quad b_1 \definedAs \tfrac{\ell a_1}{\ell a_1 + (1-\ell)a_2}\\ 
						\vec{v}_2 \definedAs \left[\begin{smallmatrix}
						x_2\\y_2
						\end{smallmatrix}\right] \quad&\quad a_2 \definedAs x_2 p + y_2 (1-p) \quad&\quad b_2 \definedAs \tfrac{(1-\ell) a_2}{\ell a_1 + (1-\ell)a_2}\\ 
					\end{matrix}
				\end{equation*}
				The cases of $a_1=0$ and $a_2=0$ are covered by the convention $0f(\frac{0}{0})=0$ (Notation \ref{not:f-div}). Therefore, we can assume they are non-zero and utilize the following two relations: $0 < b_1 < 1$ and $b_2 = 1 - b_1$. If $f$ is strictly convex, let $(\sim) = (<)$, otherwise let $(\sim) = (\leq)$.
				\begin{subequations}
					\begin{align}
						r_{f,p}(\ell\vec{v}_1+(1-\ell)\vec{v}_2) &= (\ell a_1 + (1-\ell)a_2)\cdot f\left(\frac{\ell x_1 + (1-\ell)x_2}{\ell a_1 + (1-\ell)a_2}\right)& \text{(by Definition \ref{def:generalized-ineq})}\\
						&= (\ell a_1 + (1-\ell)a_2)\cdot f\left(b_1\frac{x_1}{a_1} + b_2\frac{x_2}{a_2}\right)\\
						&\sim (\ell a_1 + (1-\ell)a_2)\cdot\left( b_1 f\left(\frac{x_1}{a_1}\right) + b_2 f\left(\frac{x_2}{a_2}\right)\right)& \text{(by convexity of $f$)}\\
						&= \ell a_1 f\left(\frac{x_1}{a_1}\right) + (1-\ell)a_2 f\left(\frac{x_2}{a_2}\right)\\
						&= \ell r_{f,p}(\vec{v}_1)+(1-\ell)r_{f,p}(\vec{v}_2)
					\end{align}
				\end{subequations}
			\end{itemize}
			\item Triangle inequality in $\vec{v}$:
			\begin{itemize}
				\item Assume $\textnormal{Slope}(\vec{v}_1) = \textnormal{Slope}(\vec{v}_2)$: Then there exists an $\ell\in\mathbb{R}$ such that $\ell\vec{v}_1=\vec{v}_2$.
				\begin{subequations}
					\begin{align}
					r_{f,p}(\vec{v}_1+\vec{v}_2) = r_{f,p}(\vec{v}_1+\ell\vec{v}_1) &= r_{f,p}((1+\ell)\vec{v}_1)\\
					 &= (1+\ell) r_{f,p}(\vec{v}_1) &\text{(by Theorem \ref{thm:rf-props} nr. 1.c)} \\
					 &= r_{f,p}(\vec{v}_1)+\ell r_{f,p}(\vec{v}_1) \\
					 &= r_{f,p}(\vec{v}_1)+ r_{f,p}(\ell\vec{v}_1)&\text{(by Theorem \ref{thm:rf-props} nr. 1.c)} \\
					 &= r_{f,p}(\vec{v}_1)+r_{f,p}(\vec{v}_2)
					\end{align}
				\end{subequations}
				\item Assume $\textnormal{Slope}(\vec{v}_1) \neq \textnormal{Slope}(\vec{v}_2)$:\newline If $f$ is strictly convex, let $(\sim) = (<)$, otherwise let $(\sim) = (\leq)$.
				\begin{subequations}
					\begin{align}
						r_{f,p}(\ell\vec{v}_1+(1-\ell)\vec{v}_2) &~\sim~ \ell r_{f,p}(\vec{v}_1)+(1-\ell)r_{f,p}(\vec{v}_2)&\text{(by Theorem \ref{thm:rf-props} nr. 1.d)} \\
						r_{f,p}(0.5(\vec{v}_1+\vec{v}_2)) &~\sim~ 0.5 r_{f,p}(\vec{v}_1)+0.5 r_{f,p}(\vec{v}_2)&\text{(let $\ell=0.5$)} \\
						0.5 r_{f,p}(\vec{v}_1+\vec{v}_2) &~\sim~ 0.5 r_{f,p}(\vec{v}_1)+0.5 r_{f,p}(\vec{v}_2)&\text{(by Theorem \ref{thm:rf-props} nr. 1.c)} \\
						r_{f,p}(\vec{v}_1+\vec{v}_2) &~\sim~ r_{f,p}(\vec{v}_1)+ r_{f,p}(\vec{v}_2)						
					\end{align}
				\end{subequations}
			\end{itemize}
		\end{enumerate}
		\item Properties of $I_{f,p}(\mathbf{S})$:
		\begin{enumerate}
			\item Bottom element $\bot_\mathbf{S}=\{1\}$:
			\begin{equation}
				I_{f,p}(\bot_\mathbf{S}) = \sum_{\vec{v}\in\kappa(\{1\})}r_{f,p}(\vec{v})
				 = r_{f,p}(\left[\begin{smallmatrix}1\\1\end{smallmatrix}\right]) = 0 \qquad\qquad\text{(by Theorem \ref{thm:rf-props} nr. 1.a)}
			\end{equation}
			\item Zonogon order:
			\begin{itemize}
			\item Assume $Z_\kappa(\mathbf{S}_1) = Z_\kappa(\mathbf{S}_2)$, which equals $\langle\mathbf{S}_1\rangle = \langle\mathbf{S}_2\rangle$ by Definition \ref{def:pop-equivalence} and Notation \ref{not:pop-classes}: In this case, both zonogons have the same boundary. Since the boundary consists of the generating vectors sorted by slope, the generating vectors of identical slope have the same sum:
			\begin{subequations}
				\begin{align}
					\forall x\in\mathbb{R}:&&\sum_{\stackrel{\vec{v}\in\kappa(\mathbf{S}_1)}{\textnormal{Slope}(\vec{v})=x}} \vec{v} &=\sum_{\stackrel{\vec{v}\in\kappa(\mathbf{S}_2)}{\textnormal{Slope}(\vec{v})=x}} \vec{v} &\text{(by $Z_\kappa(\mathbf{S}_1) = Z_\kappa(\mathbf{S}_2)$)}\\
					\forall x\in\mathbb{R}:&&r_{f,p}\left(\sum_{\stackrel{\vec{v}\in\kappa(\mathbf{S}_1)}{\textnormal{Slope}(\vec{v})=x}} \vec{v}\right) &=r_{f,p}\left(\sum_{\stackrel{\vec{v}\in\kappa(\mathbf{S}_2)}{\textnormal{Slope}(\vec{v})=x}} \vec{v}\right)\\
					\forall x\in\mathbb{R}:&&\sum_{\stackrel{\vec{v}\in\kappa(\mathbf{S}_1)}{\textnormal{Slope}(\vec{v})=x}} r_{f,p}\left(\vec{v}\right) &=\sum_{\stackrel{\vec{v}\in\kappa(\mathbf{S}_2)}{\textnormal{Slope}(\vec{v})=x}} r_{f,p}\left(\vec{v}\right) &\text{(by Theorem \ref{thm:rf-props} nr. 1.e)}\\
					&&\sum_{x\in\mathbb{R}} \sum_{\stackrel{\vec{v}\in\kappa(\mathbf{S}_1)}{\textnormal{Slope}(\vec{v})=x}} r_{f,p}\left(\vec{v}\right)&=\sum_{x\in\mathbb{R}}\sum_{\stackrel{\vec{v}\in\kappa(\mathbf{S}_2)}{\textnormal{Slope}(\vec{v})=x}} r_{f,p}\left(\vec{v}\right)\\
					&& \sum_{{\vec{v}\in\kappa(\mathbf{S}_1)}} r_{f,p}\left(\vec{v}\right)&=\sum_{{\vec{v}\in\kappa(\mathbf{S}_2)}} r_{f,p}\left(\vec{v}\right)\\
					&&I_{f,p}(\mathbf{S}_1) &= I_{f,p}(\mathbf{S}_2)
				\end{align}
			\end{subequations}
			\item Assume $Z_\kappa(\mathbf{S}_1) \subset Z_\kappa(\mathbf{S}_2)$, which equals $\langle\mathbf{S}_1\rangle \sqsubset \langle\mathbf{S}_2\rangle$ by Definition \ref{def:pop-equivalence} and Definition \ref{def:pop-class-lattice}: In this case, there exists a stochastic matrix $\lambda$ which combines some vectors from $\kappa_2$ with different slope:
			\begin{equation}
				\kappa_1 \definedAs \kappa(\mathbf{S}_1) \qquad \qquad \kappa_2 \definedAs \kappa(\mathbf{S}_2) \qquad \qquad \kappa_1 = \kappa_2 \lambda
			\end{equation}
			Let $\kappa_1$ be a $2\times a$ stochastic matrix, $\kappa_2$ be a $2\times b$ stochastic matrix and $\lambda$ be a $b\times a$ stochastic matrix. We write $\kappa_2[:,i]$ to refer to the $i^\text{th}$ column of matrix $\kappa_2$ and write  $\lambda[i,j]$ for the element at row $i\in\{1,..,b\}$ and column $j\in\{1,..,a\}$. Since $\lambda$ is a stochastic matrix, its rows sum to one $\forall i\in\{1,..,b\}~:~\sum_{j=1}^a\lambda[i,j] = 1$. 
			If $f$ is strictly convex, let $(\sim) = (<)$, otherwise let $(\sim) = (\leq)$.
			\begin{subequations}
				\begin{align}
					I_{f,p}(\mathbf{S}_1) &= \sum_{j=1}^a r_{f,p}(\kappa_1[:,j]) &\text{(by Definition \ref{def:generalized-ineq})}\\
					&= \sum_{j=1}^a r_{f,p}(\sum_{i=1}^{b}\kappa_2[:,i]\lambda[i,j]) &\text{(by $\kappa_1 = \kappa_2 \lambda$)}\\
					&\sim \sum_{j=1}^a \sum_{i=1}^b r_{f,p}(\kappa_2[:,i]\lambda[i,j]) &\text{(by Theorem \ref{thm:rf-props} nr. 1.e)}\\
					&= \sum_{j=1}^a \sum_{i=1}^b \lambda[i,j] r_{f,p}(\kappa_2[:,i])\qquad &\text{(by Theorem \ref{thm:rf-props} nr. 1.c)}\\
					&= \sum_{i=1}^b r_{f,p}(\kappa_2[:,i])\qquad & \text{(by $\sum_{j=1}^a \lambda[i,j] = 1$)}\\
					&= I_{f,p}(\mathbf{S}_2)&\text{(by Definition \ref{def:generalized-ineq})}
				\end{align}
			\end{subequations}
			\end{itemize}
		\end{enumerate}
	\end{enumerate}
\end{proof}

\subsection{Additivity of f-inequality}
\label{subap:f-ineq-additivity}
\textbf{Proof of Lemma \ref{lem:minkowski-sum-subsets} from Section \ref{subsec:generalized-ineq-measure}:}\newline
\textit{Consider two non-empty sets of populations with equal cardinality ($|\mathbf{A}| = |\mathbf{B}|$), then:
	\begin{subequations}
		\begin{align}
			\text{$f$-inequality:}&&\sum_{\mathbf{S}\in \mathbf{A}}Z_\kappa(\mathbf{S}) = \sum_{\mathbf{S}\in \mathbf{B}}Z_\kappa(\mathbf{S}) &\Longrightarrow \sum_{\mathbf{S}\in \mathbf{A}}I_{f,p}(\mathbf{S}) = \sum_{\mathbf{S}\in \mathbf{B}}I_{f,p}(\mathbf{S})\\
			\text{weak $f$-inequality:}&&\sum_{\mathbf{S}\in \mathbf{A}}Z_\kappa(\mathbf{S}) \subseteq \sum_{\mathbf{S}\in \mathbf{B}}Z_\kappa(\mathbf{S}) &\Longrightarrow \sum_{\mathbf{S}\in \mathbf{A}}I_{f,p}(\mathbf{S}) \leq \sum_{\mathbf{S}\in \mathbf{B}}I_{f,p}(\mathbf{S})\\
			\text{strict $f$-inequality:}&&\sum_{\mathbf{S}\in \mathbf{A}}Z_\kappa(\mathbf{S}) \subset \sum_{\mathbf{S}\in \mathbf{B}}Z_\kappa(\mathbf{S}) &\Longrightarrow \sum_{\mathbf{S}\in \mathbf{A}}I_{f,p}(\mathbf{S}) < \sum_{\mathbf{S}\in \mathbf{B}}I_{f,p}(\mathbf{S})
		\end{align}
	\end{subequations}}
\begin{proof}
	Let $m = |\mathbf{A}| = |\mathbf{B}|$ and $(\sim,\approx)\in\{(=,=),(\subseteq,\leq),(\subset,<)\}$. We use the notation $\mathbf{A}[i]$ and $\mathbf{B}[i]$ with $1\leq i \leq m$ to indicate a specific population within the set $\mathbf{A}$ and $\mathbf{B}$ respectively.
	\begin{equation*}
		\begingroup
		\allowdisplaybreaks
		\begin{aligned}
			\sum_{i=1}^m Z_\kappa(\mathbf{A}[i]) &\sim \sum_{i=1}^m Z_\kappa(\mathbf{B}[i])\\
			Z\left(\begin{bmatrix} \kappa(\mathbf{A}[1]) &\dots&\kappa(\mathbf{A}[m]) \end{bmatrix}\right) &\sim Z\left(\begin{bmatrix}
				\kappa(\mathbf{B}[1]) &\dots&\kappa(\mathbf{B}[m])
			\end{bmatrix}\right) & & \text{(by Definition \ref{def:zono-sum})}\\
			Z\left(\frac{1}{m}\begin{bmatrix} \kappa(\mathbf{A}[1]) &\dots&\kappa(\mathbf{A}[m]) \end{bmatrix}\right) &\sim Z\left(\frac{1}{m}\begin{bmatrix}
				\kappa(\mathbf{B}[1]) &\dots&\kappa(\mathbf{B}[m])
			\end{bmatrix}\right) & & \text{(scale zonogon to (1,1))}\\
			\sum_{\vec{v}\in\frac{1}{m}\left[\begin{smallmatrix} \kappa(\mathbf{A}[1]) &\dots&\kappa(\mathbf{A}[m]) \end{smallmatrix}\right]} r_{f,p}(\vec{v}) &\approx 
			\sum_{\vec{v}\in\frac{1}{m}\left[\begin{smallmatrix} \kappa(\mathbf{B}[1]) &\dots&\kappa(\mathbf{B}[m]) \end{smallmatrix}\right]} r_{f,p}(\vec{v}) & & \text{(by Def. \ref{def:pop-class-lattice}, Thm. \ref{thm:rf-props} nr. 2.b, Def. \ref{def:generalized-ineq})}\\
			\sum_{\vec{v}\in\left[\begin{smallmatrix} \kappa(\mathbf{A}[1]) &\dots&\kappa(\mathbf{A}[m]) \end{smallmatrix}\right]} \frac{1}{m}r_{f,p}(\vec{v}) &\approx 
			\sum_{\vec{v}\in\left[\begin{smallmatrix} \kappa(\mathbf{B}[1]) &\dots&\kappa(\mathbf{B}[m]) \end{smallmatrix}\right]} \frac{1}{m}r_{f,p}(\vec{v}) & & \text{(by Theorem \ref{thm:rf-props} nr. 1.c)}\\
			\sum_{i=1}^{m}~\ \sum_{\vec{v}\in\kappa(\mathbf{A}[i])} r_{f,p}(\vec{v}) &\approx 
			\sum_{i=1}^{m}~\ \sum_{\vec{v}\in\kappa(\mathbf{B}[i])}r_{f,p}(\vec{v}) & & \text{(multiply $m$, split sum)}\\
			\sum_{i=1}^{m}~\ I_{f,p}(\mathbf{A}[i]) &\approx 
			\sum_{i=1}^{m}~\ I_{f,p}(\mathbf{B}[i]) & & \text{(by Definition \ref{def:generalized-ineq})}\\
			\sum_{\mathbf{S}\in\mathbf{A}} I_{f,p}(\mathbf{S}) &\approx 
			\sum_{\mathbf{S}\in\mathbf{B}} I_{f,p}(\mathbf{S}) & &\text{(change notation)} \\
		\end{aligned}
		\endgroup
	\end{equation*}
\end{proof}

\subsection{Special cases of f-inequality}
\label{subap:f-ineq-cases}
\textbf{Proof of Theorem \ref{thm:special-cases} from Section \ref{subsec:generalized-ineq-measure}:}\newline
\textit{The Pietra index and Generalized Entropy index are special cases of $f$-inequality:
	\begin{subequations}
		\begin{align}
			R(\mathbf{S}) &= I_{f,p}(\mathbf{S}) &\text{where:~} & p=0~\text{ and }~ f(t) = \frac{|t-1|}{2}\\
			\text{GE}_{c}(\mathbf{S}) &= I_{f,p}(\mathbf{S}) &\text{where:~} & p=0~\text{ and }~f(t) = \frac{t^{1-c} - t}{c(c - 1)}\\
			\text{GE}_1(\mathbf{S}) &= I_{f,p}(\mathbf{S}) &\text{where:~} & p=0~\text{ and }~f(t) = -\ln\left(t\right)\\
			\text{GE}_0(\mathbf{S}) &= I_{f,p}(\mathbf{S}) &\text{where:~} & p=0~\text{ and }~f(t) =  t\ln\left(t\right)
		\end{align}
	\end{subequations}}
\begin{proof}
	We can simplify the generalized inequality function for $p=0$ as shown in \refEq{eq:if0}:
	\begin{equation}
		I_{f,0}\left(\mathbf{S}\right) = \frac{1}{|\mathbf{S}|}\sum_{s\in\mathbf{S}}\frac{s}{\overline{\mathbf{S}}}\cdot f\left(\frac{\overline{\mathbf{S}}}{s}\right)
		\label{eq:if0}
	\end{equation}
	\begin{itemize}
		\item Pietra index:
		\begin{equation}
			\begin{aligned}
				I_{f,0}(\mathbf{S}) &= \frac{1}{|\mathbf{S}|}\sum_{s\in\mathbf{S}}\frac{s}{\overline{\mathbf{S}}}\cdot f\left(\frac{\overline{\mathbf{S}}}{s}\right)\\
				&= \frac{1}{|\mathbf{S}|}\sum_{s\in\mathbf{S}}\frac{s}{\overline{\mathbf{S}}}\cdot \frac{|\frac{\overline{\mathbf{S}}}{s}-1|}{2}&\text{using:}~&  f(t) = \frac{|t-1|}{2}\\
				&= \frac{1}{2|\mathbf{S}|}\sum_{s\in\mathbf{S}} \frac{|\overline{\mathbf{S}}-s|}{\overline{\mathbf{S}}}&\text{using:}~& s\geq0\\
				&= \frac{1}{2|\mathbf{S}|}\sum_{s\in\mathbf{S}} \frac{|s - \overline{\mathbf{S}}|}{\overline{\mathbf{S}}}&\text{using:}~&  |a-b|=|b-a|\\
				&= R(\mathbf{S})&\text{(by }& \text{\refEq{eq:pietra})}
			\end{aligned}
		\end{equation}
		The function $f(t) = \frac{|t-1|}{2}$ is a well known generator function for an $f$-divergences from the total variation distance.
		\item $\text{GE}_{c}(\mathbf{S})$ index with $c\notin\{0,1\}$:
		\begin{equation}
			\begin{aligned}
				I_{f,0}(\mathbf{S}) &= \frac{1}{|\mathbf{S}|}\sum_{s\in\mathbf{S}}\frac{s}{\overline{\mathbf{S}}}\cdot f\left(\frac{\overline{\mathbf{S}}}{s}\right)\\
				&= \frac{1}{|\mathbf{S}|}\sum_{s\in\mathbf{S}}\frac{s}{\overline{\mathbf{S}}}\cdot \frac{\left(\frac{\overline{\mathbf{S}}}{s}\right)^{1-c} - \frac{\overline{\mathbf{S}}}{s}}{c(c - 1)}&\text{using: }&  \frac{t^{1-c} - t}{c(c - 1)} \\
				&= \frac{1}{|\mathbf{S}|}\sum_{s\in\mathbf{S}}\frac{s}{\overline{\mathbf{S}}}\cdot \frac{\frac{\overline{\mathbf{S}}}{s} \cdot\left(\frac{s}{\overline{\mathbf{S}}}\right)^{c} - \frac{\overline{\mathbf{S}}}{s}}{c(c - 1)}\\
				&= \frac{1}{c(c - 1)}\frac{1}{|\mathbf{S}|}\sum_{s\in\mathbf{S}}\left({\left(\frac{s}{\overline{\mathbf{S}}}\right)^{c} - 1}\right)\\
				&=\text{GE}_{c}(\mathbf{S})&\text{(by }& \text{\refEq{eq:gec-1})}
			\end{aligned}
		\end{equation}
		The function $f(t)$ satisfies the requirements for a generator function of an $f$-divergence:
		\begin{enumerate}
			\item $f(1) =\frac{1}{c(c-1)}\left(1-1\right) = 0$.
			\item $f(t)$ is convex for $t>0$ and $c\in\mathbb{R}\setminus\{0,1\}$ since $f''(t) = \left(\frac{1}{t}\right)^{c+1} \geq 0$
			\item $f(t)$ is finite for $t>0$ and $c\in\mathbb{R}\setminus\{0,1\}$.
		\end{enumerate}
		\item $\text{GE}_{1}(\mathbf{S})$ index with $c = 1$ (Theil index):
		\begin{equation}
			\begin{aligned}
				I_{f,0}(\mathbf{S})  &= \frac{1}{|\mathbf{S}|}\sum_{s\in\mathbf{S}}\frac{s}{\overline{\mathbf{S}}}\cdot f\left(\frac{\overline{\mathbf{S}}}{s}\right)\\
				&= \frac{1}{|\mathbf{S}|}\sum_{s\in\mathbf{S}}\frac{s}{\overline{\mathbf{S}}}\cdot \left(-\ln\left(\frac{\overline{\mathbf{S}}}{s}\right)\right)
				&\text{using: }&  f(t) = -\ln\left(t\right)\\
				&= \frac{1}{|\mathbf{S}|}\sum_{s\in\mathbf{S}}\frac{s}{\overline{\mathbf{S}}}\cdot \ln\left(\frac{s}{\overline{\mathbf{S}}}\right)\\
				&=\text{GE}_{1}(\mathbf{S})&\text{(by }& \text{\refEq{eq:gec-1})}
			\end{aligned}
		\end{equation}
		The function $f(t) = -\ln\left(t\right)$ is a well known generator function for an $f$-divergences from the reverse Kullback–Leibler divergence.
		\item $\text{GE}_{0}(\mathbf{S})$ index with $c = 0$:
		\begin{equation}
			\begin{aligned}
				I_{f,0}(\mathbf{S}) &= \frac{1}{|\mathbf{S}|}\sum_{s\in\mathbf{S}}\frac{s}{\overline{\mathbf{S}}}\cdot f\left(\frac{\overline{\mathbf{S}}}{s}\right)\\
				&= \frac{1}{|\mathbf{S}|}\sum_{s\in\mathbf{S}}\frac{s}{\overline{\mathbf{S}}}\cdot \frac{\overline{\mathbf{S}}}{s}\ln\left(\frac{\overline{\mathbf{S}}}{s}\right)
				&\text{using: }&  f(t) = t\ln\left(t\right)\\
				&= -\frac{1}{|\mathbf{S}|}\sum_{s\in\mathbf{S}}\ln\left(\frac{s}{\overline{\mathbf{S}}}\right)\\
				&=\text{GE}_{0}(\mathbf{S})&\text{(by }& \text{\refEq{eq:gec-1})}
			\end{aligned}
		\end{equation}
		The function $f(t) = t\ln\left(t\right)$ is a well known generator function for an $f$-divergences from the Kullback–Leibler divergence.
	\end{itemize}
\end{proof}
	
	\section[\appendixname~\thesection]{Decomposition properties}
	\label{ap:decomposition-properties}
%
\textbf{Proof of Theorem \ref{thm:decom-proofs} from Section \ref{subsec:decomposition-ineq-measure}:}
\textit{Definition \ref{def:def-f-union} satisfies Property \ref{prop:union-1}-\ref{prop:union-4}.}
\begin{proof}\hfill
	\begin{itemize}
		\item \textbf{Property \ref{prop:union-1}} (commutativity): The join operator (convex hull of zonogons) is invariant to the order of zonogons. Therefore, the measure $I_{f,p}^\cup(\cdot)$ of Definition \ref{def:def-f-union} is invariant to the order of attribute sets in an atom: 
		\begin{equation*}
			\begin{aligned}
			\forall\alpha\in\mathcal{A}(n):&&\bigsqcup_{\mathbf{a}\in\alpha} \langle\Gamma(\mathbf{a},\mathbf{M})\rangle &= \bigsqcup_{\mathbf{a}\in\{\sigma(\mathbf{a})~:~\mathbf{a}\in\alpha\}} \langle\Gamma(\mathbf{a},\mathbf{M})\rangle\quad& \text{(commutativity of lattice join)}\\
			&&I^\cup_{f,p}(\alpha,\mathbf{M}) &= I^\cup_{f,p}(\{\sigma(\mathbf{a})~:~\mathbf{a}\in\alpha\},\mathbf{M}) & \text{(by Theorem \ref{thm:rf-props} nr. 2.b)}
			\end{aligned}
		\end{equation*}
		\item \textbf{Property \ref{prop:union-2}} (monotonicity): The join element is monotonically increasing and thus also the measure $I_{f,p}^\cup(\cdot)$ of Definition \ref{def:def-f-union}.
		\begin{equation*}
			\begin{aligned}
				\forall\alpha\in\mathcal{A}(n),~\forall \mathbf{a}\in\mathcal{P}(\{1,..,n\}):\\
				\bigsqcup_{\mathbf{b}\in\alpha} \langle\Gamma(\mathbf{b},\mathbf{M})\rangle &\sqsubseteq  \left(\bigsqcup_{\mathbf{b}\in\alpha} \langle\Gamma(\mathbf{b},\mathbf{M})\rangle\right)\sqcup\langle\Gamma(\mathbf{a},\mathbf{M})\rangle& \text{(monotonicity of lattice join)}\\
				I^\cup_{f,p}(\alpha,\mathbf{M}) &\leq I^\cup_{f,p}(\alpha\cup\{\mathbf{a}\},\mathbf{M})& \text{(by Theorem \ref{thm:rf-props} nr. 2.b)}
			\end{aligned}
		\end{equation*}
		\item \textbf{Property \ref{prop:union-3}} (self-inequality): The join of a single population is an identity such that the measure $I_{f,p}^\cup(\cdot)$ of Definition \ref{def:def-f-union} equals $I_{f,p}(\cdot)$.
		\begin{equation*}
				\forall \mathbf{a}\in\mathcal{P}(\{1,..,n\}):~ I^\cup_{f,p}(\{\mathbf{a}\},\mathbf{M}) = I_{f,p}\left(\bigsqcup_{\mathbf{b}\in\{\mathbf{a}\}}\langle\Gamma(\mathbf{b},\mathbf{M})\rangle\right)= I_{f,p}\left(\Gamma(\mathbf{a},\mathbf{M})\right)\quad \text{(by Definition \ref{def:def-f-union})}
		\end{equation*}
		\item \textbf{Property \ref{prop:union-4}} (non-negativity):
		We begin with the required preliminaries: 
		\begin{itemize}
			\item The join on the union lattice can be expressed as shown in \refEq{eq:joinuion}~\cite{mages2024}, using the function `reduce` of \refEq{eq:imp-partial}.
			\item The function $\textnormal{reduce}(\subset,\cdot)$ does not affect the convex hull of the underlying zonogons, which provides \refEq{eq:joinuion-2} from \refEq{eq:joinuion}.
			\item From \refEq{eq:joinuion-2}, we obtain the monotonicity of the cumulative measure on the lattice shown in \refEq {eq:joinuion-3}.
			\item We notate the set of immediate successors of $\alpha$ as $\alpha^+$ (\refEq{eq:cover}).
			\item The Möbius inverse can be computed using an inclusion-exclusion relation as shown in \refEq{eq:möbineq}~\cite{williams-beer,dualdecomposition,mages2024}.
			\item Since the inclusion-exclusion of a constant is the constant itself, we obtain \refEq{eq:möbineq-new} from \refEq{eq:möbineq}.
			\item Using \refEq{eq:möbineq-new} and Definition \ref{def:union-möbius} (\refEq{subeq:möbius-inv}) provides \refEq{eq:möbineq-2}.
		\end{itemize}
		\begin{subequations}
			\begin{align}
				&&\alpha \curlyvee \beta &= \textnormal{reduce}(\subset,~\alpha\cup\beta) \label{eq:joinuion}\\
				&&\bigsqcup_{\mathbf{c}\in(\alpha \curlyvee \beta)} \langle\Gamma(\mathbf{c},\mathbf{M})\rangle&= \bigsqcup_{\mathbf{c}\in(\alpha \cup \beta)} \langle\Gamma(\mathbf{c},\mathbf{M})\rangle \label{eq:joinuion-2}\\
				&&\alpha \preceq \beta &\Longrightarrow I^\cup_{f,p}(\alpha,\mathbf{M})\leq  I^\cup_{f,p}(\beta,\mathbf{M})\label{eq:joinuion-3}\\
				&&\alpha^+ &\definedAs \{\beta\in\mathcal{A}(n)~:~\alpha\prec\beta~\text{ and }~\neg(\exists\gamma\in\mathcal{A}(n))[\alpha \prec \gamma ~\text{ and }~ \gamma\prec\beta]\}\label{eq:cover}\\
				\alpha \neq \top_\cup:&& \sum_{\beta\in\dot{\uparrow}\alpha} I_{f,p}^\delta(\beta,\mathbf{M}) &=
				\sum_{\emptyset\neq\mathbf{B}\subseteq\alpha^+} (-1)^{|\mathbf{B}|-1} \left(I_{f,p}^\cup(\top_\cup,\mathbf{M}) - I_{f,p}^\cup(\bigcurlyvee_{\beta\in\mathbf{B}}\beta,\mathbf{M})\right) \label{eq:möbineq}\\
				\alpha \neq \top_\cup:&& \sum_{\beta\in\dot{\uparrow}\alpha} I_{f,p}^\delta(\beta,\mathbf{M}) &=I_{f,p}^\cup(\top_\cup,\mathbf{M}) - 
				\sum_{\emptyset\neq\mathbf{B}\subseteq\alpha^+} (-1)^{|\mathbf{B}|-1} I_{f,p}^\cup(\bigcurlyvee_{\beta\in\mathbf{B}}\beta,\mathbf{M}) \label{eq:möbineq-new}\\
				\alpha \neq \top_\cup:&& I_{f,p}^\delta(\alpha,\mathbf{M}) &=- I_{f,p}^\cup(\alpha,\mathbf{M}) + \sum_{\emptyset\neq\mathbf{B}\subseteq\alpha^+} (-1)^{|\mathbf{B}|-1} I_{f,p}^\cup(\bigcurlyvee_{\beta\in\mathbf{B}}\beta,\mathbf{M})\label{eq:möbineq-2}
			\end{align}
			\label{eq:non-neg-prelim}
		\end{subequations}
		To show the non-negativity of partial contributions, we split the case of $\alpha=\top_\cup$:
		\begin{enumerate}
			\item Assume $\alpha=\top_\cup=\{\{1,..,n\}\}$: In this case, its strict upset is the empty set ($\dot{\uparrow}\alpha = \{\}$). We obtain from Definition \ref{def:union-möbius} that its partial contribution is zero and thus non-negative:
			\begin{equation}
				\alpha=\top_\cup:~\qquad I^\delta_{f,p}(\alpha,\mathbf{M}) = I^\cup_{f,p}(\top_\cup,\mathbf{M}) - I^\cup_{f,p}(\alpha,\mathbf{M}) = 0 \geq 0
				\label{eq:non-negativity-case-1}
			\end{equation}
			\item Assume $\alpha\neq\top_\cup$: In this case, the strict upset of $\alpha$ is non-empty. To show the non-negativity, we construct a lower and upper bound for the two components of \refEq{eq:möbineq-2}. First, we obtain the lower bound on the inclusion-exclusion relation from Corollary \ref{col:f-ineq-inc-enc} as shown in \refEq{eq:non-neg-bound-1} and simplify it using \refEq{eq:joinuion-2}.
			\begin{adjustwidth}{-\extralength}{0cm}
			\begin{subequations}
				\begin{align}
					\alpha\neq\top_\cup:&&I_{f,p}\left(\bigsqcap_{\gamma\in\alpha^+}\bigsqcup_{\mathbf{c}\in\gamma}\langle\Gamma(\mathbf{c},\mathbf{M})\rangle\right)&\leq \sum_{\emptyset\neq\mathbf{C}\subseteq\alpha^+} (-1)^{|\mathbf{C}|-1} I_{f,p}\left(\bigsqcup_{\gamma\in\mathbf{C}}\bigsqcup_{\mathbf{c}\in\gamma}\langle\Gamma(\mathbf{c},\mathbf{M})\rangle\right)\label{eq:non-neg-bound-1a}\\
					\alpha\neq\top_\cup:&&I_{f,p}\left(\bigsqcap_{\gamma\in\alpha^+}\bigsqcup_{\mathbf{c}\in\gamma}\langle\Gamma(\mathbf{c},\mathbf{M})\rangle\right)&\leq \sum_{\emptyset\neq\mathbf{C}\subseteq\alpha^+} (-1)^{|\mathbf{C}|-1} I_{f,p}\left(\bigsqcup_{\mathbf{c}\in\left(\bigcup_{\gamma\in\mathbf{C}}(\gamma)\right)}\langle\Gamma(\mathbf{c},\mathbf{M})\rangle\right)\\
					\alpha\neq\top_\cup:&&I_{f,p}\left(\bigsqcap_{\gamma\in\alpha^+}\bigsqcup_{\mathbf{c}\in\gamma}\langle\Gamma(\mathbf{c},\mathbf{M})\rangle\right)&\leq \sum_{\emptyset\neq\mathbf{C}\subseteq\alpha^+} (-1)^{|\mathbf{C}|-1} I_{f,p}^\cup(\bigcurlyvee_{\gamma\in\mathbf{C}}\gamma,\mathbf{M})\label{eq:non-neg-bound-1b}
				\end{align}
				\label{eq:non-neg-bound-1}
			\end{subequations}
			\end{adjustwidth}
			Second, we obtain an upper bound for atom $\alpha$ based on its immediate successors as shown in \refEq{eq:non-neg-bound-2} from \refEq{eq:joinuion-3}.
			\begin{subequations}
				\begin{align}
					\alpha\neq\top_\cup,\ \forall \gamma\in\alpha^+:&& \bigsqcup_{\mathbf{a}\in\alpha} \langle\Gamma(\mathbf{a},\mathbf{M})\rangle &\sqsubseteq \bigsqcup_{\mathbf{c}\in\gamma} \langle\Gamma(\mathbf{c},\mathbf{M})\rangle\\
					\alpha\neq\top_\cup:&&\bigsqcup_{\mathbf{a}\in\alpha} \langle\Gamma(\mathbf{a},\mathbf{M})\rangle &\sqsubseteq \bigsqcap_{\gamma\in\alpha^+}\bigsqcup_{\mathbf{c}\in\gamma} \langle\Gamma(\mathbf{c},\mathbf{M})\rangle\\
					\alpha\neq\top_\cup:&&I_{f,p}\left(\bigsqcup_{\mathbf{a}\in\alpha} \langle\Gamma(\mathbf{a},\mathbf{M})\rangle\right) &\leq I_{f,p}\left(\bigsqcap_{\gamma\in\alpha^+}\bigsqcup_{\mathbf{c}\in\gamma} \langle\Gamma(\mathbf{c},\mathbf{M})\rangle\right)\label{eq:non-neg-bound-2d}
				\end{align}
				\label{eq:non-neg-bound-2}
			\end{subequations}
		\end{enumerate}
		By transitivity, we obtain \refEq{eq:non-neg-bound-3a} from \refEq{eq:non-neg-bound-1b} and \refEq{eq:non-neg-bound-2d}. Re-arranging both terms demonstrates the desired non-negativity of the partial contributions.
		\begin{subequations}
			\begin{align}
				\alpha\neq\top_\cup:&&I^\cup_{f,p}(\alpha,\mathbf{M}) &\leq \sum_{\emptyset\neq\mathbf{C}\subseteq\alpha^+} (-1)^{|\mathbf{C}|-1} I_{f,p}^\cup(\bigcurlyvee_{\gamma\in\mathbf{C}}\gamma,\mathbf{M}) \label{eq:non-neg-bound-3a}\\
				\alpha\neq\top_\cup:&&0 &\leq - I^\cup_{f,p}(\alpha,\mathbf{M}) + \sum_{\emptyset\neq\mathbf{C}\subseteq\alpha^+} (-1)^{|\mathbf{C}|-1} I_{f,p}^\cup(\bigcurlyvee_{\gamma\in\mathbf{C}}\gamma,\mathbf{M})\\
				\alpha\neq\top_\cup:&&0 &\leq I^\delta_{f,p}(\alpha,\mathbf{M})\qquad\qquad\qquad\qquad\qquad\text{(applying  \refEq{eq:möbineq-2})}\label{eq:non-neg-bound-3c}
			\end{align}
			\label{eq:non-neg-bound-3}
		\end{subequations}
	\end{itemize}
	From \refEq{eq:non-negativity-case-1} and \refEq{eq:non-neg-bound-3c} we obtain the non-negativity of the decomposition and thus Property \ref{prop:union-4}.
	\begin{equation}
		\forall\alpha\in\mathcal{A}(n):\quad 0 \leq I^\delta_{f,p}(\alpha,\mathbf{M})
	\end{equation}
\end{proof}
	
	\section[\appendixname~\thesection]{Implementation suggestion}
	\label{ap:implementation}
%
This Section demonstrates the correctness of the suggested implementation in Section \ref{subsec:decomposition-ineq-measure}:
\begin{enumerate}
	\item Assume $\alpha=\top_\cup$: The partial contribution of the top element is always zero as shown in \refEq{eq:non-negativity-case-1}.
	\item Assume $\alpha\neq\top_\cup$: 
	Using~\cite[Lemma 5]{mages2024}, the set of immediate successors on the union lattice can be computed as shown in \refEq{eq:cover-2} with the function `dual` of \refEq{eq:imp-partial}.
	\begin{equation}
		\alpha \neq \top_\cup:\qquad \alpha^+ = \{\alpha \curlyvee \{\mathbf{b}\}~:~\mathbf{b}\in\textnormal{dual}(\alpha)\}\label{eq:cover-2}
	\end{equation}
	Using the properties of \refEq{eq:non-neg-prelim} and \refEq{eq:cover-2}, we can show the correctness of the suggested implementation by continuing the simplification of \refEq{eq:möbineq-2}:
	\begin{subequations}
		\begin{align}
			\alpha \neq \top_\cup:&& I_{f,p}^\delta(\alpha,\mathbf{M}) &=- I_{f,p}^\cup(\alpha,\mathbf{M}) + \sum_{\emptyset\neq\mathbf{B}\subseteq\alpha^+} (-1)^{|\mathbf{B}|-1} I_{f,p}^\cup(\bigcurlyvee_{\beta\in\mathbf{B}}\beta,\mathbf{M})\\
			\alpha \neq \top_\cup:&& I_{f,p}^\delta(\alpha,\mathbf{M}) &=- I_{f,p}^\cup(\alpha,\mathbf{M}) + \sum_{\emptyset\neq\mathbf{B}\subseteq\{\alpha \curlyvee \{\mathbf{b}\}~:~\mathbf{b}\in\textnormal{dual}(\alpha)\}} (-1)^{|\mathbf{B}|-1} I_{f,p}^\cup(\bigcurlyvee_{\beta\in\mathbf{B}}\beta,\mathbf{M})\label{eq:möbineq-3}\\
			\alpha \neq \top_\cup:&& I_{f,p}^\delta(\alpha,\mathbf{M}) &=- I_{f,p}^\cup(\alpha,\mathbf{M}) + \sum_{\emptyset\neq\mathbf{B}\subseteq\{\{\mathbf{b}\}~:~\mathbf{b}\in\textnormal{dual}(\alpha)\}} (-1)^{|\mathbf{B}|-1} I_{f,p}^\cup(\bigcurlyvee_{\beta\in\mathbf{B}}(\alpha \curlyvee \beta),\mathbf{M})\label{eq:möbineq-4}\\
			\alpha \neq \top_\cup:&& I_{f,p}^\delta(\alpha,\mathbf{M}) &=- I_{f,p}^\cup(\alpha,\mathbf{M}) + \sum_{\emptyset\neq\beta\subseteq\textnormal{dual}(\alpha)} (-1)^{|\beta|-1} I_{f,p}^\cup(\alpha \curlyvee \beta,\mathbf{M})\label{eq:möbineq-5}\\
			\alpha \neq \top_\cup:&& I_{f,p}^\delta(\alpha,\mathbf{M}) &=\sum_{\beta\subseteq\textnormal{dual}(\alpha)} (-1)^{|\beta|-1} I_{f,p}^\cup(\textnormal{reduce}(\subset,\alpha \cup \beta),\mathbf{M})\label{eq:möbineq-6}
		\end{align}
	\end{subequations}
\end{enumerate}
Therefore, the suggested implementation is correct for all $\alpha\in\mathcal{A}(n)$.
	
\end{document}